\PassOptionsToPackage{svgnames}{xcolor}
\documentclass[11pt,reqno]{article}
\usepackage{amsmath,hyperref,amssymb, amsthm}
\usepackage{authblk}
\usepackage[all]{xy}
\usepackage{graphicx,url}
\usepackage{color}
\usepackage{slashed}
\usepackage{tikz}
\usepackage{lipsum}
\usepackage{caption}
\usepackage{enumitem}
\usepackage[style=ieee]{biblatex}
\newcommand{\be}{\begin{eqnarray}}
\newcommand{\ee}{\end{eqnarray}}
\newcommand{\eeq}{\end{equation}}
\newcommand{\beq}{\begin{equation}}

\allowdisplaybreaks \numberwithin{equation}{section}

\DeclareSymbolFont{AMSa}{U}{msa}{m}{n}
\DeclareSymbolFont{AMSb}{U}{msb}{m}{n}
\DeclareMathSymbol{\fieldR}{\mathalpha}{AMSb}{"52}

\DeclareMathOperator{\Aut}{Aut}
\DeclareMathOperator{\rk}{rk}
\renewcommand{\Im}{\imag}
\DeclareMathOperator{\imag}{Im}

\DeclareMathOperator{\Hom}{Hom}
\newcommand{\CA}{{\cal A}}

\newcommand{\CH}{\mathcal{H}}
\newcommand{\Hh}{\mathcal{H}}
\newcommand{\CG}{\mathcal{G}}
\newcommand{\CI}{{\cal I}}
\newcommand{\CL}{{\cal L}}
\newcommand{\CN}{\mathcal{N}}

\newcommand{\calC}{\mathcal{C}}
\newcommand{\calR}{\mathcal{R}}

\newcommand{\CV}{\mathcal{V}}
\newcommand{\CZ}{\mathcal{Z}}
\newcommand{\CW}{\mathcal{W}}

\newcommand{\HRR}{\text{RR}}
\newcommand{\NSNS}{\text{NSNS}}
\newcommand{\R}{\text{R}}
\newcommand{\NS}{\text{NS}}
\newcommand{\GTVW}{\text{GTVW}}

\DeclareMathOperator{\Tr}{Tr}

\newcommand{\NN}{\mathbb{N}}
\newcommand{\ZZ}{\mathbb{Z}}
\newcommand{\RR}{\mathbb{R}}
\newcommand{\CC}{\mathbb{C}}
\newcommand{\QQ}{\mathbb{Q}}
\newcommand{\FF}{\mathbb{F}}

\newcommand{\Pn}{{\Pi^{\natural}}}

\newcommand{\cTop}{\mathsf{Top}}
\newcommand{\normord}[1]{\mathop{:} #1 \mathop{:}}

\def\beq{\begin{equation}}
\def\eeq{\end{equation}}
\def\bea{\begin{eqnarray}}
\def\eea{\end{eqnarray}}

\def\<{\langle}

\newtheorem{theorem}{Theorem}

\newtheorem{conjecture}[theorem]{Conjecture}
\newtheorem{corollary}[theorem]{Corollary}
\newtheorem{lemma}{Lemma}

\newtheorem{question}{Question}

\DeclareMathOperator{\ch}{ch}
\DeclareMathOperator{\diag}{diag}

\usepackage{tcolorbox}
\tcbuselibrary{skins,breakable}
\usetikzlibrary{shadings,shadows}
    {\endtcolorbox}

    {\endtcolorbox}

    {\endtcolorbox}

\title{Non-invertible defects from the Conway SCFT to K3 sigma models I: general results}

\author[1]{Roberta Angius \thanks{roberta.angius@csic.es}}
\author[2,4]{Stefano Giaccari}
\author[3]{Sarah M. Harrison}
\author[4]{Roberto Volpato\thanks{volpato@pd.infn.it}}
{\small \affil[1]{\small  Instituto de F\'{\i}sica Te\'orica
IFT-UAM/CSIC,  C/ Nicol\'as Cabrera 13-15, Campus de Cantoblanco, 28049 Madrid, Spain}
\affil[2]{\small Istituto Nazionale di Ricerca Metrologica, Strada delle Cacce 91, I-10135 Torino, Italy}
\affil[3]{\small  Department of Physics and Department of Mathematics,  Northeastern University, Boston, MA 02115, USA}
\affil[4]{\small Dipartimento di Fisica e Astronomia `Galileo Galilei', Universit\`a di Padova \& INFN, sez. di Padova, Via Marzolo 8, 35131, Padova, Italy}}

\begin{document}

\maketitle

\abstract{We initiate the study of supersymmetry-preserving topological defect lines (TDLs) in the Conway moonshine module $V^{f \natural}$. We show that the tensor category of such defects, under suitable assumptions, admits a surjective but non-injective ring homomorphism into the ring of $\ZZ$-linear maps of the Leech lattice into itself.  This puts strong constraints on possible defects and their quantum dimensions. We describe a simple construction of non-invertible TDLs from 
orbifolds of holomorphic (super)vertex operator algebras, which yields non-trivial examples of TDLs satisfying our main theorem. We conjecture a correspondence between four--plane--preserving TDLs in $V^{f\natural}$ and supersymmetry--preserving TDLs in K3 non-linear sigma models, which extends the correspondence between symmetry groups to the level of tensor category symmetry. We establish evidence for this conjecture by constructing non-invertible TDLs in special K3 non-linear sigma models.
}

\newpage

\tableofcontents

\section{Introduction}\label{s:intro}

In this article, we initiate a study of  topological defect lines (TDLs) in the Conway moonshine module $V^{f\natural}$ \cite{Frenkel:1988flm,Duncan:2006}, the unique holomorphic superconformal field theory (SCFT) with $c=12$ and no fields of weight 1/2. This physical theory has a fully rigorous mathematical definition as a super-vertex operator algebra (SVOA); in many respects, it is the supersymmetric analogue of the famous Frenkel-Lepowsky-Meurman Monster module $V^\natural$ \cite{Frenkel:1988flm} that plays a prominent role in the Monstrous moonshine conjectures. $V^{f\natural}$ has recently been studied in various contexts, including its relation to new supersymmetric moonshine conjectures \cite{Duncan:2014eha,Cheng:2014owa,Cheng:2015fha,Johnson-Freyd:2023uvb} and Borcherds--Kac--Moody superalgebras \cite{Harrison:2018joy}, supersymmetric string compactifications in low dimensions \cite{Harrison:2021gnp}, and as a tool to study versions of the Stolz-Teichner conjecture regarding topological modular forms \cite{Gaiotto:2018ypj,Albert:2022gcs}. As we explain below, it also has a deep connection to K3 surfaces, which has been explored from a variety of angles at the level of both CFT \cite{Duncan:2015xoa,Creutzig:2017fuk,Taormina:2017zlm,Harvey:2020jvu} and string theory \cite{Cheng:2015kha,Cheng:2016org,Paquette:2017gmb}. 

Our goal in this article is to study the category of TDLs in $V^{f\natural}$ that commute with the $\CN=1$ superconformal algebra (SCA) and satisfy some additional technical conditions, see section \ref{section:TDLs_Vfnat}. The motivation is twofold. First of all, an important open problem in conformal field theory is the study of non-invertible symmetries that preserve a non-rational chiral algebra, such as the $\CN=1$ superVirasoro algebra at $c=12$. The Conway module $V^{f\natural}$ provides a highly non-trivial example of an SCFT that is nevertheless under rigorous mathematical control. Furthermore, the presence of supersymmetry will allow us to derive some general properties of the category of TDLs.

A second motivation comes from the mysterious relationship between $V^{f\natural}$ and supersymmetric K3 non-linear sigma models (NLSMs), which are $\CN=(4,4)$ superconformal field theories with central charges $c=\bar c=6$, arising on the worldsheet of compactification of type II superstring theory on a K3 surface.\footnote{This relationship was originally formulated in \cite{Duncan:2015xoa}, building upon the connection established in \cite{Gaberdiel:2011fg} between symmetry groups of K3 NLSMs and the Conway group. We would be remiss not to mention the driving motivator for both of these works: the still--unexplained Mathieu moonshine observation of \cite{Eguchi:2010ej} relating the largest Mathieu group $M_{24}$ to the elliptic genus of K3 NLSMs. For (somewhat) recent reviews expanding on various aspects of this connection, with many citations, see \cite{Duncan:2014vfa,Anagiannis:2018jqf,Harrison:2022zee}. We will leave the exploration of any connections between the results of this article and the fascinating Mathieu moonshine phenomenon to future work.}
The evidence for this relation is based on `experimental' coincidences, and the deep nature of the connection is still somewhat obscure. What has been observed firstly is that all symmetry groups of K3 sigma models preserving the $\CN=(4,4)$ SCA and the spectral flow can be obtained as subgroups of the group of symmetries of $V^{f\natural}$ \cite{Gaberdiel:2011fg}, which is the Conway group $Co_0$. These subgroups $G_\Pn$ can be characterized by a choice of four-plane $\Pn$ in {\bf 24}, the 24-dimensional irreducible representation of $Co_0$. Moreover, for each such choice of $\Pn$, there exists a corresponding K3 NLSM with precisely $G_\Pn$ as its $\CN=(4,4)$--preserving symmetry group. 

Secondly, there is very close connection between the action of the $\mathcal N=(4,4)$--preserving symmetry groups on the 1/4--BPS spectrum of K3 NLSMs and the action of the symmetry groups $G_\Pn$ on the spectrum of $V^{f\natural}_{tw}$, the canonically-twisted module associated to $V^{f\natural}$ (Ramond sector). This begins with the observation of \cite{Duncan:2015xoa} that one can define (via the choice of $\Pn$) an $\widehat{su}(2)_1$ affine subalgebra of the $Spin(24)$ symmetry of $V^{f\natural}$ such that the function
\be\label{eq:VfnatEG}
\phi(V^{f\natural},\tau,z) = \Tr_{V_{tw}^{f\natural}}\left ((-1)^F q^{L_0-1/2} y^{J_0^3}\right),
\ee
 where $(-1)^F$ is the fermion number, and $J_0^3$ is a Cartan generator of the $\widehat{su}(2)_1$,  is none other than the elliptic genus, defined by
\be\label{eq:EG}
\phi(\calC,\tau,z)= \Tr_{\HRR}\left ((-1)^{F+\bar F} y^{J^3_0}q^{L_0-c/24}\bar q^{\bar L_0-\bar c/24}\right)
\ee
for a K3 NLSM $\calC$.\footnote{Note that in equations \eqref{eq:VfnatEG} and \eqref{eq:EG}---and throughout the paper---$q:=e^{2\pi i \tau}$ and $y:=e^{2\pi i z}$.} In equation \eqref{eq:EG}, $\Tr_{\HRR}$ denotes the trace in the Ramond--Ramond Hilbert space $\CH_\HRR$ of the worldsheet NLSM, $(-1)^F, (-1)^{\bar F}$ denote the left-- and right--moving fermion numbers,  $L_0, \bar L_0$ are the zero modes of the left-- and right--moving stress tensors, $J^3_0$  is the zero mode of the Cartan left--moving $\widehat{su}(2)_1$ in the $\CN=4$ SCA, and $c=\bar c=6$.  Notably, $\phi(\calC,\tau,z)$ is a holomorphic function in $\tau$, as the elliptic genus of a compact CFT only receives contributions from right--moving ground states, and thus it is a signed count of 1/4--BPS representations of the $\CN=(4,4)$ SCA. Moreover, it is a supersymmetric index in the sense that it is invariant under supersymmetry--preserving marginal deformations in the K3 moduli space; therefore it is the same function for all K3 NLSMs $\calC$. Specifically, it the unique weight zero, index one weak Jacobi form with Fourier expansion $2y^{-1} + 20 + 2 y + O(q)$ first computed in \cite{Eguchi:1988vra}.

Furthermore, as described in detail in section \ref{s:VfnatK3conn}, one can extend this connection to the so-called twining genus, defined by 
\be\label{eq:gTwineVfnat}
\phi^g(V^{f\natural},\tau,z):=\Tr_{V_{tw}^{f\natural}}\left(g(-1)^F q^{L_0-1/2} y^{J_0^3}\right)
\ee
for any $g\in Co_0$ which preserves a four--plane in the {\bf 24}, and 
\be\label{eq:gTwineK3}
\phi^g(\calC,\tau,z):= \Tr_{\HRR}\left (g(-1)^{F+\bar F} y^{J^3_0}q^{L_0-c/24}\bar q^{\bar L_0-\bar c/24}\right)
\ee
 for any symmetry $g$ which preserves the $\mathcal N=(4,4)$ superconformal algebra of $\calC$.\footnote{This K3-SVOA correspondence is often formulated in terms of  $V^{s\natural}$, which is an SVOA isomorphic to $V^{f\natural}$ considered here, but that carries a different action of the Conway group $Co_0$, see section \ref{s:Vfnat_def}. With respect to the relationship with K3 NLSMs, this difference is immaterial: the twining genera $\phi^g$ in $V^{f\natural}$ and $V^{s\natural}$ are the same for all $g\in Co_0$ that fix a four-plane $\Pn$.} Notably, the functions \eqref{eq:gTwineVfnat} only depend on the conjugacy class of $g\in Co_0$ (up to a lift to the spin group in certain cases \cite{Duncan:2015xoa}), and the functions \eqref{eq:gTwineK3} only depend on the conjugacy class of $g\in O^+(\Gamma^{4,20})$, the discrete T-duality group acting on the moduli space of K3 NLSMs, see \cite{Cheng:2016org} for a detailed discussion. Though these conjugacy classes are not in a one-to-one correspondence, surprisingly, it has been established that
almost all twining genera \eqref{eq:gTwineK3} of K3 NLSMs are reproduced by a corresponding twining genus in $V^{f\natural}$  of the form \eqref{eq:gTwineVfnat} \cite{Duncan:2015xoa}. The fact that this is not true for all twining genera arising in K3 NLSMs was discussed in \cite{Cheng:2016org}, and we refer the reader to Table 4 of that paper for an exhaustive review of the few exceptions. 

In fact, a more direct relation has been found between $V^{f\natural}$ and one particular K3 model, namely the model $\calC_{\GTVW}$ with 
`the largest symmetry group', discussed in \cite{Gaberdiel:2013psa}, and to which we refer  as the GTVW model after the authors of that paper. Indeed, in \cite{Creutzig:2017fuk,Taormina:2017zlm} it was argued that the holomorphic SCFT $V^{f\natural}$ can be obtained from $\calC_{\GTVW}$ by a curious `reflection procedure', that maps every field of $\calC_{\GTVW}$ with conformal weights $(h,\bar h)$ to a holomorphic field in $V^{f\natural}$ with conformal weight $h+\bar h$. Such a reflection procedure only yields a sensible holomorphic SCFT for very special NLSMs in the K3 moduli space. See \cite{Creutzig:2017fuk,Taormina:2017zlm} for further examples and  discussion of necessary conditions. This reflection procedure perhaps suggests why the twining genera of $\calC_{\GTVW}$ and $V^{f\natural}$ can be identified; an explanation for why the twining genera of other K3 models are related to the ones of $V^{f\natural}$ is still incomplete.

Given this connection between symmetry groups and twining genera of K3 NLSMs and $V^{f\natural}$, a natural question is whether there is a generalization to fusion (or, more generally, tensor\footnote{Tensor categories, unlike fusion categories, allow for infinite number of simple objects.}) category symmetry of the two theories, which notably includes non-invertible TDLs. This leads us to formulate the following general question regarding the correspondence between $V^{f\natural}$ and K3 NLSMs:
\begin{question}\label{q1}{Can the tensor category of TDLs preserving the $\CN=(4,4)$ superconformal algebra and spectral flow in every K3 NLSM be identified with a subcategory of $\CN=1$--preserving TDLs in $V^{f\natural}$? Can this be true at least at the level of fusion rings, rather than fusion categories?}
\end{question}\noindent
Actually, the arguments in \cite{Cheng:2016org} can be used to provide a negative answer to this question, at least in its strongest form concerning equivalence of fusion categories, see the discussion at the end of section \ref{s:topdefK3}. It is still an open question whether at least the fusion rings of all categories of TDLs in K3 models are reproduced by some category in $V^{f\natural}$.
Of course, as we elaborate on below, due to the fact that we are considering non-rational chiral algebras on both sides of this correspondence, these tensor categories contain infinitely many simple objects, and we must refine this question and introduce assumptions in order to make concrete progress. Nevertheless, we believe the spirit of the question is natural given the correspondence between symmetry groups on the two sides, and from this correspondence flows the converse of this question:
\begin{question}\label{q2}{Consider the tensor subcategory of $\CN=1$--preserving TDLs in $V^{f\natural}$ which fixes a given choice of four-plane in the {\bf 24}. Does there exist a K3 NLSM whose tensor category of $\CN=(4,4)$--preserving TDLs is precisely the aforementioned tensor subcategory?}\end{question}\noindent
Finally, assuming an affirmative answer to Questions \ref{q1} and  \ref{q2} (under suitable assumptions), we are naturally led to consider the correspondence at the level of defect twining genera, defined in equations \eqref{twiningVf} and \eqref{K3twining}, and which are the obvious generalizations of equations \eqref{eq:gTwineVfnat} and \eqref{eq:gTwineK3}:
\begin{question}\label{q3}{Given a correspondence between a tensor category of $\CN=1$--preserving TDLs stabilizing a four--plane in $V^{f\natural}$ and a tensor category of $\CN=(4,4)$--preserving TDLs of a particular K3 NLSM, are the corresponding defect twining genera the same?}
\end{question}

These questions lead us to investigate the topological defect lines of $V^{f\natural}$, which is an interesting question on its own. In order to have a well-posed problem, one should specify which chiral algebra we require our topological defects to preserve. The minimal requirement is that they preserve the Virasoro (or maybe the $\CN=1$ superVirasoro) algebra; however, there will be infinitely many such defects, and a classification looks hopeless. 

It is instructive to compare this situation with the well--studied case of rational CFTs with finitely many primary operators. In rational CFTs, one of the main tools to study non-invertible symmetries is the method pioneered in \cite{Petkova:2000ip}, which employs a Cardy-like condition on the defects. Specifically, a topological defect preserving  rational chiral and anti-chiral algebras in a CFT depends on finitely many complex parameters, which determine how the defect acts on the (finitely  many) primary operators. One can consider a torus partition function with the insertion of a generic defect $\CL$ along a space-like circle at fixed Euclidean time (we call it an $\CL$-twined partition function), which clearly depends on the complex parameters associated with $\CL$. On the other hand, its modular S-transformation must admit an interpretation as a partition function of a certain $\CL$-twisted Hilbert space, which itself must decompose as a  direct sum over (finitely many) irreducible representations of the chiral and anti-chiral algebra. Imposing that such an interpretation is sensible puts strong constraints on the possible values of the complex parameters. In fact, the set of solutions to such constraints forms a lattice (i.e. a free $\ZZ$-module) in the space of parameters. It is usually easy to identify a suitable set of generators of this lattice as the set of simple defects generating the fusion category of TDLs of the CFT. In the simplest case of a diagonal theory, these are the Verlinde lines \cite{Verlinde:1988sn}.

Such a construction becomes terribly complicated when the topological defects only preserve a non-rational chiral algebra. In this case, the number of parameters determining the defect becomes infinite. Furthermore, the S-transformation of the $\CL$-twined partition function involves an infinite sum, and sometimes an integral, over families of characters of irreducible representations of the algebra. The generic situation is therefore unsolvable via the standard approach.\footnote{To the best of our knowledge, a classification of topological defects in the non-rational case has been attempted only in the case of free bosons CFTs, see for example \cite{Fuchs:2007tx,Bachas:2012bj,Thorngren:2021yso}.}

In this paper, we consider  a variant of the Cardy-like condition that works for topological defects preserving some amount of worldsheet supersymmetry (and satisfying some additional technical conditions), and apply this method to the SVOA $V^{f\natural}$. Our procedure is not sufficient to determine exactly the tensor category of topological defects of $V^{f\natural}$, but it can be used to put strong constraints on the possible fusion (more generally, Grothendieck) rings for such a category. The main idea is to consider  the $\CL$-twined partition function on a torus with fully periodic boundary conditions for the fermions. Because the defect preserves the $\CN=1$ superconformal algebra, a standard argument in supersymmetry implies that the $\CL$-twined partition function is  a constant with respect to the modular parameter $\tau$ -- it is an $\CL$-equivariant version of the Witten index. On the other hand, because the torus partition function is obviously invariant under the modular S-transformation, the same constant must also represent the Witten index of the $\CL$-twisted sector, and therefore it must be an integer. Such an integrality condition must hold not only for a single defect $\CL$, but also for any product $\CL_1\CL_2\cdots$ of supersymmetry preserving defects. In the case of $V^{f\natural}$, it is known that the group of invertible symmetries preserving a given $\CN=1$ supercurrent is isomorphic to the Conway group $Co_0$. Therefore, given an $\CN=1$--preserving topological defect $\CL$, the $\CL\CL_g$--equivariant Witten index must be an integer for all invertible defects $\CL_g$, $g\in Co_0$. This simple argument is the main tool in our analysis.

Using this argument, we will prove (see Theorem \ref{th:main} in section \ref{subsec:3.2}) that if $\cTop$ is a category of $\CN=1$--preserving topological defects of $V^{f\natural}$ that contains all invertible $\CL_g$, $g\in Co_0$, then every $\CL\in \cTop$ is associated with a $\ZZ$-linear map $\rho(\CL):\Lambda\to \Lambda$ from the Leech lattice $\Lambda$ to itself. The assignment $\CL\mapsto \rho(\CL)$ is a ring homomorphism, i.e.
\be \rho(\CL_1\CL_2)=\rho(\CL_1)\rho(\CL_2)\ ,\qquad \rho(\CL_1+\CL_2)=\rho(\CL_1)+\rho(\CL_2)\ ,
\ee mapping the dual defect to the transpose endomorphism $\rho(\CL^*)=\rho(\CL)^t$, and it is surjective (but not injective) onto ${\rm End}(\Lambda)$. As a consequence, we prove in Corollary \ref{th:integralqdim} that if $\CL$ fixes one of the fields in a distinuished basis of ground states in $V^{f\natural}_{tw}$, it necessarily has integral quantum dimension. More generally, if $\CL$ fixes a `sufficiently generic' vector in this space, it is necessarily a multiple of the identity; see Corollary \ref{identity}.

Theorem \ref{th:main} in this article is formally very similar to the general consistency conditions on topological defects in K3 sigma models that were proved in \cite{Angius:2024evd}. This observation leads us to formulate  Conjecture \ref{conj:K3relation} (see section \ref{s:topdefK3}) which proposes a correspondence between subcategories of $\cTop$, whose objects are non-invertible defects in $V^{f\natural}$ that preserve a $4$-dimensional subspace $\Pn$ of Ramond ground fields, and categories of topological defects in NLSMs on K3, that preserve the $\CN=(4,4)$ algebra and spectral flow. This part of the conjecture refines and proposes an answer to Questions \ref{q1} and \ref{q2}. Concretely, we also conjecture an affirmative answer to Question \ref{q3}. We propose that there is an exact matching between the $\CL$-twined genera in the two theories, and moreover that the action of $\CL$ on the space of Ramond ground fields of $V^{f\natural}$ is the same as on the RR ground fields of the K3 models, for a suitable choice of isomorphism between the two $24$-dimensional spaces. This conjectural correspondence is the natural generalization of the observed relationship between (invertible) symmetries in the two theories.

In order to establish evidence for this conjecture, it is imperative to determine non-invertible defects in both $V^{f\natural}$ and corresponding K3 NLSMs. This is in general very difficult, though in section \ref{s:K3matching} and in a companion article \cite{Angius:2025xxx} we will construct examples  where the conjecture is indeed verified. More conceptually, we observe that the conjecture is compatible in a non-trivial way with all known consistency conditions on topological defects on both sides of the correspondence.

The article is structured as follows. In section \ref{s:Vfnat_intro} we review the Conway module $V^{f\natural}$, including its constructions, symmetries, twisted and twined partition functions, and its connection to K3 NLSMs. After beginning with a brief review of TDLs in 2d CFTs, in section \ref{section:TDLs_Vfnat} we present our two main results---Theorem \ref{th:main} and Conjecture \ref{conj:K3relation}---in sections \ref{subsec:3.2} and \ref{s:topdefK3}, respectively.  We review the method we use to construct new TDLs in orbifold theories, and work out several explicit examples for theories with finite non-abelian symmetry groups in section \ref{s:orbifolds}. The discussion in this section is fully general. In section \ref{s:K3defects}, we apply this method to $V^{f\natural}$ and certain orbifold K3 NLSMs, where we  illustrate a precise matching of defect twining functions on both sides (section \ref{s:K3matching}). Finally, we conclude with a discussion of our results in section \ref{s:discussion}. 

We also include two appendices. In appendix \ref{a:conjcons} we provide an argument for the consistency of Conjecture \ref{conj:K3relation} in the case of defects with integral quantum dimension, where we make use of lattice embedding techniques. Some additional technical calculations are found in appendix \ref{a:Zg2p2}. A companion paper \cite{Angius:2025xxx} will discuss examples of Tambara--Yamagami categories in $V^{f\natural}$ and corresponding K3 NLSMs whose duality defects satisfy our conjecture, including models  with defects of irrational quantum dimension.

\section{The Conway module $V^{f\natural}$}\label{s:Vfnat_intro}

In this section, we review  constructions of the holomorphic SVOA $V^{f\natural}$ (section \ref{s:Vfnat_def}), its symmetries and twining partition functions (section \ref{s:Conwaytwining}), and its relationship with non-linear sigma models on K3 (section \ref{s:VfnatK3conn}). 

\subsection{Constructions}\label{s:Vfnat_def}

Let us first describe the theory $V^{f\natural}$ \cite{Frenkel:1988flm,Duncan:2006}. We will present three distinct yet equivalent constructions of this theory. The simplest is as follows. Consider $12$ holomorphic free bosons ($u(1)$ currents) $i\partial X^i(z)=\sum_{n=-\infty}^{+\infty} \alpha^i_nz^{-n-1}$, $i=1,\ldots,12$, of weight $1$ with standard OPE 
$$ \partial X^i(z)\partial X^j(w)=-\frac{\delta_{ij}}{(z-w)^2}\ .
$$ We also introduce the holomorphic vertex operators 
\be  \CV_k(z):=c_k :e^{ik\cdot X(z)}:\ee
where the `momenta' $k\equiv (k^1,\ldots,k^{12})$ take values in the odd unimodular lattice $D_{12}^+$. This lattice is defined to be the union 
\be D_{12}^+:=D_{12} \cup (s+D_{12})
\ee of the root lattice of the $so(24)$ Lie algebra
\be D_{12}=\{(x_1,\ldots, x_{12})\in \ZZ^{12}\mid \sum_i x_i\in 2\ZZ\}\ ,
\ee and its translate by the vector 
\be s=(\frac{1}{2},\frac{1}{2},\ldots,\frac{1}{2})\in \RR^{12}\ . \label{spinRep:D12_real}
\ee The coset $s+D_{12}$ is the set of weights of $so(24)$ representations in the `spinor' class. The vertex operators $\CV_k(z)$ corresponds to states $|k\rangle$ that have charges $(k_1,\ldots,k_{12})$ with respect to the $u(1)$ currents $\partial X_i$, i.e.
$$ \alpha_0^i|k\rangle=k^i|k\rangle\ .
$$ The stress tensor is the standard $T(z)=\frac{1}{2}\sum_{i=1}^{12} \normord{\partial X^i\partial X^i}(z)$ with central charge $c=12$; the states $|k\rangle$ have conformal weight $k^2/2$ and are bosons (fermions) for even (odd) $k^2$. Equivalently, they are bosons when $k\in D_{12}$ and fermions when $k\in s+D_{12}$. Note that, because the lattice $D_{12}^+$ has no vectors $k$ of length $k^2=1$, the smallest conformal weight for fermions is $3/2$. There are $2^{11}=2048$ states $|k\rangle$ with this conformal weight, corresponding to $k$ of the form $(\pm \frac{1}{2},\ldots,\pm\frac{1}{2})$ with an even number of minus signs. It can be proved that $V^{f\natural}$ is the only holomorphic SVOA with $c= 12$ with no states of weight $1/2$ \cite{Creutzig:2017fuk}. In this sense, it is the supersymmetric analogue of the Frenkel-Lepowsky-Meurman Monster module $V^\natural$, which is the (conjecturally) unique holomorphic bosonic VOA with $c=24$ and no states of conformal weight $1$.

The bosonic subalgebra of $V^{f\natural}$ is generated by $i\partial X^i(z)$ and $\CV_k(z)$ for $k\in D_{12}$ and is  the VOA associated with the affine Kac-Moody algebra $\widehat{so}(24)_1$. This algebra has four distinct irreducible modules up to isomorphism: apart from the algebra itself (the vacuum module $V_0$), the remaining three modules, $V_s$, $V_v$ and $V_c$, correspond to the non-trivial cosets $s+D_{12}$, $v+D_{12}$, $c+D_{12}$ of the quotient $D_{12}^*/D_{12}$, with $s$ given in \eqref{spinRep:D12_real} and
\be v=(1,0,0,\ldots,0)\in \RR^{12}\ ,\qquad c=(-\frac{1}{2},\frac{1}{2},\ldots,\frac{1}{2})\in \RR^{12},
\ee which respectively label the spinor ($s$), vector ($v$) and conjugate spinor ($c$) representations of the Lie algebra $so(24)$. The fermionic part of $V^{f\natural}$ corresponds to the spinor representation $V_s$ of $\widehat{so}(24)_1$, so that we have the decomposition 
\be V^{f\natural}=V_0\oplus V_s\ee 
in terms of $\widehat{so}(24)_1$ modules.

In physics parlance, what we have just described is the so-called NS sector of the fermionic CFT. The Ramond sector is what mathematicians call the canonically twisted module $V^{f\natural}_{tw}$, i.e. the module twisted by the fermion number $(-1)^F$, which is the canonical $\ZZ_2$ symmetry of $V^{f\natural}$ acting trivially on bosons and by $-1$ on fermions.
The Ramond sector $V^{f\natural}_{tw}$ is generated by vertex operators $\CV_k(z)$ (and their $u(1)^{12}$ descendants)  where $k$ now takes values in the set $(v+D_{12})\cup (c+D_{12})$.
In terms of $\widehat{so}(24)_1$ modules, the Ramond sector $V^{f\natural}_{tw}$ decomposes as
\be V^{f\natural}_{tw}=V_v\oplus V_c\ .
\ee The fermion number symmetry $(-1)^F$ can be extended to $V^{f\natural}_{tw}$ by assigning fermion number $+1$ to $V_v$ and $-1$ to $V_c$. (This assignment is conventional; the opposite choice is also consistent.) 

The group of automorphisms of this holomorphic SCFT is 
\be \Aut(V^{f\natural})\cong Spin(24) ,\ee and is generated by the zero modes of the $\widehat{so}(24)_1$ currents. The center of $Spin(24)$ is $\ZZ_2\times \ZZ_2$, and one of its central $\ZZ_2$ subgroups acts trivially on $V^{f\natural}$ (i.e. the NS sector), such that the group acting faithfully is  $Spin(24)/\ZZ_2$. The generator $\mathfrak{k}$ of this trivial $\ZZ_2$ subgroup is a lift to $Spin(24)$ of $-1$ in $SO(24)$, while the generator of the other central $\ZZ_2$ subgroup is the fermion number $(-1)^F$. On the other hand, the element $\mathfrak{k}$ of $Spin(24)$ acts non-trivially (as minus the identity) on the Ramond sector $V^{f\natural}_{tw}$, so it is more convenient to consider $Spin(24)$ as the group of automorphisms, rather than $Spin(24)/\ZZ_2$.

A second construction of $V^{f\natural}$ is given by a $\ZZ_2$ orbifold of the theory $F(24)$ of $24$ free chiral fermions $\psi^i(z)$, $i=1,\ldots,24$, of weight $1/2$, where the $\ZZ_2$ symmetry acts by $\psi^i\to -\psi^i$ for all $i$. The orbifold procedure projects out all states of weight $1/2$. The untwisted sector is the affine algebra $\widehat{so}(24)_1$, with currents given by fermion bilinears $\normord{\psi^i\psi^j}(z)$, where $1\le i<j\le 24$. The twisted sector corresponds to the module $V_s$ of $\widehat{so}(24)_1$, such that the final theory is $V^{f\natural}$. Conversely, $F(24)$ can be obtained from $V^{f\natural}$ by orbifolding by the fermion number $(-1)^F$. The $24$ free fermions $\psi^1,\ldots,\psi^{24}$ of $F(24)$ correspond to the $24$ ground fields of the Ramond sector $V^{f\natural}_{tw}$.
Note that $F(24)$ can also be described in terms of a lattice theory based on the odd unimodular lattice $\ZZ^{12}$, which is given by the union
\be \ZZ^{12}:=D_{12}\cup (v+D_{12})\ .
\ee In this description the orbifold procedure projects out the states $|k\rangle$ with $k\in v+D_{12}$ and introduces the twisted states $|k\rangle$ for $k\in s+D_{12}$.

The final construction of $V^{f\natural}$ we present here is as an orbifold of the theory $V^{fE_8}$, an $\CN=1$ supersymmetrized version of the bosonic lattice VOA (chiral CFT) $V^{E_8}$ based on the $E_8$ lattice. The bosonic $E_8$ theory is obtained from $8$ chiral free bosons $X^1(z),\ldots,X^8(z)$, with vertex operators $:e^{i\lambda \cdot X(z)}:$ where the $\lambda$ take values in the $E_8$ lattice.  
The $\CN=1$ supersymmetrized version $V^{fE_8}$ is given by the product
\be V^{fE_8}=V^{E_8}\otimes F(8);
\ee
i.e. it is obtained by adding $8$ chiral free fermions $\psi^i(z)$ to the bosonic theory $V^{E_8}$, which are the superpartners of the $8$ free bosons $X^i(z)$. The $\CN=1$ supercurrent is  $\tau(z)=\sum_i \psi^i\partial X^i$. The theory $V^{f\natural}$ can be obtained by taking the orbifold of $V^{fE_8}$ by the $\ZZ_2$ symmetry acting by $X^i\to -X^i$ and $\psi^i\to -\psi^i$, that commutes with the supercurrent $\tau(z)$. The fact that the orbifold is isomorphic to $V^{f\natural}$ is quite non-trivial: an explicit isomorphism is provided in \cite{Duncan:2006}.

The latter construction shows that $V^{f\natural}$ can be endowed with an $\CN=1$ superVirasoro algebra. Duncan proved that the choice of an $\CN=1$ supercurrent $\tau(z)$ in $V^{f\natural}$, generating the standard $\CN=1$ superVirasoro algebra is unique up to the action of $Spin(24)$ automorphisms \cite{Duncan:2006}. 
He also proved that the subgroup $\Aut_\tau(V^{f\natural})$ of $Spin(24)$ that preserves a given supercurrent $\tau(z)$ is a finite group, isomorphic to the Conway group $Co_0$,
\be \Aut_\tau(V^{f\natural}):=\{g\in \Aut(V^{f\natural})\mid g(\tau)=\tau\}\cong Co_0\ .
\ee The group $Co_0$ is known to be the group of automorphisms of the Leech lattice $\Lambda$, the only $24$-dimensional positive-definite even unimodular lattice with no roots (vectors of length-squared $2$).  More precisely, the central subgroup $\ZZ_2$ of $Co_0$ is generated by $\mathfrak{k}$, so it acts trivially on the NS sector, and the group acting faithfully on $V^{f\natural}$ is $Co_0/\ZZ_2\equiv Co_1$, which is a simple group. 

As an aside, following \cite{Duncan:2014eha}, we can also define 
\be \label{VSnat} V^{s\natural}=V_0\oplus V_c\ ,\qquad\qquad V^{s\natural}_{tw}=V_v\oplus V_s\ ,
\ee which is almost the same as the definitions of $V^{f\natural}, V^{f\natural}_{tw}$, except that we have exchanged the spinor and conjugate spinor representations. One can show that $V^{s\natural}$ is also a well-defined holomorphic SCFT with central charge $c=12$ and $V^{s\natural}_{tw}$ is its canonically twisted module; in fact, one can prove that, as an SCFT, it is isomorphic to $V^{f\natural}$. So, it looks like there is nothing new if we consider $V^{s\natural}$ rather than $V^{f\natural}$. 
However, there is a subtle difference if we consider the action of $Co_0\subset Spin(24)$ fixing a weight $3/2$ field $\tau(z)$ in $V_s$. The field $\tau(z)$ is an NS field in $V^{f\natural}$ (as is natural for an $\CN=1$ supercurrent), but is a Ramond field in $V^{s\natural}$. On the other hand, the action of $Co_0$ does not fix any field of weight $3/2$ in $V_c$, i.e. it does not fix any $\CN=1$ supercurrent in  $V^{s\natural}$. In other words $V^{s\natural}$ and $V^{f\natural}$ are different representations of $Co_0$. Notice that if $G\subset Co_0$ is any subgroup fixing at least one vector in the $24$-dimensional representation of $Co_0$, then $V^{f\natural}$ and $V^{s\natural}$ are isomorphic as $G$-representations. The reason is that in this case there  is a Ramond ground field $\psi\in V^{f\natural}_{tw}(1/2)$ that is fixed by $G$, and the zero mode $\psi_0$ establishes an isomorphism of $G$-representations between $V_s$ and $V_c$.

Notably, in \cite{Creutzig:2017fuk} $V^{f\natural}, V^{fE_8}$, and $F(24)$ are proven to be the three unique (up to isomorphism) self-dual SVOAs of central charge 12. There are many ways to move between these theories by gauging subgroups of their global symmetry groups, which may or may not preserve the superconformal structure of the corresponding theories. See discussions in \cite{Duncan:2006, Duncan:2014eha, Creutzig:2017fuk, Anagiannis:2020hkk, Harrison:2020wxl}. Thus, there are many further constructions of $V^{f\natural}$ via orbifolds of $V^{fE_8}$ and $F(24)$ that we do not describe in detail here.

Finally, we note that there is a unique non-degenerate invariant bilinear form $(\cdot,\cdot):V^{f\natural} \times V^{f\natural}\to \CC$ on $V^{f\natural}$, such that the vacuum is normalized $(1,1)=1$. With resect to this bilinear form, the $L_0$-eigenspaces with different eigenvalues are orthogonal to each other. In physics, this bilinear form is defined in terms of the two-point functions on the sphere
\be (\phi_1,\phi_2):=\lim_{z\to 0} \langle Y(e^{zL_1}(-1)^{L_0^2+2L_0}z^{-2L_0}\phi_1,1/z)Y(\phi_2,0)\rangle\ ,
\ee where $Y(\phi,z)$ is the vertex operator associated to the vector $\phi\in V^{f\natural}$. One can also define an antilinear involution (CPT operator)
\be \theta:V^{f\natural}\to V^{f\natural}
\ee such that the product
\be \langle \phi_1|\phi_2\rangle:=(\theta(\phi_1),\phi_2)\ ,
\ee is the positive definite hermitian product, and defines a unitary structure on $V^{f\natural}$ (see \cite{Carpi:2023onx,Gaudio:2024zxu} for more details). The real subspace of CPT-fixed operators
\be {}^\RR V^{f\natural}\equiv \bigoplus_{n\in \frac{1}{2}\ZZ}{}^\RR V^{f\natural}(n):=\{v\in V^{f\natural}\mid \theta(v)=v\}\ ,
\ee has the structure of real vertex operator superalgebra, with $V^{f\natural}={}^\RR V^{f\natural}\otimes \CC$, and such that the invariant bilinear form is positive definite on ${}^\RR V^{f\natural}$. In particular, the stress-tensor $T(z)$ and the supercurrent $\tau(z)$ are fixed by $\theta$. Furthermore, the hermitian product and the bilinear form are both invariant under the group $\Aut_\tau(V^{f\natural})$ of automorphisms that preserve the $\CN=1$ supercurrent $\tau(z)$, and every element in $\Aut_\tau(V^{f\natural})$ commutes with the CPT operator $\theta$, so that $\Aut_\tau(V^{f\natural})$ acts on each component ${}^\RR V^{f\natural}(n)$ by (special) orthogonal transformations. The bilinear form, positive definite hermitian product, and CPT operator are also defined on the Ramond sector $V^{f\natural}_{tw}$, and the bilinear form is positive definite on the CPT-fixed subspace ${}^\RR V^{f\natural}_{tw}$.

\subsection{Symmetries and twining functions}\label{s:Conwaytwining}

Given a holomorphic SVOA $V$ with $c=12$ (i.e. $V^{fE_8}$, $F(24)$ or $V^{f\natural}$), one can define four torus partition functions $Z_{\NS}^+$, $Z_{\NS}^-$, $Z_{\R}^+$, $Z_{\R}^-$, corresponding to the four choices of spin structure:
\begin{align}
	Z^\pm_{\NS}(V,\tau)&:=\Tr_{V}(q^{L_0-\frac{c}{24}}(\pm 1)^F)\\
		Z^\pm_{\R}(V,\tau)&:=\Tr_{V_{tw}}(q^{L_0-\frac{c}{24}}(\pm 1)^F)\ .
\end{align} It is useful to organize them into a $4$-component vector
\be \CZ:=(Z^+_{\NS}, Z^-_{\NS}, Z^+_{\R}, Z^-_{\R})^t
\ee that transforms as a vector-valued modular function under $SL_2(\ZZ)$-transformations
\be \CZ(\tau+1)=\rho(T)\CZ(\tau)\ ,\qquad\qquad \CZ(-1/\tau)=\rho(S)\CZ(\tau)
\ee where the $SL_2(\ZZ)$-generators $T=\left(\begin{smallmatrix}
	1 & 1\\ 0 &1
\end{smallmatrix}\right)$ and $S=\left(\begin{smallmatrix}
0 & -1 \\1 &0
\end{smallmatrix}\right)$ are represented by the matrices\footnote{For a generic central charge $c\in \frac{1}{2}\ZZ$, the minus signs in $\rho(T)$ should be replaced by $e^{-2\pi i c/24}$. }
\be\label{modularspin} \rho(T)=\begin{pmatrix}
0 & -1 & 0 & 0\\
-1 & 0 & 0 & 0\\
0 & 0 & 1 & 0\\
0 & 0 & 0 & 1
\end{pmatrix}\ ,\qquad \rho(S)=\begin{pmatrix}
1 & 0 & 0 & 0\\
0 & 0 & 1 & 0\\
0 & 1 & 0 & 0\\
0 & 0 & 0 & 1
\end{pmatrix}\ .
\ee
Note that $Z^-_{\R}(V,\tau)$ is modular invariant, and in fact it is a constant that equals the Witten index
\be Z^-_{\R}(V,\tau)=\Tr_{V_{tw}(1/2)}((-1)^F)\ ,
\ee counting the number of bosonic minus fermionic Ramond ground states. This follows because for unitary representations of the Ramond $\CN=1$ algebra, each $L_0$-eigenspace contains the same number of bosons and fermions, except possibly for the lowest possible $L_0$-eigenvalue $h=c/24=1/2$, corresponding to the eigenspace $V_{tw}(1/2)$.  The Witten index of $V=V^{f\natural}$ is $24$, whereas it is zero for the other two self-dual SVOAs of $c=12$, $V^{fE_8}$ and $F(24)$. 

Consider a finite symmetry group $G$ fixing an $\CN=1$ superconformal algebra in $V$, and whose action on both the NS and R sectors $V$ and $V_{tw}$ commutes with $(-1)^F$. Then, one can define a corresponding $g$-twined partition function $\CZ^g:=(Z^{g,+}_{\NS}, Z^{g,-}_{\NS}, Z^{g,+}_{\R}, Z^{g,-}_{\R})^t$ with components
 \begin{align}\label{eq:SymTwineNS}
 	Z^{g,\pm}_{\NS}(V,\tau)&:=\Tr_{V}(g q^{L_0-\frac{c}{24}}(\pm 1)^F)=\sum_{h\in \frac{1}{2}\ZZ} q^{h-\frac{c}{24}}\Tr_{V(h)}(g(\pm 1)^F)\\\label{eq:SymTwineR}
 	Z^{g,\pm}_{\R}(V,\tau)&:=\Tr_{V_{tw}}(g q^{L_0-\frac{c}{24}}(\pm 1)^F)=\sum_{h\in \frac{1}{2}+\ZZ} q^{h-\frac{c}{24}}\Tr_{V_{tw}(h)}(g(\pm 1)^F)\ ,
 \end{align}
 which only depends on the conjugacy of $g\in G$ and which transforms as a vector-valued modular form under a suitable subgroup of $SL(2,\ZZ)$. 
Furthermore, one also has the $g$-twisted partition function $\CZ_g:=(Z^+_{g,\NS}, Z^-_{g,\NS}, Z^+_{g,\R}, Z^-_{g,\R})^t$ with components
 \begin{align}\label{eq:SymTwistNS}
 	Z^\pm_{g,\NS}(V,\tau)&:=\Tr_{V_g}(q^{L_0-\frac{c}{24}}(\pm 1)^F)=\sum_h q^{h-\frac{c}{24}}\Tr_{V_g(h)}((\pm 1)^F)\\\label{eq:SymTwistR}
 	Z^\pm_{g,\R}(V,\tau)&:=\Tr_{V_{g,tw}}(q^{L_0-\frac{c}{24}}(\pm 1)^F)=\sum_h q^{h-\frac{c}{24}}\Tr_{V_{g,tw}(h)}((\pm 1)^F)\ .
 \end{align}
where $V_g, V_{g,tw}$ correspond to $g$-twisted NS and R Hilbert spaces, respectively. Here, we assume that $(-1)^F$ can be extended to a well-defined involution on both $V_g$ and  $V_{g,tw}$ (see section \ref{s:TDLsuperVOA} and \ref{subsec:3.2} for a discussion).  Notably, $\CZ^g$ and $\CZ_g$ are related by the modular S-transformation,
\be
\CZ_g(-1/\tau)=\rho(S)\CZ^g(\tau),
\ee
where $\rho(S)$ is given in \eqref{modularspin}.

Specializing to the case of $V=V^{f\natural}$, where the $\mathcal N=1$--preserving symmetry group is $G=Co_0$ (acting faithfully only on $V^{f\natural}_{tw}$), we can give a simple formula for $\CZ^g$ for all $g\in Co_0$, following \cite{Duncan:2014eha} (see also section 6 of \cite{Harrison:2018joy}).  First, 
to each element $g\in Co_0$, we can associate a 24-dimensional Frame shape $\pi_g$ encoding the eigenvalues of $g$ in $\rho_{24}$, the 24-dimensional irreducible representation of $Co_0$. More precisely, define the formal product
\be \label{FrameShape}
\pi_g:= \prod_{\ell | N} \ell^{k_\ell},
\ee
where $N=o(g)$ is the order of $g$, and $k_\ell$ are integers determined by the characteristic polynomial,
\be
\det(t {\bf 1}_{24}-\rho_{24}(g))= \prod_{\ell |N}(t^\ell-1)^{k_\ell},
\ee
of $g$ in $\rho_{24}$. The roots of this polynomial determine the 24 eigenvalues, appearing in complex conjugate pairs, of $\rho_{24}(g)$, which we will parametrize as
\be\label{gEigenvalues}
\{e^{-2\pi i \beta_{g,1}}, e^{2\pi i \beta_{g,1}}, \ldots,e^{-2\pi i \beta_{g,12}}, e^{2\pi i \beta_{g,12}}\},
\ee
for $\beta_{g,i}\in [0,1/2]$.
Moreover, let $-g \in Co_0$ denote the product of $g$ with the generator $\mathfrak{k}$ of the $\ZZ_2$ center of $Co_0$, to which we can associate the dual Frame shape $\pi_{-g}$ encoding the eigenvalues $\beta_{-g,i}:=1/2-\beta_{g,i}$, $i=1,\ldots,12$. 

Now, given a Frame shape (or dual Frame shape), define the eta product $\eta_g$ to be
\be
\eta_g(\tau):= \prod_{\ell | N}\eta(\ell\tau)^{k_\ell},
\ee
where $\eta(\tau)$ is the Dedekind eta function 
\be
\eta(\tau)=q^{1/24}\prod_{n=1}^\infty (1-q^n).
\ee
Finally, introduce the products $\calC_g:=\pm\prod_{i=1}^{12}(e^{\pi i \beta_{g,i}}-e^{-\pi i \beta_{g,i}})$ (and analogously for $\calC_{-g}$), which are determined up to a sign corresponding to the choice of square root of the 24 eigenvalues of $\rho_{24}(g)$. The two choices of sign correspond to the two lifts of $\rho_{24}(\mathfrak{k}g)\in SO(24)$ to $\Aut(V^{f\natural})\cong Spin(24)$. Only one of these two lifts fixes the supercurrent $\tau$, and this is the correct choice of sign for $\calC_g$; see \cite{Duncan:2015xoa} for a list of the correct values. This can also be determined in terms of irreducible characters of $Co_0$ by
\be \calC_{\pm g}=\Tr_{\bf 1}(g)+\Tr_{\bf 276}(g)+\Tr_{\bf 1771}(g)\mp \Tr_{\bf 24}(g)\mp \Tr_{\bf 2024}(g)\ .
\ee Given this data, we can now determine, for all $g\in Co_0$, a general formula for $\CZ^g$ given by 
 \begin{align}\label{eq:NSgtwining}
 	Z^{g,\pm}_{\NS}(V^{f\natural},\tau)&={1\over 2}\left({\eta_g(\tau/2)\over \eta_g(\tau)} + {\eta_{-g}(\tau/2)\over \eta_{-g}(\tau)} \pm \calC_g\eta_g(\tau) \pm \calC_{-g}\eta_{-g}(\tau)\right )\\\label{eq:Rgtwining}
 	Z^{g,\pm}_{\R}(V^{f\natural},\tau)&={1\over 2}\left(-{\eta_g(\tau/2)\over \eta_g(\tau)} + {\eta_{-g}(\tau/2)\over \eta_{-g}(\tau)} \mp \calC_g\eta_g(\tau) \pm \calC_{-g}\eta_{-g}(\tau)\right ).
 \end{align}
  Note that \eqref{eq:NSgtwining} is symmetric under interchange of the Frame shape $\pi_g$ with its dual Frame shape $\pi_{-g}$, indicating that only $Co_0/\ZZ_2=Co_1$ acts faithfully on $V^{f\natural}$.

Later we will generalize these partition functions to include both non-invertible TDLs \eqref{eq:DefectTwineNS}--
\eqref{eq:DefectTwistR} and twisted--twined versions \eqref{eq:TwistTwineNS},\eqref{eq:TwistTwineR}.

\subsection{The connection with K3 non-linear sigma models}\label{s:VfnatK3conn}

Let us now describe the connection with K3 NLSMs. As mentioned above, the theory $V^{f\natural}$ (the NS sector) does not contain any states with conformal weight $1/2$; however, the Ramond sector $V^{f\natural}_{tw}$ contains a $24$-dimensional space ${}^\RR V^{f\natural}_{tw}(1/2)\cong \mathbb{R}^{24}$  of states with weight $1/2$, and $Co_0\subset SO(24)$ acts on these states in the standard vector representation. 
For any $4$-dimensional subspace $\Pn\subset {}^\RR V^{f\natural}_{tw}(1/2)$, let 
\be G_\Pn:=\{g\in Co_0\subset SO(24)\mid g_{\rvert \Pn}=\mathrm {id}_\Pn\}\subset Co_0\ee 
be the subgroup of $Co_0$ fixing $\Pn$ pointwise. As shown in \cite{Gaberdiel:2011fg}, for each such group $G_\Pn$, there exists a NLSM on K3 whose faithful group of symmetries commuting with the $\CN=(4,4)$ algebra and spectral flow is isomorphic to $G_\Pn$. Conversely, any such group of symmetries of a K3 NLSM is isomorphic to $G_\Pn\subset Co_0$ for some choice of $\Pn$. 

There is also a close relationship between the action of $G_\Pn$ on the fields of $V^{f\natural}$ and its action on the fields of the corresponding K3 sigma model. The choice of $\Pn\subset \mathbb{R}^{24}$ determines a subgroup $SO(4)\times SO(20)\subset  SO(24)$, and the group $G_\Pn \equiv Co_0\cap SO(20)$ commutes with the subgroup $SO(4)\subset SO(24)$. This means that the action of $G_\Pn$ on $V^{f\natural}$ leaves invariant an affine subalgebra $\widehat{so}(4)_1=\widehat{su}(2)_1\oplus \widehat{su}(2)_1$ of $\widehat{so}(24)_1$. Following \cite{Duncan:2015xoa}, choose 
one of the $\widehat{su}(2)_1$ that is fixed by $G_\Pn$. As explained in detail in section 3 of \cite{Cheng:2015kha}, from this $\widehat{su}(2)_1$ and the supercurrent $\tau(z)$ one can construct a copy of the small $\CN=4$ superconformal algebra at central charge $c=6$, precisely the chiral algebra present in K3 NLSMs, which is also fixed by $G_\Pn$. (See equations (3.6) and (3.7) of loc. cit.) 

Now, again following \cite{Duncan:2015xoa}, choose a generator $J_0^3$ in the $\widehat{su}(2)_1$ which forms part of the $c=6, \CN=4$ SCA.\footnote{Though we do not make use of it in this article, we also note that in \cite{Cheng:2015kha} (see also \cite{Taormina:2017zlm}), it is explained that one can further grade the Hilbert space of $V^{f\natural}, V^{f\natural}_{tw}$ by another $U(1)$ current which lies in the {\it second} $\widehat{su}(2)_1 \subset \widehat{so}(4)_1$ and which commutes with the $c=6$, $\CN=4$ SCA described above. See equations (3.14)--(3.18) of loc. cit. } Moreover, for each four--plane--preserving element $g\in G_\Pn$, define
\be \phi^g(V^{f\natural},\tau,z):=\Tr_{V^{f\natural}_{tw}}(g (-1)^Fq^{L_0-\frac{1}{2}}y^{J_0^3})\ .
\ee  
As mentioned in section \ref{s:intro}, for $g=e$, this function is a weight zero, index one weak Jacobi form precisely equal to the elliptic genus of a K3 NLSM.
We can write  an explicit formula for this function given the Frame shape $\pi_g$ of a four--plane--preserving element of $Co_0$ \cite{Duncan:2015xoa},
\be\label{eq:D-MCtwining}
\phi^g(V^{f\natural},\tau,z) = {1\over 2\eta^{12}(\tau)}\sum_{j=1}^4  \epsilon_{g,j} \, \theta_j^2(\tau,z) \prod_{i=1}^{10}  \theta_j^2(\tau,\beta_{g,i}),
\ee
where the $\theta_j(\tau,z)$ are the standard Jacobi theta functions, the $\beta_{g,i}$ are defined in \eqref{gEigenvalues}, we set $\beta_{g,11}=\beta_{g,12}=0$ since $g$ preserves a four--plane, 
and 
 $$
\epsilon_{g,j} = \begin{cases}\mp 1& j=1\\ -{\Tr_{\bf 4096} g \over 4\prod^{10}_{i=1} ( e^{- \pi i\beta_{g,i}}+e^{ \pi i\beta_{g,i}})}&j=2 \\ 1 & j=3\\ -1
&j=4 \end{cases}$$
and where ${\bf 4096} ={\bf 1 + 276 + 1771 + 24 + 2024}$ is the decomposition of the ${\bf 4096}$ into irreps of $Co_0$.  It should be noted that there is sometimes an ambiguity in lifting an element $g$ of $Co_0\cap SO(20)$ to the spin group $Spin(20)$, and this might lead to two different twining genera $\phi^g$ corresponding to the choice of sign in $\epsilon_{g,1}$. This ambiguity can occur only when the subspace of $\RR^{24}$ fixed by $g$ is exactly four-dimensional.  The functions defined in eq. \eqref{eq:D-MCtwining} are weak Jacobi forms of weight zero, index one, and level $N=o(g)$, possibly with multiplier.

The conjecture of \cite{Duncan:2015xoa} is that the functions $\phi^g$ obtained in this way are exactly the twining genera of K3 models for the corresponding symmetry, i.e.\footnote{Note that the twining genera for the groups $G_\Pn\subset Co_0$ are the same in the SVOA $V^{f\natural}$ and $V^{s\natural}$. This is true because, as discussed in section \ref{s:Vfnat_def}, $V^{f\natural}$ and $V^{s\natural}$ are isomorphic as $G$-representations for any subgroup $G\subset Co_0$ fixing at least one vector in the $24$-dimensional representation of $Co_0$.}
\be\label{K3Vfnatmathc}
\phi^g(V^{f\natural},\tau,z)=\phi^g(\calC,\tau,z)
\ee
for some K3 NLSM $\calC$ with symmetry $g$ of the same order and all four--plane--preserving $g\in Co_0$. Moreover, it was conjectured that the $g$-twined elliptic genus in any K3 NLSM can be reproduced by some $\phi^g(V^{f\natural},\tau, z)$ as in \eqref{eq:D-MCtwining}.

In \cite{Cheng:2016org} it was argued that this was true for almost all twining genera arising in K3 NLSMs, but with some exceptions, i.e. there exist a few twining genera in K3 NLSMs which are not reproduced by a formula of the form \eqref{eq:D-MCtwining} for some four--plane--preserving element $g\in Co_0$.\footnote{A review of the subtleties of this connection is beyond the scope of this article. Considerations of worldsheet parity in the K3 moduli space combined with evidence from symmetries of UV Landau--Ginzburg orbifolds \cite{Cheng:2015rby} suggest that a construction in \cite{Cheng:2014zpa} based on the Niemeier lattices may be necessary for describing all twining genera of K3 NLSMs; see \cite{Cheng:2016org} for a more complete discussion and conjectures, and \cite{Paquette:2017gmb} for a physics proof.}  In the cases where there is an ambiguity in lifting to the spin group, it turns out that the both twining genera $\phi^g$ are realized in (possibly different) K3 models.

One of our primary goals is to generalize this relationship to non-invertible TDLs, see Conjecture \ref{conj:K3relation}
in section \ref{s:topdefK3}.

\section{Topological defects in $V^{f\natural}$}
\label{section:TDLs_Vfnat}

In this section, after a general review of topological defects in two dimensional conformal and superconformal field theories (sections \ref{s:topdefCFT} and \ref{s:TDLsuperVOA}), in section \ref{subsec:3.2}, we describe the main results of our article concerning the tensor category of $\CN=1$--preserving topological defects in the SVOA $V^{f\natural}$. Furthermore, we formulate a conjecture relating particular subcategories of such defects in $V^{f\natural}$ to topological defects in K3 sigma models in section \ref{s:topdefK3}. Section \ref{s:proof} contains the proof of Theorem \ref{th:main}.

\subsection{Generalities on topological defects}\label{s:topdefCFT}

In this section, we briefly review some of the main properties of TDLs in two dimensional conformal field theory. In particular, we  focus on compact and unitary theories, and ultimately we will specialize the discussion to the holomorphic case (i.e. for holomorphic vertex operators algebras). See, for example, \cite{Moller:2024xtt,Carqueville:2023jhb, Chang_2019,Frohlich:2009gb} for more details.

Let us start by considering a unitary bosonic conformal field theory with a unique vacuum on a $2$-dimensional Euclidean spacetime (worldsheet) $\Sigma$. One can consider correlation functions
\be \langle O_1(z_1)\cdots O_n(z_n) \CL_1(\gamma_1)\cdots \CL_m(\gamma_m)\rangle_\Sigma\ ,
 \ee where, besides the local operators $O_i(z_i)$ supported at points $z_1,\ldots, z_n\in \Sigma$, we have also insertions of defects $\CL_i(\gamma_i)$ supported on oriented lines $\gamma_i\subset \Sigma$. The lines can be either closed or open, although in the latter case suitable `defect starting' and `defect ending' point operators must be specified at the endpoints (see below). A line defect $\CL(\gamma)$ is \emph{topological} if all correlation functions are invariant under infinitesimal deformations of the support $\gamma$, as long as the line is not moved across the support of some other (point or line) operator insertion. One can define a topological defect line $\CL$ by specifying how correlation functions change when the support of $\CL$ is moved across any local operator $O(z)$. All line defects in this article will be topological, so we will sometimes omit the word `topological' for short.

 We will say that a local operator $O(z)$ is preserved by a topological defect line $\CL (\gamma)$, or that $\CL$ and $O$ are transparent to each other, if the support of $\CL $ can be moved across the support of $O(z)$ without changing any correlation function. In particular, the holomorphic and anti-holomorphic stress-energy tensors $T(z)$ and $\tilde T(\bar z)$ are always transparent to any topological line defect $\CL$. More generally, if a certain set $\{T(z),\phi_1(z),\phi_2(z),\ldots \}$ of holomorphic fields is preserved by a given defect $\CL$, then  $\CL$ automatically preserves  the full chiral algebra $\mathcal{A}$  generated by these fields. In this case, we will also say that $\mathcal{A}$ commutes with $\CL$. An analogous statement holds for the anti-chiral algebra, which is generated by the anti-holomorphic fields that are preserved by $\CL$.

Let us consider the case where $\Sigma=S^1\times \RR$, where we interpret $S^1$ as the compactified `space' direction and $\RR$ as the Euclidean time direction. We denote by $\Hh$ the Hilbert space of states on $S^1$. By the standard conformal mapping of $S^1\times \RR$ to $\CC\setminus \{0\}$, an asymptotic state at time $t\to - \infty$, or $t\to +\infty$, is mapped to a local operator at the origin, or at $\infty$, of the Riemann sphere $ \hat{\CC}=\CC\cup \{\infty\}$, thus  implementing the usual state-operator correspondence.

In this setup, a given topological line defect can be inserted along the space circle $S^1$ at fixed time; in particular, if one can move the support of $\CL$ to $t\to -\infty$ without crossing any other operator, then the initial state $|\psi\rangle$ is mapped to a new state, that we denote by $\hat\CL|\psi\rangle$. Thus, with every $\CL$ one can associate a linear operator $\hat\CL:\Hh\to \Hh$ on the space of states. By conformal mapping to the Riemann sphere, the equivalent statement is that if a line defect $\CL$ on a small circle surrounding a local operator $\psi(0)$ at the origin is shrunk to a point, then we get a new local operator $(\hat\CL\psi)(0)$\footnote{Strictly speaking, there can in principle be a phase mismatch in passing from the operator $\hat\CL$ on the cylinder $S^1\times \RR$ to the corresponding operator acting on the Riemann sphere $\hat\CC$. This is sometimes called an isotopy anomaly (see for example \cite{Chang_2019} for more details). For definiteness, we will define the operator $\hat\CL$ in the cylinder.}, see Figure \ref{fig1}. The linear map $\hat\CL:\Hh\to\Hh$ commutes with the (holomorphic and antiholomorphic) Virasoro algebras, and more generally with all modes in the chiral and anti-chiral algebras $\mathcal{A}$ and $\tilde{\mathcal{A}}$ that are preserved by $\CL$. This means that $\hat\CL$ maps primary states  of $\mathcal{A}\times\tilde{\mathcal{A}}$ to primary states in the same representation, and that the action on such primary states is sufficient to determine $\hat\CL$ completely. Note that the vacuum $|0\rangle \in \Hh$ is necessarily an eigenstate of $\hat\CL$
\be \hat\CL|0\rangle=\langle\CL\rangle |0\rangle\ ,
\ee where the eigenvalue $\langle \CL\rangle$ is called the quantum dimension of $\CL$. In unitary, compact theories (i.e. where $L_0$ and $\bar L_0$ have discrete spectrum) with a unique vacuum, $\langle \CL\rangle$ is a real number satisfying the condition $\langle \CL\rangle\ge 1$ \cite{Chang_2019}.

\begin{figure}[h!]
    \centering
    \includegraphics[width=0.9\linewidth]{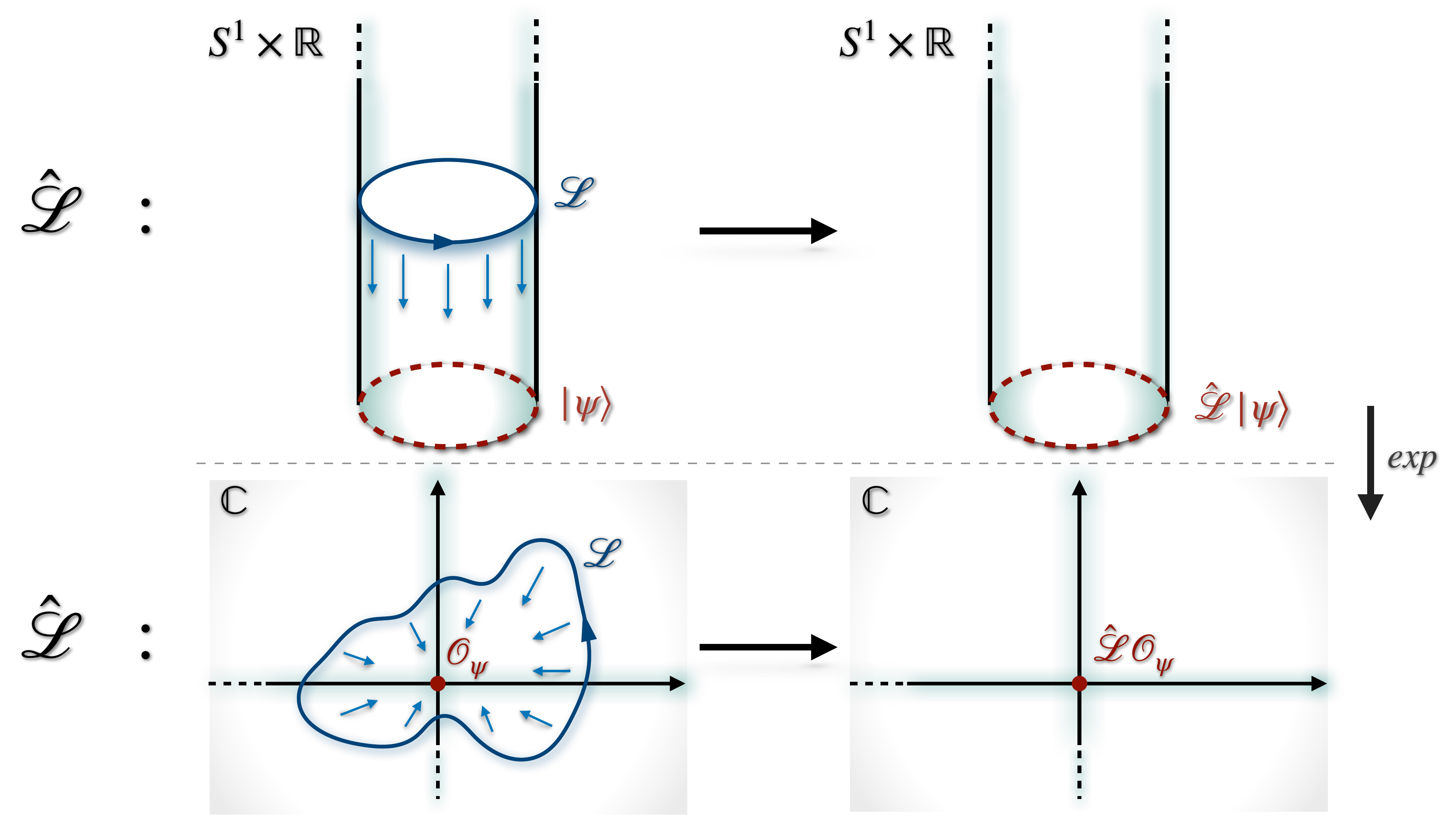}
    \caption{\small{Representation of the action of the linear operator $\hat{\mathcal{L}}$ on the Hilbert space of states on $S^1$ and on the corresponding set of local operators on $\hat{\mathbb{C}}$ as related by the conformal mapping.}  }
    \label{fig1}
\end{figure}

One can also insert a line defect $\CL$ on $S^1\times \RR$ along the Euclidean time direction at a fixed  point in space. In this case, the space of states $\Hh$ on $S^1$ is modified to a new Hilbert space $\Hh_\CL$ (states on $S^1$ in the presence of the defect), that we will call the $\CL$-twisted Hilbert space. By conformal mapping to $\hat\CC$, the space $\Hh_\CL$ is identified with the space of point operators starting a defect line $\CL$; see Figure \ref{fig2}. The space $\Hh_\CL$ is an honest representation of the preserved subalgebras  $\mathcal{A}$ and $\tilde{\mathcal{A}}$. More generally, one expects  the OPE between a local operator $\psi(z)$, with $\psi\in \Hh$, and a defect starting operator $\phi(0)$, $\phi\in \Hh_\CL$, to produce a new operator in $\Hh_\CL$. This OPE is in general not local due to the presence of the line defect. In particular, it is not invariant under moving the local operator $\psi(z)$ along a circle around $\phi(0)$.

\begin{figure}[h!]
    \centering
    \includegraphics[width=0.6\linewidth]{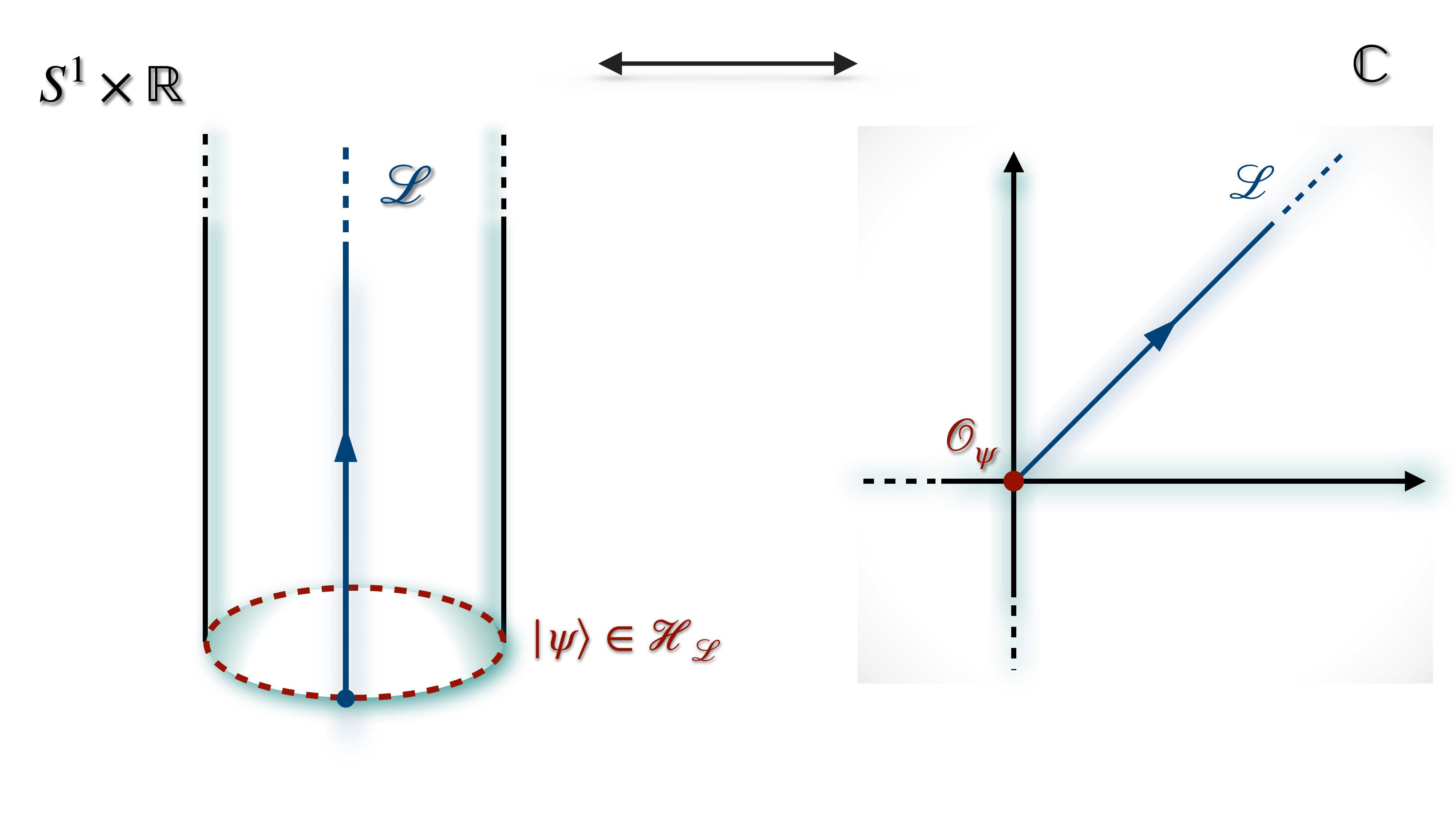}
    \caption{\small{The insertion of a line defect $\mathcal{L}$ along the Euclidian time direction on $S^1 \times \mathbb{R}$ generates the new $\mathcal{L}-$twisted Hilbert space $\mathcal{H}_{\mathcal{L}}$, which states are mapped in defect starting operator for the defect $\mathcal{L}$ in $\hat{\mathbb{C}}$.}}
    \label{fig2}
\end{figure}

In summary, each type of topological defect $\CL$  is associated with a linear operator $\hat\CL:\Hh\to \Hh$ and an $\CL$-twisted space $\Hh_\CL$. These two pieces of data are related by modularity.  Specifically, consider the CFT on a torus $S^1\times S^1$, where the first circle is the space direction and the second circle the compactified Euclidean time, with complex structure determined by the modular parameter $\tau$ in the  upper half-space  $\mathbb{H}:=\{\tau\in\CC\mid \Im\tau>0\}$ . Then by inserting the line defect along the space $S^1$ at fixed time, we get the $\CL$-twined partition function
\be \label{eq:genLtwine} Z^\CL(\tau):=\Tr_{\Hh}(q^{L_0-\frac{c}{24}}\bar q^{\bar L_0-\frac{c}{24}}\hat\CL)\ ,
\ee while inserting $\CL$ along the time circle at fixed space we get the $\CL$-twisted partition function
\be  Z_\CL(\tau):=\Tr_{\Hh_\CL}(q^{L_0-\frac{c}{24}}\bar q^{\bar L_0-\frac{c}{24}})\ .
\ee Therefore, one expects that $Z^\CL$ and $Z_\CL$ are related by a modular transformation that exchanges the two circles 
\be Z_\CL(\tau)=Z^\CL(-1/\tau)\ .
\ee

From a mathematical perspective, topological defects are expected to be described as objects in a fusion or, more generally, a tensor category (see for example \cite{Etingof:2015} for an introduction to the subject).  In particular, the set of topological defect lines in a given CFT is expected to exhibit the following properties:
\begin{itemize}
\item There is always an identity defect $\CI$, whose insertion does not modify any correlation function. The associated linear operator and $\CI$-twisted sector are $\hat\CI={\rm id}_{\Hh}$ and $\Hh_{\CI}=\Hh$.  
    \item By moving two parallel lines $\CL_1$ and $\CL_2$ very close to each other, we get a new topological line, which is the fusion $\CL_1\otimes \CL_2$, or simply $\CL_1\CL_2$. The corresponding line operator is just the product of operators $\widehat{\CL_1\CL_2}=\hat\CL_1\hat\CL_2$, and the $\CL_1\CL_2$-twisted sector is a fusion product $\Hh_{\CL_1\CL_2}\cong \Hh_{\CL_1}\otimes\Hh_{\CL_2}$. 
    \item There is an involution $\CL\mapsto \CL^*$ acting on the set of topological defects, which corresponds to changing the orientation of the support, i.e.
    \be \CL^*(\gamma)\cong \CL(\bar\gamma)\ ,
    \ee where $\bar\gamma$ is the line $\gamma$ with reverse orientation. As we defined the $\CL$-twisted sector as the space of $\CL$-defect starting operators, the $\CL^*$-twisted sector space $\Hh_{\CL^*}$ can be interpreted as the space of $\CL$-defect ending operators. We assume the existence of a non-degenerate two-point function $\langle \phi(z)\psi(w)\rangle_{\hat \CC}$ between $\phi\in \Hh_\CL$ and $\psi\in \Hh_{\CL^*}$. This establishes an isomorphism $\Hh_{\CL^*}\cong (\Hh_\CL)^*$ between $\Hh_{\CL^*}$ and the dual of the Hilbert space $\Hh_\CL$. Furthermore, the operators $\hat\CL$ and $\hat{\CL^*}$ are related by
    \be (\hat{\CL^*}u,v)=(u,\hat\CL v)\ ,
    \ee where $(\cdot,\cdot)$ is the bilinear form defined by the $2$-point function. This means that the adjoint $\hat\CL^\dag$ with respect to the Hilbert space scalar product is given by
    \be \hat\CL^\dag=\theta \hat{\CL^*}\theta\ ,
    \ee where $\theta=\theta^{-1}$ is the CPT antilinear involution. For invertible defects, $\hat{\CL}$ is a unitary operator and $\hat{\CL}^*$ is its inverse, so that $\hat{\CL^*}=\hat\CL^\dag$, i.e  $\hat\CL$ and $\hat\CL^*$ commute with the CPT operator $\theta$. The relation $\hat{\CL^*}=\hat\CL^\dag$ is actually true also for all examples of non-invertible defects in unitary theories that we are aware of.
    \item There is a superposition $\CL_1+\CL_2$ such that $\widehat{\CL_1+\CL_2}=\hat\CL_1+\hat\CL_2$ and $\Hh_{\CL_1+\CL_2}=\Hh_{\CL_1}\oplus\Hh_{\CL_2}$.
    \item The set of morphisms $\Hom(\CL_1,\CL_2)$ is a finite dimensional $\CC$-vector space of topological junction operators that join an (incoming) line $\CL_1$ to an (outgoing) line $\CL_2$. As usual, the word `topological' means that the junction can be moved (while staying attached to the lines $\CL_1$ and $\CL_2$) without changing any correlation function, as long as the deformation does not cross the support of another operator. Every $u\in \Hom(\CL_1,\CL_2)$, Figure \ref{fig3}-(a), corresponds to a linear map $u:\Hh_{\CL_1}\to \Hh_{\CL_2}$ that commutes with the preserved algebras $\mathcal{A}$ and $\tilde{\mathcal{A}}$. One can consider topological three way junctions $v\in \Hom(\CL_1\otimes \CL_2,\CL_3)$, Figure \ref{fig3}-(b), or $w\in \Hom(\CL_1, \CL_2\otimes \CL_3)$, or more general $k$-way topological junctions, that are obtained by composing $3$-way junctions.

    \begin{figure}[h!]
        \centering
        \includegraphics[width=0.6\linewidth]{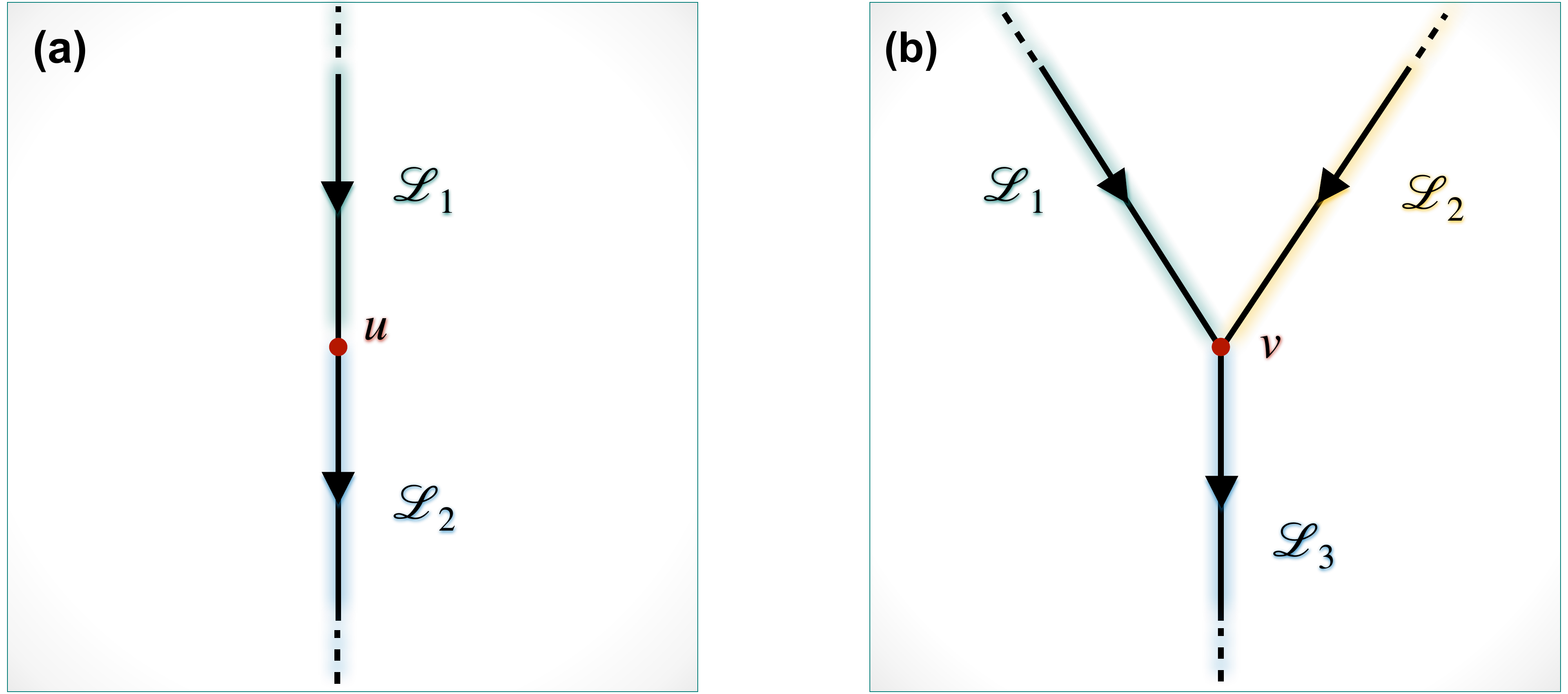}
        \caption{\small{(a) Topological 2-junction corresponding to the operator $u \in \Hom(\CL_1,\CL_2)$. (b) Topological 3-junction corresponding to the operator $v \in \Hom(\CL_1\otimes \CL_2,\CL_3)$.}}
        \label{fig3}
    \end{figure}
    
    \item For each topological defect $\CL$, there is always the identity topological junction between $\CL$ and itself, so that $\dim_\CC \Hom(\CL,\CL)\ge 1$. A defect $\CL$ is \emph{simple} if $\dim_\CC \Hom(\CL,\CL)= 1$. We will always assume that the category of topological defects is semisimple, i.e. the identity $\CI$ is simple and every $\CL$ can be decomposed as the superposition of simple defects $\CL=\sum_{\text{simple }i}n_i\CL_i$, $n_i\in \ZZ_{\ge 0}$. In general, fusion categories contain only a finite number of (isomorphism classes of) simple defects. In this article, we do not want to impose this condition, so, strictly speaking, we are in the realm of tensor category.
    \item The fusion of two simple defects $\CL_i$ and $\CL_j$ is not necessarily simple, and decomposes as
    \be \CL_i\CL_j=\sum_{\text{simple }k} n_{ij}^k \CL_k\ ,
    \ee for some non-negative integral fusion coefficients $n_{ij}^k\in \ZZ_{\ge 0}$; see Figure \ref{fig4} for a pictorial representation. One has \be n_{ij}^k=\dim_\CC\Hom(\CL_i\otimes\CL_j,\CL_k)\ ,\ee i.e. $n_{ij}^k$ is the dimension of the space of $3$-way junctions with incoming $\CL_i$, $\CL_j$ and outgoing $\CL_k$. The Grothendieck ring\footnote{This is a fusion ring when the number of simple objects is finite; more generally, it is a unital based ring, see \cite{Etingof:2015}.} of the tensor category is the ring given by formal $\ZZ$-linear combinations of (isomorphism classes of) objects, with a distinguished basis given by simple objects, the sum given by superposition and the product by the fusion product, i.e. by the coefficients $n_{ij}^k$. Quantum dimensions satisfy
    \be \langle \CL_i\rangle\langle \CL_j\rangle=\sum_k n_{ij}^k \langle \CL_k\rangle\ .
    \ee Given the fact that quantum dimensions are always $\ge 1$, this implies that the fusion of any two simple defects decomposes into a finite number of simple defects.
    \begin{figure}
        \centering
        \includegraphics[width=0.6\linewidth]{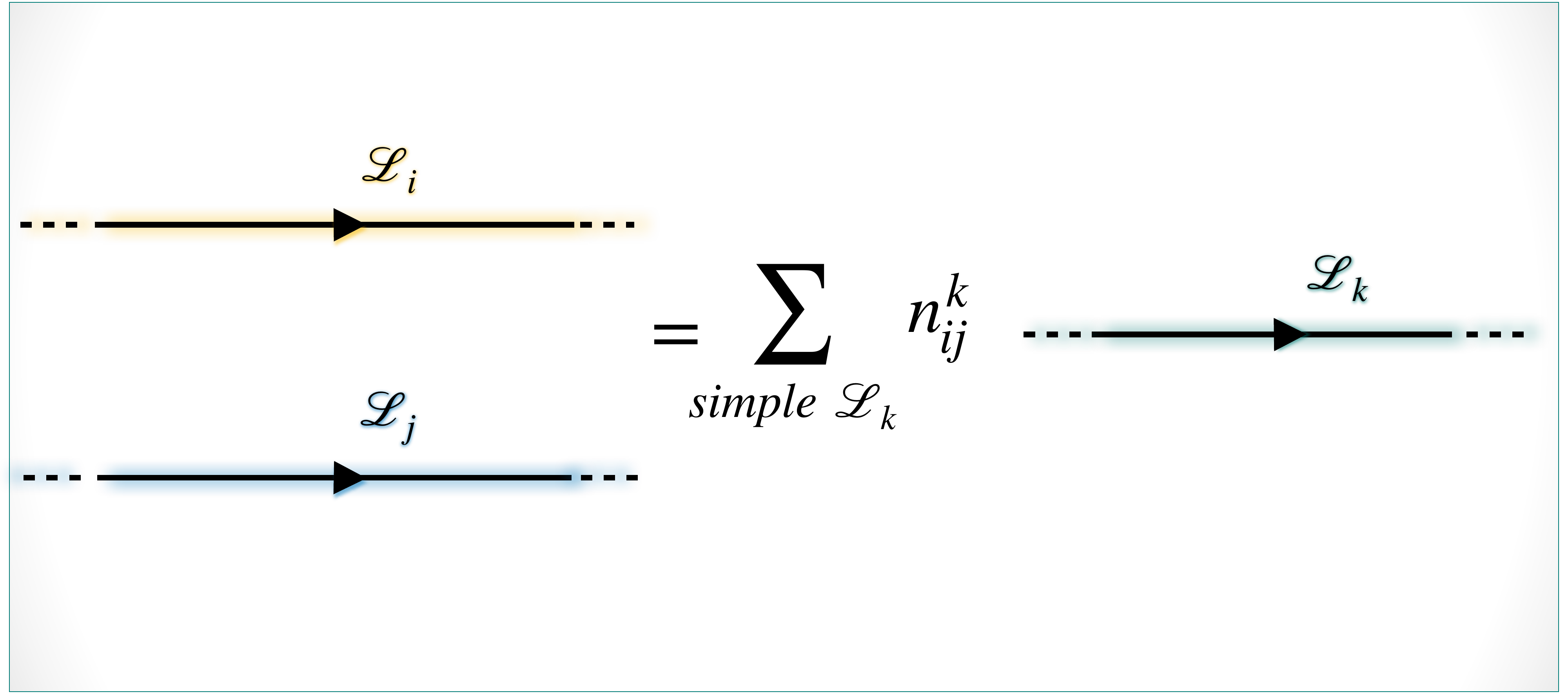}
        \caption{\small{Decomposition of the fusion product of two simple defects $\mathcal{L}_i$, $\mathcal{L}_j$ in terms of simple defects with non-negative integral coefficients.}}
        \label{fig4}
    \end{figure}
    \item For each (not necessarily simple) $\CL$, we have
    \be \CL\CL^*=n\CI +\ldots,
    \ee where $\ldots$ is a sum over simple defects not isomorphic to the identity, and $n\ge 1$, with $n=1$ if and only if $\CL$ is simple. The defect $\CL$ is called invertible if and only if $\CL\CL^*=\CI$. Note that $\CL$ is invertible if and only if its quantum dimension  $\langle \CL\rangle=1$. The set of all invertible topological defects forms a group with respect to tensor products, and it can be identified with the group of global symmetries preserving the chiral algebra of the CFT.
\end{itemize}

We will be interested in the case where all fields of the CFT are holomorphic, so that the  CFT itself can be identified with its (maximally extended) chiral algebra, while the anti-chiral algebra is trivial. In this case, the CFT can be described mathematically as a (unitary) vertex operator algebra (VOA) $V$ where the vector space $V$ is identified with the Hilbert space as $V\equiv \Hh$. The state-operator correspondence assigns to each state $v\in V$ a vertex operator $Y(v,z):=\sum_{n\in \ZZ} v(n)z^{-n-1}$, where $v(n)\in {\rm End}_\CC(V)$. In order to describe a well-defined CFT, $V$ must be a holomorphic VOA, i.e. there must be a  unique (up to isomorphism)  irreducible representation of $V$, which is just $V$ itself.  

One can provide a rigorous mathematical definition for the fusion category of topological defects in $\Hh\cong V$ preserving a subVOA $\mathcal{A}\subset V$ that is strongly rational (i.e. rational, $C_2$-cofinite, and of CFT type; see \cite{Moller:2024xtt} for more details). In this case, it is known that the category $\mathrm{Rep}(\CA)$ of representations of $\CA$ is a modular tensor category \cite{Huang1,Huang2}. The VOA $V$ itself can be seen as a representation over the subalgebra $\CA$, so that it corresponds to an object $V \in {\rm Rep}(\CA)$. Furthermore, the OPE between fields in $V$ determines a morphism $m\in \Hom_{{\rm Rep}(A)}(V\otimes V,V)$ satisfying a number of consistency conditions; this morphism makes $V$ an \emph{algebra object} in ${\rm Rep}(\CA)$. Now, the physical intuition suggests that for every topological defect $\CL$ preserving $\CA$ there should be an $\CL$-twisted sector $V_\CL$ that is a representation of $\CA$. Therefore, every $\CL$ is associated with an object $V_\CL\in {\rm Rep}(\CA)$. Furthermore, one expects a (possibly non-local) OPE between any vertex operator in $V$ and any operator in $V_\CL$ that produces another operator in $V_\CL$. This is formalized by the existence of a morphism $m_\CL\in \Hom_{{\rm Rep}(\CA)}(V\otimes V_{\CL},V_{\CL})$ satisfying suitable compatibility conditions with $m:V\otimes V\to V$;  the object $V_\CL\in {\rm Rep}(\CA)$, equipped with the morphism $m_\CL$,  is called a (non-local) module for the algebra object $V$. Thus, we can identify the category of topological defects preserving $\CA$ in $V$ as the category of modules for the algebra object $V$ in ${\rm Rep}(\CA)$; this is indeed a fusion category satisfying the properties that are expected from physics. See \cite{Moller:2024xtt} and \cite{Volpato:2024goy} for more details and discussions.

The best studied example, discussed in more detail in section \ref{s:orbifolds}, is the case where $\CA=V^G$ is the subalgebra of $V$ that is fixed by a finite group $G\subset \Aut(V)$ of automorphisms of $V$. Assuming that $V^G$ is strongly rational\footnote{This has been proved when $G$ is an abelian or, more generally, a solvable finite group \cite{Carnahan:2016guf}.}, the category of topological defects of $V$ preserving $V^G$ contains only the invertible defects $\CL_g$, $g\in G$, as simple defects, with fusion given by the group law $\CL_g\otimes \CL_h\cong \CL_{gh}$. In this case, the objects $V_g\equiv V_{\CL_g}$ in the abstract category of modules for the algebra object $V\in {\rm Rep}(V^G)$ have a  concrete description as irreducible admissible $g$-twisted  $V$-modules. The latter are defined as $L_0$-graded vector spaces $V_g=\oplus_{r\in \lambda+\frac{1}{N}\ZZ} V_g(r)$, for some $\lambda\in \QQ_{\ge 0}$, equipped with a linear map $V\ni v\mapsto Y_g(v,z)$, where the vertex operator $Y_g(v,z)$ is a $z$-dependent endomorphism of $V_g$ satisfying a `$g$-twisted analogue' of the standard VOA axioms (see for example \cite{DongLepowsky,DLM1,DLM3,DLM4,DLM5}). Physically, $Y_g(v,z)$ represent the fields of the original theory $V$ when acting on $g$-twisted sector states. As far as we know, no analogous mathematically rigorous definition of $\CL$-twisted vertex operators $V\ni v\mapsto  Y_\CL(v,z)$ as $z$-dependent endomorphisms of $V_\CL$ have been provided for generic (possibly non-invertible) topological defects $\CL$.

For a given (finite) group $G$, the inequivalent fusion categories generated by simple defects $\CL_g$, $g\in G$, with $G$-group fusion rules are in one-to-one correspondence with classes $[\alpha]\in H^3(G,U(1))$ in the third cohomology group of $G$ \cite{Dixon:1986jc,Narain:1986qm,Dijkgraaf:1989pz,Roche:1990hs}. 
Physically, $[\alpha]$ is interpreted as the 't Hooft anomaly, i.e. the obstruction to gauging (i.e. taking the orbifold) the global symmetry $G$. For each representative $\alpha$ of $[\alpha]$ the corresponding fusion category is usually denoted by $Vec^\alpha_G$, and is the category of $G$-graded finite dimensional vector spaces, with an associator isomorphism $(V_g\otimes V_h)\otimes V_k\stackrel{\cong}{\to} V_g\otimes (V_h\otimes V_k)$ `twisted' by $\alpha(g,h,k)\in U(1)$. 

In this article, we are interested in the case where the preserved subalgebra $\CA$ is not necessarily rational. This case is much less understood from a mathematical viewpoint; we will assume that most of the properties expected from physics still hold for the defects we are interested in.

\subsection{Topological defects in supersymmetric VOAs}\label{s:TDLsuperVOA}

In this section we extend the construction from the previous section to holomorphic fermionic VOAs (or super vertex operator algebras, SVOAs). In these settings, the space of states $V$ (NS sector) is $\ZZ_2$-graded by the fermion number $(-1)^F$, which acts on the states with $L_0$-eigenvalue $\frac{r}{2}+\ZZ$, $r\in \ZZ/2\ZZ$ by multiplication by $(-1)^r$. Furthermore, there is a Ramond sector $V_{tw}$, i.e. a $(-1)^F$-twisted module for $V$; in all holomorphic SVOAs that we consider, one can choose a $\ZZ_2$-grading by $(-1)^F$ on $V_{tw}$ compatible with the one on $V$. 

Consider a category of topological defects in $V$ preserving a rational subSVOA $\CA\subset V$; for simplicity let us here assume that $\CA$ is strongly rational, where we say that an SVOA is rational if its even (bosonic)  subalgebra is.
In our cases of interest, $\CA$ always contains the $\CN=1$ superVirasoro algebra, and thus the odd part is non-zero.

In this case, every topological defect $\CL$ commuting with $\CA$ is associated with an operator $\hat\CL: V\to V$ preserving the $L_0$ eigenspaces and therefore the $(-1)^F$-grading. We also want to associate with $\CL$ a linear operator $\hat\CL:V_{tw}\to V_{tw}$ acting on the Ramond sector, as well as the $\CL$-twisted and the $(-1)^F\CL$-twisted sectors $V_{\CL}$ and $V_{\CL,tw}$ that are, respectively, untwisted and $(-1)^F$-twisted representations for the superalgebra $\CA$. However, defining these concepts in a way that is compatible with the $(-1)^F$ grading can sometimes present complications, namely:
\begin{enumerate}
    \item Even if $\hat\CL:V\to V$ commutes with $(-1)^F$, it might not be possible to define $\hat\CL:V_{tw}\to V_{tw}$ such that it commutes with the $\ZZ_2$-grading by $(-1)^F$ on the Ramond sector.
    \item The irreducible $\CL$-twisted and $(-1)^F\CL$-twisted sectors $V_\CL$ and $V_{\CL,tw}$ might not admit a well-defined $\ZZ_2$-grading compatible with the one on $\CA$. Suppose, for example, that $(-1)^F$ is not well-defined on $V_\CL$. What happens in this case is that one can define two irreducible $\CL$-twisted modules $V_{\CL}$ and $V'_{\CL}$, that are the same as vector spaces, but with different vertex operators $Y_{V_\CL}(v,z)$ and $Y_{V'_\CL}(v,z)$, $v\in \CA$ acting on them, with the relation $Y_{V'_\CL}(v,z)=(-1)^{F}Y_{V_\CL}(v,z)(-1)^{F}$. This means that $(-1)^F$ is well-defined only on the direct sum $V_{\CL}\oplus V'_{\CL}$, which is not irreducible as an $\CL$-twisted module.
\end{enumerate}
See \cite{Chang:2022hud,Runkel:2022fzi,Grigoletto:2021zyv,Bhardwaj:2024ydc} for more details.
An example of the first issue can be found for an invertible $\ZZ_2$ defect in the SVOA $V$ generated by two  holomorphic spin-$1/2$ free Majorana fermions $\psi^1(z),\psi^2(z)$, that can be combined into a complex fermion $\chi=\frac{1}{\sqrt{2}}(\psi^1+i\psi^2)$ and $\chi^*=\frac{1}{\sqrt{2}}(\psi^1-i\psi^2)$. The Ramond sector of the theory contains two ground states $|\pm\rangle$, satisfying $\chi_0|-\rangle=0$ and $\chi^*_0|-\rangle=|+\rangle$, such that $\chi_0^*|+\rangle=0$. There are two consistent definitions of $(-1)^F$ as an involution on the Ramond sector, acting either by $(-1)^F|\pm\rangle =\pm |\pm \rangle$ or with the opposite signs on the ground states. Consider the symmetry $g$ acting by $g(\psi^1)=\psi^1$ and $g(\psi^2)=-\psi^2$ on the fermions. Clearly, $g$ commutes with $(-1)^F$ on the NS sector. However, because $g(\chi^*)=\chi$, the action of $g$ on the Ramond ground states must exchange $|+\rangle$ and $|-\rangle$, and therefore it cannot commute with $(-1)^F$.

An example of the second issue arises for the Ramond sector of the SVOA $V$ of a single free Majorana fermion $\psi(z)$ of spin $1/2$. We would like to define the Ramond space $V_{tw}$ to be an irreducible highest weight representation for the algebra of integrally moded fermionic modes $\psi_r$, $r\in \ZZ$, obeying the usual algebra $\{\psi_r,\psi_s\}=2\delta_{r,-s}$. In such a module, there is a single Ramond ground state that is an eigenvector of the zero mode $\psi_0$ and is annihilated by all $\psi_r$ with $r>0$. Because $(\psi_0)^2=1$, there are actually two different irreducible $(-1)^F$-twisted representations $V_{tw}$ and $V'_{tw}$ for the SVOA $V$, one with a ground state $|+\rangle$ such that $\psi_0|+\rangle=+|+\rangle$ and one with ground state $|-\rangle$ such that $\psi_0|-\rangle=-|-\rangle$. In either of these modules, one cannot assign a well-defined fermion number to the single ground state $|+\rangle$ or  $|-\rangle$ in a way compatible with the operators $\psi_r$ being odd. If one insists on having a well-defined fermion number in the Ramond sector, then one needs to take the direct sum $V_{tw}\oplus V'_{tw}$, so that there are two ground states $|-\rangle$ and $|+\rangle$, and define $(-1)^F|\pm \rangle=|\mp\rangle$. In this case, however, the Ramond sector is not an irreducible module for the algebra $\psi_r$, $r\in \ZZ$. A consequence of this problem is that the naive modularity relations do not hold. In particular, if we define the NS $\ZZ_2$-graded partition function in the usual way \be
Z_{\NS}^-(\tau):=\Tr_V(q^{L_0-\frac{c}{24}}(-1)^F)\ ,
\ee then the S-transformation yields the Ramond partition function only up to factors $\sqrt{2}$: 
\be Z_{\NS}^-(-1/\tau)=\sqrt{2}\Tr_{V_{tw}}(q^{L_0-\frac{c}{24}})=\sqrt{2}\Tr_{V'_{tw}}(q^{L_0-\frac{c}{24}})=\frac{1}{\sqrt{2}}\Tr_{V_{tw}\oplus V'_{tw}}(q^{L_0-\frac{c}{24}})\ .
\ee
As discussed, for example, in \cite{DongNgRen:2021} or \cite{DongRenYang:2022}, this phenomenon can arise in general for invertible defects preserving a strongly rational subSVOA.

For a group $G$ of invertible defects, the interplay between the group action and the fermion number can be described in terms of an anomaly, represented by a class in the supercohomology group $SH^3(G,U(1))$, that generalizes the group $H^3(G,U(1))$ of 't Hooft anomalies in the bosonic case \cite{Gu:2012ib,Kapustin:2014dxa,Gaiotto:2015zta,Brumfiel:2016vpy,Gaiotto:2017zba}\footnote{These supercohomology classes admit a more geometrical interpretation as elements of $ {\rm Hom}(\Omega^{spin}_3(BG),U(1))\cong SH^3(G,U(1))$, i.e. as  $U(1)$-valued homomorphisms defined from the group $\Omega^{spin}_3(BG)$ of spin-bordism classes $[(M,f)]$ \cite{Gaiotto:2015zta,Brumfiel:2016vpy,Gaiotto:2017zba}. Here, $M$ is a $3$-manifold with spin structure, and $f:M\to BG$ is a continuous map to the universal classifying space of $G$. This isomorphism fits with the general understanding of anomalies as classes in suitable cobordism groups.}. An element in $SH^3(G,U(1))$ can be described as a triple $(\omega,\mu,\alpha)$, where $\omega:G\to \ZZ_2$ (the so-called Majorana layer) is a $\ZZ_2$-valued $1$-cocycle (i.e. a homomorphism), $\mu:G\times G\to \ZZ_2$ (Gu-Wen layer) a $\ZZ_2$-valued $2$-cochain, and $\alpha:G\times G\times G\to U(1)$ (bosonic or 't Hooft layer) a $U(1)$-valued $3$-cochain. In particular, $\omega(g)$ is $+1$ or $-1$ depending on whether the $\CL_g$ twisted sector admits or not a well-defined fermion number (the latter case corresponds to issue (2) above). The Gu-Wen layer $\mu(g,h)$ indicates whether the isomorphism $V_g\otimes V_h\stackrel{\cong}{\longrightarrow} V_{gh}$ is even or odd. Finally, $\alpha$ is analogous to the 't Hooft anomaly in bosonic theories. In this article, we will only consider cases where the Majorana layer is trivial. In this situation, $\mu$ must be a $2$-cocyle (i.e. it satisfies $\mu(g,h)\mu(gh,k)=\mu(g,hk)\mu(h,k)$), and therefore it defines a $\ZZ_2$ central extension $\tilde G=\ZZ_2.G$ of the group $G$. Physically, $G$ is the group acting faithfully on the NS sector, while $\tilde G$ has a well-defined action on both the NS and Ramond sector; in particular, $\tilde G$ is the extension of $G$ by the order $2$ symmetry acting by $-1$ in the Ramond sector. The $2$-cocycle $\mu$ and the $3$-cochain $ \alpha$ can be combined into a $3$-cocycle $\tilde\alpha$ for the group $\tilde G$. If $G$ includes the fermion number as a central element, then in order to avoid issue (1), we also require $\omega$ and $\alpha$ to be trivial whenever one of their arguments is $(-1)^F$. In practice, when the Majorana layer is trivial, and there are no (mixed) anomalies involving the fermion number, one can safely ignore the complications due to the fermions, and work as if one had a bosonic VOA with symmetry group $\tilde G$ and 't Hooft anomaly $\tilde\alpha$. An important example of non-trivial Gu-Wen layer is given by $V=V^{f\natural}$ itself, where $G\cong Co_1$ acts faithfully on the NS sector, while the group with a well-defined (non-projective) action on the Ramond sector is the double cover $\tilde G\cong \ZZ_1.Co_1\cong Co_0$ \cite{Johnson-Freyd:2017ble}.

This analysis generalizes to the case of a category including non-invertible simple defects \cite{Chang:2022hud,Runkel:2020zgg,Runkel:2022fzi}, even though, as far as we know, the various possibilities cannot be represented by different classes in some (super-)cohomology. The analogue of the Majorana layer is the distinction between simple defects of $m$-type or of $q$-type, i.e. such that $\CL$-twisted sector admits or not a well-defined fermion number \cite{Gaiotto:2015zta,Aasen:2017ubm}. Notice that the fusion product of two $m$-type defects is still $m$-type; this is analogous to the cocycle condition for $\alpha$. The analogue of the Gu-Wen layer is the possibility that some topological $3$-junction operators $V_{\CL_i}\otimes V_{\CL_j}\to V_{\CL_k}$ are odd, i.e. they invert the fermion number. In analogy with the group-like case, when all defects are of $m$-type, one can avoid introducing such odd junctions by simply replacing them by even junctions $V_{\CL_i}\otimes V_{\CL_j}\to V_{\xi\CL_k}$, where $ V_{\xi\CL_k}$ is the same as $V_{\CL_k}$ but with reversed fermion number. Compatibility with modular transformations implies that the defect $\xi\CL_k$ is just the fusion of $\CL_k$ with the invertible order $2$ defect $\xi$, acting by $-1$ on the Ramond sector and trivially on the NS sector. Considering $\CL$ and $\xi\CL$ as distinct defects is analogous to working with the double cover $\tilde G$ of $G$.

In the next section, for simplicity, we will  restrict ourselves to categories of topological defects where neither issue (1) nor issue (2) occurs; in particular, all topological defects are of $m$-type and the fusion matrices involving the fermion number are trivial.

\subsection{Topological defects in $V^{f\natural}$ and main theorem}
 \label{subsec:3.2}

Let us focus on topological defects for a particular holomorphic SVOA $V$ with central charge $c\in \frac{1}{2}\NN$ that commute with a given $\CN=1$ superconformal algebra and are well-behaved with respect to the fermion number $(-1)^F$. More specifically,  consider defects $\CL$ satisfying the following conditions:
\begin{enumerate}
	\item The linear map $\hat\CL:V\to V$ commutes with $(-1)^F$ and with the action of the $\CN=1$ superVirasoro algebra. The same holds for the induced linear map $\hat\CL:V_{tw}\to V_{tw}$.
	\item Both the $\CL$-twisted sector $V_\CL$ and the $(-1)^F\CL$-twisted sector $V_{\CL,tw}$ are unitary representations of, respectively, the NS and Ramond $\CN=1$ superVirasoro algebras, with a well-defined and compatible $\ZZ_2$-grading by $(-1)^F$. Furthermore, the spectrum of $L_0$ is discrete,\footnote{It may be possible to reproduce the results of this article for defects obeying weaker conditions, for example by allowing $L_0$ to have continuous spectrum (maybe with a gap) on the $\CL$-twisted sector; see section \ref{s:discussion} for further comments. In this article we will not consider these more general possibilities.} and each $L_0$-eigenspace is a finite dimensional $\ZZ_2$-graded vector space.  \end{enumerate}
   In the language of section \ref{s:TDLsuperVOA}, this implies that $\CL$ is of $m$-type. The spaces $V_\CL $ and $V_{\CL,tw}$ are called the NS and Ramond $\CL$-twisted sectors, respectively. Note that unitarity implies that the conformal weights (i.e. $L_0$-eigenvalues) obey $h\in \RR_{\ge 0}$ and $h\in \RR_{\ge \frac{c}{24}}$ in the NS and R $\CL$-twisted sectors, respectively, so that we have a decomposition
   \be V_\CL=\oplus_{h\in\RR_{\ge 0}} V_{\CL}(h)\ ,\qquad  V_{\CL,tw}=\oplus_{h \in\RR_{\ge c/24}} V_{\CL,tw}(h)\ .\ee
It is  useful to observe that a topological defect $\CL$ of $V$ obeying properties 1 and 2 also defines a topological defect in the SVOA $V'=V/\langle (-1)^F\rangle$ obtained by taking the orbifold of $V$ by the fermion number symmetry $(-1)^F$ (assuming that such an orbifold is a consistent SVOA). In general, however, the induced defect $\CL$ might not preserve any $\CN=1$ supercurrent in the orbifold SVOA $V'$. In particular, when $V=V^{f\natural}$, this means that $\CL$ defines also a topological defect in the SVOA $F(24)$ of $24$ free fermions. As discussed in section \ref{s:Vfnat_def}, the two theories have the same even subVOA, and are obtained from each other by exchanging the odd NS sector of one theory with the even Ramond sector of the other. 
   
   If properties 1 and 2 hold, then we have a well-defined  $\CL$-twined partition function $\CZ^\CL:=(Z^{\CL,+}_{\NS}, Z^{\CL,-}_{\NS}, Z^{\CL,+}_{\R}, Z^{\CL,-}_{\R})^t$ with components
 \begin{align}\label{eq:DefectTwineNS}
 	Z^{\CL,\pm}_{\NS}(\tau)&:=\Tr_{V}(q^{L_0-\frac{c}{24}}(\pm 1)^F\hat\CL)=\sum_{h\in \frac{1}{2}\ZZ} q^{h-\frac{c}{24}}\Tr_{V(h)}((\pm 1)^F\hat\CL)\\\label{eq:DefectTwineR}
 	Z^{\CL,\pm}_{\R}(\tau)&:=\Tr_{V_{tw}}(q^{L_0-\frac{c}{24}}(\pm 1)^F\hat\CL)=\sum_{h\in \frac{1}{2}+\ZZ} q^{h-\frac{c}{24}}\Tr_{V_{tw}(h)}((\pm 1)^F\hat\CL)\ ,
 \end{align}
 and an $\CL$-twisted partition function $\CZ_\CL:=(Z^+_{\CL,\NS}, Z^-_{\CL,\NS}, Z^+_{\CL,\R}, Z^-_{\CL,\R})^t$ with components
 \begin{align}\label{eq:DefectTwistNS}
 	Z^\pm_{\CL,\NS}(\tau)&:=\Tr_{V_\CL}(q^{L_0-\frac{c}{24}}(\pm 1)^F)=\sum_h q^{h-\frac{c}{24}}\Tr_{V_\CL(h)}((\pm 1)^F)\\\label{eq:DefectTwistR}
 	Z^\pm_{\CL,\R}(\tau)&:=\Tr_{V_{\CL,tw}}(q^{L_0-\frac{c}{24}}(\pm 1)^F)=\sum_h q^{h-\frac{c}{24}}\Tr_{V_{\CL,tw}(h)}((\pm 1)^F)\ .
 \end{align} We require the following:
\begin{enumerate}\setcounter{enumi}{2}
	\item $\CZ^\CL$ and $\CZ_\CL$ transform into one another under modular S-transformations
\be\label{Stransftwin} \CZ_\CL(-1/\tau)=\rho(S)\CZ^\CL(\tau)
\ee where $\rho(S)$ is the same as in eq.\eqref{modularspin}. 
\end{enumerate}
By property 2, $V_{\CL,tw}$ carries a unitary representation of the Ramond $\CN=1$ superVirasoro algebra with discrete $L_0$-spectrum. Therefore, an argument similar to the one given for the Witten index $Z_\R^-$ implies that $Z^{\CL,-}_{\R}$ and $Z^-_{\CL,\R}$ must be independent of $\tau$,
\be Z^{\CL,-}_{\R}=\Tr_{V_{tw}(c/24)}((-1)^F\hat\CL)\ ,\qquad Z^{-}_{\CL,\R}=\Tr_{V_{\CL,tw}(c/24)}((-1)^F)\ .
\ee Furthermore, by \eqref{Stransftwin}, they must be equal: $Z^{\CL,-}_{R}=Z^-_{\CL,R}$. It follows that 
\be \Tr_{V_{tw}(c/24)}((-1)^F\hat\CL)=\Tr_{V_{\CL,tw}(c/24)}((-1)^F)\in \ZZ\ ,\ee
and the latter is clearly an integral number.

Let $G\subseteq \Aut_\tau(V)$ be a group of symmetries preserving the $\CN=1$ supercurrent $\tau(z)$, and such that all invertible topological defects $\CL_g$, $g\in G$, satisfy the properties 1, 2, and 3 above. Let $\cTop$ be a tensor category of topological defects satisfying 1, 2, 3, and suppose that $\cTop$ contains all invertible defects $\CL_g$, $g\in G$.
Then any $\CL\in \cTop$ must satisfy
\be\label{Cardy} \Tr_{V_{tw}(c/24)}((-1)^F\hat\CL_g\hat\CL )\in \ZZ\ ,\qquad \forall g\in G\ .
\ee

The application of the constraint \eqref{Cardy} to the case $V=V^{f\natural}$, leads us to the central theorem of this article, that we describe now. Let $\CC^{24}$ be a $24$-dimensional Hilbert space, and $\Lambda\subset \CC^{24}$ an embedded copy of the Leech lattice, i.e. a free $\ZZ$-module of rank $24$ such that $\Lambda\otimes \CC\cong \CC^{24}$ and $\Lambda$ is isometric to the Leech lattice. Then the ring  ${\rm End}(\Lambda)$ of $\ZZ$-linear maps from $\Lambda$ to itself is naturally a subring of ${\rm End_\CC}(\mathbb{C}^{24})$, the rings of $\CC$-linear maps from $\CC^{24}\cong \Lambda\otimes \CC$ to itself. Furthermore, $\Lambda\otimes\RR\cong \RR^{24}\subset \CC^{24}$ is a $24$-dimensional real Euclidean space, one has ${\rm End}(\Lambda)\subset {\rm End}_\RR(\RR^{24})\subset {\rm End_\CC}(\mathbb{C}^{24})$, and hermitian conjugation on ${\rm End_\CC}(\mathbb{C}^{24})$ restricts to transposition on ${\rm End}_\RR(\RR^{24})$.
\begin{theorem}\label{th:main}
    Let $\cTop$ be a tensor category of topological defects $\mathcal{L}$ of $V^{f\natural}$ containing all invertible defects $\mathcal{L}_g$, $g\in Co_0$, and such that all objects $\CL\in \cTop$ satisfy properties 1, 2 and 3 above.  Then there is an embedding $\Lambda\hookrightarrow {}^\RR V^{f\natural}_{tw}(1/2)$ of the Leech lattice $\Lambda$ into ${}^\RR V^{f\natural}_{tw}(1/2)\subset V^{f\natural}_{tw}(1/2)\cong \mathbb{C}^{24}$, such that for every $\mathcal{L}\in \cTop$, the $\CC$-linear map $\hat\CL_{\rvert V^{f\natural}_{tw}(1/2)}:V^{f\natural}_{tw}(1/2)\to V^{f\natural}_{tw}(1/2)$ is contained in ${\rm End}(\Lambda)$. The assignment \begin{align}
       \rho:\cTop &\to {\rm End}(\Lambda)\ ,\\ \CL &\mapsto \rho(\CL):=\hat\CL_{\rvert V^{f\natural}_{tw}(1/2)} 
    \end{align} defines a surjective, non-injective ring homomorphism from the Grothendieck ring $Gr(\cTop)$ to ${\rm End}(\Lambda)$, such that $\rho(\CL^*)=\rho(\CL)^\dag=\rho(\CL)^t$.
\end{theorem}

See section \ref{s:proof} for a proof of this theorem. A note of caution is necessary about the hypotheses of this theorem. We are not able to prove that if two defects $\CL$ and $\CL'$ satisfy conditions 1, 2, and 3, then their fusion $\CL\CL'$ satisfies such properties as well -- in fact, this might not be true in general, especially when the subVOA preserved by $\CL$, $\CL'$ is non-rational.\footnote{For example, the condition that the $\CL$-twisted sector has discrete $L_0$-spectrum is quite delicate if $\CL$ does not preserve a rational algebra, and in general it seems unlikely that such a property closes under fusion of defects.} If $\CL\CL'$ fails one of these three conditions, then a tensor category that includes both $\CL$ and $\CL'$ would not satisfy the hypotheses of Theorem \ref{th:main}.   In particular, a single tensor category containing \emph{all} topological defects satisfying properties 1, 2, and 3, and only such defects, might not exist. 
This observation affects the discussion in the following sections, where we will describe some examples of topological defects $\CL$ in $V^{f\natural}$ that satisfies  properties 1, 2, and 3, and cannot be written as superposition of invertible defects. We will (somehow implicitly) assume that all such defects are contained in a large tensor category $\cTop$ satisfying the properties of theorem \ref{th:main}. While it seems plausible to us that such a category exists, at least for the examples that we discuss explicitly, we warn the reader that we haven't rigorously proved this fact.

We also stress that $\rho$ in theorem \ref{th:main} is just a homomorphism of rings and not of fusion (or, more generally, based) rings; in particular, there is no distinguished basis of ${\rm End}(\Lambda)$ to which the (isomorphism classes of) simple objects of $\cTop$ are mapped.

The fact that the homomorphism is not injective implies that one cannot reconstruct the Grothendieck ring $Gr(\cTop)$ from ${\rm End}(\Lambda)$. However, one can use this result to put restrictions on the properties of the topological defects in $\cTop$. For example:
 \begin{enumerate}[label=(\alph*)]
    \item The theorem establishes that, for a given choice of the $\CN=1$ supercurrent $\tau(z)$ in $V^{f\natural}$, there is a distinguished lattice $\Lambda\subset {}^\RR V^{f\natural}_{tw}(1/2)$ of Ramond ground states. We have the following:
  \begin{corollary}\label{th:integralqdim}  If $\CL\in \cTop$ of quantum dimension $\langle \CL\rangle$ preserves a field $\lambda\in\Lambda\subset {}^\RR V^{f\natural}_{tw}(1/2)$, i.e. if $\hat\CL$ acts by $\hat\CL|\lambda\rangle=\langle\CL\rangle |\lambda\rangle$, then necessarily $\langle \CL\rangle\in \ZZ_{\ge 1}$. More generally, if $\CL$ preserves any field in $V^{f\natural}_{tw}(1/2)$ (not necessarily in $\Lambda$), then the quantum dimension $\langle \CL\rangle$ must be an algebraic integer of degree at most $24$. 
  \end{corollary} The reason is that the linear operator $\hat\CL_{\rvert V^{f\natural}_{tw}(1/2)}$ is represented, with respect to the basis of the lattice $\Lambda$, by a $24\times 24$ matrix with integral entries, so that its minimal polynomial is a monic polynomial with integral coefficients and degree at most $24$. Thus, if the quantum dimension is one of its eigenvalues, it must also be a root of that polynomial. 
    \item Let $\bar\QQ$ denote the algebraic closure of $\QQ$, i.e. the field containing all roots of univariate polynomials with rational coefficients (algebraic numbers). Then the $\bar\QQ$-vector space $\Lambda\otimes \bar\QQ\subset \Lambda\otimes \CC$ is a countable subset of $\Lambda\otimes \CC\cong\CC^{24}$, and in particular it is zero measure  with respect to Lebesgue measure on $\CC^{24}$. We have the following: 
\begin{corollary}\label{identity} Suppose $\hat\CL$ acts trivially  on some $\psi\in V^{f\natural}_{tw}(1/2)$ (i.e. $\hat\CL|\psi\rangle=\langle\CL\rangle|\psi\rangle$) such that $\psi^\perp \cap (\Lambda\otimes \bar\QQ)=0$. Then, $\CL$ is multiple of the identity defect, i.e. $\CL=d\CI$, $d\in \NN$ (and in particular $\langle \CL\rangle=d$ is integral).
\end{corollary}
Given $\bar\lambda\in \Lambda\otimes\bar\QQ$, let $\bar\lambda^\perp\subset\Lambda\otimes\CC$ be its orthogonal complement. The condition $\psi^\perp \cap (\Lambda\otimes \bar\QQ)=0$ is equivalent to
\be \psi\notin \bigcup_{0\neq \bar\lambda \in \Lambda\otimes\bar\QQ} \bar\lambda^\perp\ .
\ee The set on the right is a countable union of codimension $1$ subspaces in $\Lambda\otimes\CC$, so the condition on $\psi$  
is the generic situation: for a randomly chosen vector $\psi\in \Lambda\otimes \CC$, the probability that it is orthogonal to some $0\neq \bar\lambda \in \Lambda\otimes\bar\QQ$ is zero (with respect to Lebesgue measure on $\CC^{24}$). 
\begin{proof}[Proof of corollary \ref{identity}.]
 The assumption implies that $\langle\CL\rangle\in \bar\QQ$ is an eigenvalue of $\hat\CL$ acting on  $V^{f\natural}_{tw}(1/2)=\Lambda\otimes \CC$, so that the corresponding eigenspace $\Sigma\subset \Lambda\otimes \CC$ has dimension $1\le \dim_\CC\Sigma\le 24$, and contains $\psi$. Because both $\Sigma$ and its orthogonal complement $\Sigma^\perp\subset \Lambda\otimes \CC$ are determined by linear equations with algebraic coefficients, they are spanned by vectors in $\Lambda\otimes\bar\QQ$. But because no vector in $\Lambda\otimes \bar\QQ$ is orthogonal to $\psi\in \Sigma$, it follows that $\Sigma^\perp=0$, so that $\dim\Sigma=24$ and $\hat\CL$ preserves all Ramond fields in $V^{f\natural}_{tw}(1/2)$. 
 A topological defect that preserves all Ramond ground fields and the supercurrent of $V^{f\natural}$ must be proportional to the identity defect. 
 This is most easily seen by considering the induced defect $\CL$ on the theory of $24$ free fermions $F(24)=V^{f\natural}/\langle (-1)^F\rangle$. The Ramond ground fields of $V^{f\natural}$ become spin $1/2$ fields in the NS sector of $F(24)$, and they generate the whole SVOA. On the other hand, the supercurrent $\tau(z)$ of $V^{f\natural}$ is in the Ramond sector of $F(24)$. Thus, a defect that preserves all $24$ free fermions of $F(24)$ and a field in the Ramond sector must preserve all the NS and Ramond sector of $F(24)$, and therefore must be proportional to the identity defect. But then the same conclusion must hold for the defect $\CL$ acting on $V^{f\natural}$, since the two theories are related by exchanging the odd NS sector with the Ramond even sector.
 
\end{proof}
\end{enumerate}

In general, if a topological defect preserves a Ramond ground field $\psi(z)\in V^{f\natural}_{tw}(1/2)$, then  it also preserves a weight $2$ NS field $\normord{\psi\partial\psi}(z)\in V^{f\natural}$, whose modes generate a copy of the $c=1/2$ Virasoro algebra $Vir_{c=1/2}$. If $W\subset V^{f\natural}$ is a proper subSVOA of $V^{f\natural}$, then one expects that the category of topological defects preserving $W$ is non-trivial, i.e. it is not just generated by the identity defect. 
Thus, proposition \ref{identity} strongly suggests that when the preserved field $\psi\in V^{f\natural}_{tw}(1/2)=\Lambda\otimes\CC$ satisfies $\psi^\perp\cap (\Lambda\otimes\bar\QQ)=0$, then the only SVOA containing both $\tau(z)$ and $\normord{\psi\partial\psi}(z)$ is $V^{f\natural}$ itself, i.e. that these two fields generate the whole $V^{f\natural}$.

An alternative explanation would be that all non-trivial topological defects preserving both $\tau(z)$ and $\normord{\psi\partial\psi}(z)$ fail to satisfy properties 1, 2 and 3 above, so that Theorem \ref{th:main} does not apply.

\subsection{The connection with topological defects in K3 non-linear sigma models}\label{s:topdefK3}

Theorem \ref{th:main} and Corollaries \ref{th:integralqdim} and \ref{identity} are formally very similar to the properties of topological defects in K3 NLSMs that were proved in \cite{Angius:2024evd} (see in particular section 3.2, properties (a), (b), (c) and Claims 1, 2, and 3), and that we recall here. Let $\calC$ be a K3 NLSM, and let $\cTop^{K3}_\calC$ be the category of topological defects of $\calC$ preserving the $\CN=(4,4)$ superconformal algebra and the spectral flow. It was shown in \cite{Angius:2024evd} that the restriction of $\hat\CL$ to the $24$-dimensional space $\Hh^{K3}_{\HRR,gr}$ of RR ground states of $\calC$ defines an endomorphism of the lattice $\Gamma^{4,20}\subset \Hh^{K3}_{\HRR,gr}$ of D-brane charges, which is an even unimodular lattice of signature $(4,20)$. This is of course very reminiscent of Theorem \ref{th:main} in the present article, with the Leech lattice $\Lambda$ replaced by $\Gamma^{4,20}$. On the other hand, for all $\CL\in \cTop^{K3}_\calC$, the lattice endomorphism $\hat\CL$ is also required to preserve a $4$-dimensional positive definite subspace $\Pi\subset \Gamma^{4,20}\otimes \RR$, which depends on the particular NLSM K3 $\calC$. Physically, $\Pi$ is identified with the subspace of spectral flow generators in  $\Hh^{K3}_{\HRR,gr}$.\footnote{It should be emphasized that we always use $\Pn\subset {}^\RR V^{f\natural}_{tw}(1/2)$ to denote a four--plane associated with $V^{f\natural}$, whereas $\Pi\subset \Gamma^{4,20}\otimes \RR$ is always a four--plane associated with a K3 NLSM.}

The analogy with the results of \cite{Angius:2024evd} for K3 sigma models---as well as the established  connections between four--plane--preserving symmetry groups and twining functions \cite{Gaberdiel:2011fg,Duncan:2015xoa,Cheng:2016org, Paquette:2017gmb} of the two theories---suggests the following construction in $V^{f\natural}$.
Let $\Pi^\natural\subset {}^\RR V^{f\natural}_{tw}(1/2)$ be a $4$-dimensional real subspace 
and define $\cTop_\Pn$ as the full subcategory of $\cTop$ whose objects $\CL\in \cTop_\Pn$ satisfy $\hat\CL_{\rvert \Pi^\natural}=\langle \CL\rangle {\rm id}_\Pn$:
\be\label{TopPi} \mathrm{Obj}(\cTop_\Pn)=\{\CL\in  \mathrm{Obj}(\cTop)\mid \hat\CL_{\rvert \Pi^\natural}=\langle \CL\rangle {\rm id}_{\Pi^\natural}\}\ .
\ee
This implies that every $\CL\in \cTop_\Pn$ commutes with all Ramond ground fields in $\Pi^\natural$, and therefore with the subalgebra $\widehat{so}(4)_1 = \widehat{su}(2)_1 \oplus \widehat{su}(2)_1$ of $\widehat{so}(24)_1$. As described in section \ref{s:VfnatK3conn}, choose a generator $J_0^3$ of $\widehat{su}(2)_1$ in  a $c=6$, $\CN=4$ superconformal algebra. Then, for each $\CL\in \cTop_\Pn$, let us define the $\CL$-twined genus
\be\label{twiningVf} \phi^\CL(V^{f\natural},\tau,z):=\Tr_{V^{f\natural}_{tw}}(\hat\CL q^{L_0-\frac{c}{24}}y^{J_0^3}(-1)^F)\ ,
\ee which is the obvious generalization of \eqref{eq:D-MCtwining} to the case of (possibly) non-invertible defects (see section \ref{s:Conwaytwining}).

It is tempting to conjecture that the connection between invertible symmetries in $V^{f\natural}$ and K3 models described in section \ref{s:VfnatK3conn} extends to the topological defects in $\cTop_\Pn$. In particular, 
given a K3 model $\calC$, for each $\CL\in \cTop^{K3}_{\calC}$, one can define the $\CL$-twined genus
\be\label{K3twining} \phi^{\CL}(\calC,\tau,z):=\Tr_{\Hh^{K3}_{\HRR}}(\hat\CL q^{L_0-\frac{c}{24}}{\bar q}^{\bar L_0-\frac{\bar c}{24}}y^{J_0^3}(-1)^{F+\bar F})\ee
which is the analogue of \eqref{eq:gTwineK3}. Then, one can expect the $\CL$-twined genera $\phi^\CL(V^{f\natural})$ to match with the $\CL$-twined genera $\phi^\CL(\calC)$ in some suitable K3 model $\calC$. 

More precisely, we propose the following conjecture:
\begin{conjecture}\label{conj:K3relation}
    Let $\Pn\subset \Lambda\otimes\RR\subset V^{f\natural}_{tw}$ be a subspace of dimension $4$ and $\cTop_\Pn$ the full subcategory of $\cTop$ defined in \eqref{TopPi}. Then, there exists a non-linear sigma model $\calC$ on K3 such that:\begin{enumerate}[label=(\roman*)]
        \item there in an equivalence $F:\cTop_\Pn\longrightarrow \cTop^{K3}_\calC$ of tensor categories between $\cTop_\Pn$ and the   category $\cTop^{K3}_\calC$ of topological defects of $\calC$ preserving $\CN=(4,4)$ and spectral flow;
        \item for every $\CL\in \cTop_\Pn$, the $\CL$-twining genus $\phi^\CL(V^{f\natural},\tau,z)$  computed in $V^{f\natural}$ (eq.\eqref{twiningVf}) coincides with the $F(\CL)$-twining genus in the K3 model $\calC$ (eq.\eqref{K3twining}):
 \be        \phi^\CL(V^{f\natural},\tau,z)=\phi^{F(\CL)}(\calC,\tau,z).
        \ee 
\item there is an isometry of Hilbert spaces $\varphi:V^{f\natural}_{tw}(1/2)\stackrel{\cong}{\longrightarrow} \Hh^{K3}_{\HRR,gr}$ such that for each $\CL\in \cTop_\Pn$
\be \varphi\circ \hat\CL_{\rvert V^{f\natural}_{tw}(1/2)} =\widehat{F(\CL)}_{\rvert \Hh^{K3}_{\HRR,gr}}\circ \varphi\ .
\ee
    \end{enumerate}
\end{conjecture}\noindent
This conjecture proposes an answer to Questions \ref{q1},  \ref{q2}, and \ref{q3} from section \ref{s:intro}.
We emphasize that even if, on both sides, we restrict to the subcategories of topological defects generated by invertible defects $\CL_g$, $g\in G_\Pn$, Conjecture \ref{conj:K3relation} is slightly stronger than  previous conjectures in \cite{Duncan:2015xoa,Cheng:2016org, Paquette:2017gmb}, discussed in section \ref{s:VfnatK3conn}. Indeed, the claim that the `group-like' categories of invertible defects on the two sides are equivalent, does not only imply that the underlying groups of symmetries are isomorphic, but also that their 't Hooft anomalies are represented by the same cohomology class $[\alpha]\in H^3(G_\Pn,U(1))$ in both CFTs. The fact that the twining genera match exactly (see \eqref{K3Vfnatmathc}) provides strong evidence that the anomalies are indeed the same, but as far as we know it is not sufficient to prove the anomaly matching for all possible subgroups $G_\Pn\subset Co_0$.
In section \ref{s:K3matching}, we will provide some more evidence in favor of Conjecture \ref{conj:K3relation}, by showing examples of categories non-invertible defects in $V^{f\natural}$ and K3 NLSMs constructed from torus orbifolds where the matching is verified.

From a more conceptual point of view, it would be desirable to show  at least that the conjecture is compatible with what is currently known about topological defects on the two sides of the correspondence. In particular, one should prove that, for each choice of $\Pi^\natural\subset {}^{\RR} V^{f\natural}_{tw}(1/2)$, there is a choice of a vector space isomorphism $\varphi:{}^\RR V^{f\natural}_{tw}(1/2)\stackrel{\cong}{\longrightarrow} \Gamma^{4,20}\otimes \RR$ such that
\begin{enumerate}
    \item The space $\Pi:=\varphi(\Pi^\natural)$ is a positive definite $4$-dimensional subspace of $\Gamma^{4,20}\otimes \RR$ such that $\Pi$ is not orthogonal to any root $\delta\in \Gamma^{4,20}$, $\delta^2=-2$. This ensures that there is a NLSM $\calC_\Pi$ on K3 such that $\Gamma^{4,20}\otimes \CC\cong \Hh^{K3}_{\HRR,gr}$, $\Gamma^{4,20}$ is the lattice of D-brane charges, and $\Pi$ is the subspace of spectral flow generators.
    \item $\varphi:V^{f\natural}_{tw}\stackrel{\cong}{\longrightarrow} \Hh^{K3}_{\HRR,gr}$ is an isometry between Hilbert spaces, where we identified $\Gamma^{4,20}\otimes \CC\cong \Hh^{K3}_{\HRR,gr}$ and extended $\varphi$ by linearity.
    \item For every $\CL\in \cTop_{\Pi^{\natural}}$, the map $\varphi\circ \rho(\CL)\circ \varphi^{-1}$ is an endomorphism of the lattice $\Gamma^{4,20}$.
\end{enumerate}
These properties ensure that all $\varphi\circ \rho(\CL)\circ \varphi^{-1}$ satisfy the expected consistency conditions for the action of a topological defects in $\cTop^{K3}_{\calC_\Pi}$ on the space $\Hh^{K3}_{\HRR,gr}$.

In appendix \ref{a:conjcons}, we will show that such a map $\varphi$ can always be defined when the category $\cTop_{\Pi^\natural}$ contains only defects with integral quantum dimension. As shown in Corollary \ref{th:integralqdim}, this happens for example when $\Pi^\natural$ contains a non-zero vector in the lattice $\Lambda$.
It would be desirable to extend this argument to the case where $\cTop_\Pn$ contains defects with non-integral quantum dimension. 

As mentioned in the introduction, in \cite{Cheng:2016org} it was argued that the twining genera of certain invertible symmetries of K3 models are not reproduced by any (invertible) twining genus in $V^{f\natural}$. In fact, the arguments of \cite{Cheng:2016org} really show that the 't Hooft anomaly of the group of symmetries of certain K3 models is not the same as the anomaly for the isomorphic subgroups of $Co_0$ in $V^{f\natural}$. This means that the corresponding fusion categories of invertible defects, while having isomorphic fusion ring, are not equivalent.  We stress that this result does not contradict conjecture \ref{conj:K3relation}, that states that every suitable fusion category of defects in $V^{f\natural}$ admits a corresponding category in some K3 model. It is rather an obstruction to the inverse conjecture: in a sense, there are `too many' fusion categories $\cTop_\calC^{K3}$ of defects in K3 models, and not all of them can be reproduced by some $\cTop_{\Pi^\natural}$ in $V^{f\natural}$. This argument provides a negative answer to question $1$ in the introduction, at least in its strongest form. The weaker version of this question, concerning just the fusion rings structures, is still open.

\subsection{Proof of theorem \ref{th:main}}\label{s:proof}

In order to prove Theorem \ref{th:main}, we will use an explicit description of the Golay code $\CG_{24}$ and of the Leech lattice $\Lambda$, suitable for our purposes.  We refer to \cite{ConwaySloane} for proofs and more details.

The (extended binary) Golay code $\CG_{24}$ is a $12$-dimensional subspace of the vector space $\FF_2^{24}$ over the finite field $\FF_2=\{0,1\}$ with two elements. Let us denote every element of $\FF_2^{24}$ by a sequence $u\equiv (u_1,\ldots,u_{24})$ with $u_i=0$ or $1$, and let the weight of $u$ be given by 
$$ wt(u):=\sum_{i=1}^{24}u_i\in \NN$$ i.e. by the number of non-zero entries in $(u_1,\ldots, u_{24})$. Then, the  $2^{12}$ elements (words) of $\CG_{24}$ have all weight $wt(u)\in 4\ZZ$ (i.e. $\CG_{24}$ is doubly even), and the shortest non-zero elements have weight $8$. Furthermore, $\CG_{24}$ is self-dual,  $\CG_{24}=\CG_{24}^\perp$, i.e. it coincides with the dual space
$$ \CG_{24}^\perp := \{ (u_1,\ldots,u_{24})\in \FF_2^{24}\mid \forall u'\in \CG_{24},\ \sum_{i=1}^{24} u_iu_i'\in 2\ZZ\}. 
$$ It can be proved that there is a unique subspace of $\FF_2^{24}$ satisfying these properties, up to permutations in $S_{24}$. We choose $\CG_{24}\subset \FF_2^{24}$ to be the subspace generated by
\begin{align*}
    (1 , 1 , 1 , 1 , 1 , 1 , 1 , 1 , 0 , 0 , 0 , 0 , 0 , 0 , 0 , 0 , 0 , 0 , 0 , 0 , 0 , 0 , 0 , 0)\ ,
   \\
( 1 , 1 , 1 , 1 , 0 , 0 , 0 , 0 , 1 , 1 , 1 , 1 , 0 , 0 , 0 , 0 , 0 , 0 , 0 , 0 , 0 , 0 , 0 , 0)\ ,
   \\
 (1 , 1 , 0 , 0 , 1 , 1 , 0 , 0 , 1 , 1 , 0 , 0 , 1 , 1 , 0 , 0 , 0 , 0 , 0 , 0 , 0 , 0 , 0 , 0)\ ,
   \\
 (1 , 0 , 1 , 0 , 1 , 0 , 1 , 0 , 1 , 0 , 1 , 0 , 1 , 0 , 1 , 0 , 0 , 0 , 0 , 0 , 0 , 0 , 0 , 0)\ ,
   \\
 (1 , 0 , 0 , 1 , 1 , 0 , 0 , 1 , 1 , 0 , 0 , 1 , 1 , 0 , 0 , 1 , 0 , 0 , 0 , 0 , 0 , 0 , 0 , 0)\ ,
   \\
 (1 , 0 , 1 , 0 , 1 , 0 , 0 , 1 , 1 , 1 , 0 , 0 , 0 , 0 , 0 , 0 , 1 , 1 , 0 , 0 , 0 , 0 , 0 , 0)\ ,
   \\
 (1 , 0 , 0 , 1 , 1 , 1 , 0 , 0 , 1 , 0 , 1 , 0 , 0 , 0 , 0 , 0 , 1 , 0 , 1 , 0 , 0 , 0 , 0 , 0)\ ,
   \\
 (1 , 1 , 0 , 0 , 1 , 0 , 1 , 0 , 1 , 0 , 0 , 1 , 0 , 0 , 0 , 0 , 1 , 0 , 0 , 1 , 0 , 0 , 0 , 0)\ ,
   \\
 (0 , 1 , 1 , 1 , 1 , 0 , 0 , 0 , 1 , 0 , 0 , 0 , 1 , 0 , 0 , 0 , 1 , 0 , 0 , 0 , 1 , 0 , 0 , 0)\ ,
   \\
 (0 , 0 , 0 , 0 , 0 , 0 , 0 , 0 , 1 , 1 , 0 , 0 , 1 , 1 , 0 , 0 , 1 , 1 , 0 , 0 , 1 , 1 , 0 , 0)\ ,
   \\
 (0 , 0 , 0 , 0 , 0 , 0 , 0 , 0 , 1 , 0 , 1 , 0 , 1 , 0 , 1 , 0 , 1 , 0 , 1 , 0 , 1 , 0 , 1 , 0)\ ,
   \\
 (0 , 0 , 0 , 0 , 0 , 0 , 0 , 0 , 0 , 0 , 0 , 0 , 0 , 0 , 0 , 0 , 1 , 1 , 1 , 1 , 1 , 1 , 1 , 1)\ .
\end{align*} For any (possibly empty) subset $S\subseteq \{1,\ldots,24\}$, let $u_S:=(u_1,\ldots,u_{24})\in \FF_2^{24}$ be the element with non-zero entries in positions $i\in S$, i.e
\be u_i=\begin{cases}
    1 & \text{if }i\in S\\
    0 & \text{if }i\notin S\ ,
\end{cases}
\ee so that $wt(u_S)=|S|$. A $\CG_{24}$-set is a subset $S\subseteq \{1,\ldots, 24\}$ such that the word $u_S$ is in $\CG_{24}\subset \FF_2^{24}$.  Besides the null vector $(0,\ldots,0)\in \CG_{24}$ of weight $wt(0)=0$, $\CG_{24}$ contains $759$ vectors of weight $8$, $1771$ of weight $12$, $759$ of weight $16$, and the vector $(1,\ldots,1)$ of weight $24$. A $\CG_{24}$-set $S$ is called an (special) octad if $|S|=8$, a (umbral) dodecad if $|S|=12$, and an octad complement if $|S|=16$.

Following \cite{ConwaySloane} (Ch.10, Section 3.2, Th. 25), the Leech lattice can be defined as the subset $\Lambda\subset \RR^{24}$ containing all vectors of the form
$$ \frac{1}{\sqrt{8}}(x_1,\ldots, x_{24})^t\in \RR^{24}\ ,
$$ where $x_1,\ldots, x_{24}$ satisfy:
\begin{enumerate}
    \item $x_1,\ldots, x_{24}$ are all integral, and they are either all odd or all even (i.e. $x_i\equiv x_j\mod 2$);
     \item $\sum_i x_i\in 8\ZZ$ if all $x_i$ are even, and $\sum_i x_i\in 4+8\ZZ$ if all $x_i$ are odd;
    \item for each $m\in \ZZ/4\ZZ$, the set of $i\in\{1,\ldots, 24\}$ for which $x_i\equiv m\mod 4$ is a $\CG_{24}$-set, where $\CG_{24}$ is the (extended) binary Golay code.
\end{enumerate}
For example, the following shapes of $(x_1,\ldots,x_{24})$ are allowed (notation: $(a^{m_a};b^{m_b};...)$ means that the entry $a$ is repeated $m_a$ times, etc):
\begin{itemize}
    \item $(\pm 4^2;0^{22})$, with non zero entries in any position and any choice of signs;
    \item $(\pm 2^{8};0^{16})$ with non-zero entries on a $\CG_{24}$-set (a Golay octad), and even number of minus signs;
    \item $(-3^1;1^{23})$ with $-3$ in any position.
\end{itemize}
The vectors listed above all have squared length $4$. We are particularly interested in the vectors of squared length $8$, that can be of the form
\begin{itemize}
    \item $(\pm 8^1; 0^{23})$, with the $\pm 8$ in any position;
    \item $(\pm 4^4;0^{20})$, with non zero entries in any position and any choice of signs;
    \item $(\pm 2^{16};0^{8})$ with non-zero entries on a $\CG_{24}$-set (the complement of a Golay octad), and even number of minus signs;
    \item $(-3^5;1^{19})$ with $-3$ in any position.
\end{itemize}
The group ${\rm Aut}(\Lambda)$ of automorphisms of $\Lambda$ is a subgroup of $SO(\Lambda\otimes\RR)=SO(24)$ isomorphic to the Conway group ${\rm Aut}(\Lambda)\cong Co_0$.

Let ${}^\RR V_{tw}^{f\natural}(1/2)\subset V_{tw}^{f\natural}(1/2)$ be the real $24$-dimensional space of Ramond ground states that are CPT self-conjugate, i.e. fixed points of the antilinear involution $\theta$. The hermitian product on $V^{f\natural}_{tw}$ restricts to a positive definite bilinear product on ${}^\RR V_{tw}^{f\natural}(1/2)$ which is given by the $2$-point functions. 

Let $\Aut_\tau(V^{f\natural})\cong Co_0$ be the group  of automorphisms of $V^{f\natural}$ preserving a fixed $\CN=1$ supercurrent $\tau(z)$. This group commutes with the antilinear CPT involution $\theta$ and preserves the bilinear form, so it acts (faithfully) on ${}^\RR V^{f\natural}_{tw}$ by orthogonal transformations $\Aut_\tau(V^{f\natural})\subset SO({}^\RR V^{f\natural}_{tw}(1/2))$.

Because the real $24$-dimensional representation of $Co_0$ is unique up to isomorphism, there exists an isometric embedding $\Lambda\hookrightarrow {}^\RR V_{tw}^{f\natural}(1/2)\subset V_{tw}^{f\natural}(1/2)$ such that the action of ${\rm Aut}(\Lambda)\subset SO(\Lambda\otimes\RR)$ coincides with $\Aut_\tau(V^{f\natural})\subset SO({}^\RR V^{f\natural}_{tw})$. Given a choice of this embedding, in the following we will identify ${}^\RR V_{tw}^{f\natural}(1/2)\cong \Lambda\otimes\RR$ and $V_{tw}^{f\natural}(1/2)\cong \Lambda\otimes\CC$.

Each topological defect $\CL\in \cTop$ is associated with a $\CC$-linear map $\rho(\CL):=\hat\CL_{\rvert V^{f\natural}_{tw}(1/2)}$ in ${\rm End}_\CC(V^{f\natural}_{tw}(1/2))$. In particular, for invertible defects $g\equiv \CL_g$, the map $g\mapsto \rho(g)$ provides the representation of the Conway group $Co_0$ on the $24$-dimensional space $V^{f\natural}_{tw}(1/2)\cong \Lambda\otimes\CC\cong\CC^{24}$.

In this section, we will prove the following:
\begin{theorem}\label{th:LeechEndo}
    Let $\hat\CL:\Lambda\otimes\CC \to \Lambda\otimes\CC$ be a $\CC$-linear map on $\Lambda\otimes\CC\cong\CC^{24}$. Then, the following are equivalent:
    \begin{enumerate}
        \item[(a)] $ \Tr_{\CC^{24}}(\hat\CL \rho(g))\in \ZZ\ \qquad \forall g\in {\rm Aut}(\Lambda)\cong Co_0\ ,
    $
    \item[(b)] $\hat\CL(\Lambda)\subseteq \Lambda$,
    \item[(c)] $\hat\CL=\sum_{g\in {\rm Aut}(\Lambda)} n(g)\rho(g)$ for some $n(g)\in \ZZ$.
    \end{enumerate} 
\end{theorem} 
Before proving this theorem, let us show how it implies Theorem \ref{th:main}. We know that the map $\CL\mapsto \rho(\CL)$ provides a ring homomorphism from $Gr(\cTop)$ to ${\rm End}_\CC(V^{f\natural}_{tw})$ such that $\rho(\CL^*)=\theta \rho(\CL)^\dag\theta$. On the other hand, all defects $\CL\in \cTop$ must satisfy eq.\eqref{Cardy}, which is essentially condition (a) of Theorem \ref{th:LeechEndo} (the operator $(-1)^F$ is just the identity on  $V^{f\natural}_{tw}(1/2)$). Thus, Theorem \ref{th:LeechEndo} implies that the image $\rho(\CL)$ also satisfies condition (b), i.e. $\rho(\CL)\in {\rm End}_\ZZ(\Lambda)\subset {\rm End}_\RR(\Lambda\otimes \RR)$. Furthermore, because on $\Lambda\otimes \RR\cong {}^\RR V^{f\natural}_{tw}(1/2)$ the involution $\theta$ is the identity and the adjoint reduces to the transpose on ${\rm End}_\RR(\Lambda\otimes \RR)\subset {\rm End}_\CC(\Lambda\otimes \CC)$, we have $\rho(\CL^*)=\rho(\CL)^t$. The equivalence of (b) and (c) in Theorem \ref{th:LeechEndo} shows that the ring ${\rm End}_\ZZ(\Lambda)$ is spanned by integral linear combinations of the form $\sum_{g\in {\rm Aut}(\Lambda)} n(g)\rho(g)$, so that the homomorphism $\rho$ is already surjective if we restrict to the subcategory of $\cTop$ generated by invertible defects. Finally, the homomorphism is obviously not injective: for example, if $\mathfrak{k}$ is the non-trivial central element of $\Aut_\tau(V^{f\natural})\cong Co_0$, then $\CI+\CL_\mathfrak{k}\in \cTop$ is mapped to $\rho(\CI)+\rho(\CL_\mathfrak{k})=0$.

It is useful to re-formulate Theorem \ref{th:LeechEndo} using the vector space isomorphisms
$$ {\rm End}(\RR^{24})\cong \RR^{24}\otimes (\RR^{24})^*\ ,\qquad  {\rm End}(\CC^{24}) \cong\CC^{24}\otimes (\CC^{24})^*
$$ as well as the $\ZZ$-module isomorphism
\be {\rm End}_\ZZ(\Lambda)\cong \Lambda\otimes \Lambda^*\subset \RR^{24}\otimes(\RR^{24})^*\ ,
\ee
where the last inclusion follows from $\Lambda\subset \RR^{24}\subset \CC^{24}$. There are positive definite hermitian form $\langle\cdot,\cdot\rangle$ on $\CC^{24}\otimes (\CC^{24})^*$ and bilinear form on $(\cdot,\cdot)$ on $\RR^{24}\otimes (\RR^{24})^*$ induced by the ones on $\Lambda\otimes \RR$ and $\Lambda\otimes\CC$,  and given, respectively, by
\be \langle A,B\rangle =\Tr_{\CC^{24}}(A^\dag B)\ ,\qquad (A,B)=\Tr_{\RR^{24}}(A^t B)\ .
\ee In particular, for any $\lambda,\mu\in \Lambda=\Lambda^*$, one has
\be (\lambda\otimes\mu,\lambda\otimes\mu)=(\lambda\cdot\lambda)(\mu\cdot\mu)\in 2\ZZ_{\ge 0}\ ,
\ee which shows that $\Lambda\otimes \Lambda^*$ is itself a positive definite $24^2$-dimensional even self-dual lattice.\footnote{Self-duality follows because if $\lambda_1\ldots,\lambda_{24}$ and $\lambda_1^*,\ldots,\lambda^*_{24}$ are dual bases of $\Lambda$, i.e. $\lambda_i\cdot\lambda_j^*=\delta_{ij}$, then $\{\lambda_i\otimes\lambda_j^*\}$ and $\{\lambda_i^*\otimes \lambda_j\}$ are dual bases of $\Lambda\otimes\Lambda^*$.} Each $g\in\Aut(\Lambda)$ defines a vector $\rho(g)\in {\rm End}_\ZZ(\Lambda)\cong \Lambda\otimes\Lambda^*$ and one can consider the sublattice generated by integral linear combinations all such vectors
\be \rho(\ZZ\Aut(\Lambda)):=\{\sum_{g\in {\rm Aut}(\Lambda)} n(g)\rho(g),\ n(g)\in \ZZ\}\subseteq \Lambda\otimes\Lambda^*\ ,
\ee
which is the image by $\rho$ of the group ring $\ZZ\Aut(\Lambda)$. Because $\Lambda\otimes\Lambda^*$ is self-dual, we have a chain of inclusions
\be\label{oneinside} \rho(\ZZ\Aut(\Lambda))\subseteq \Lambda\otimes\Lambda^* \subseteq \rho(\ZZ\Aut(\Lambda))^*\ ,
\ee where 
\be \rho(\ZZ\Aut(\Lambda))^*:=\{A\in \CC^{24}\otimes(\CC^{24})^*\mid  \langle B,A\rangle \in \ZZ, \forall B\in \rho(\ZZ\Aut(\Lambda)) \}\ee is the dual of the lattice $\rho(\ZZ\Aut(\Lambda))$. Thus, it is clear that the conditions (a), (b), (c) of Theorem \ref{th:LeechEndo} can be equivalently reformulated as, respectively, conditions (a$'$), (b$'$), and (c$'$) in the following theorem:

\medskip
\noindent {\bf Theorem  \ref{th:LeechEndo}$'$.}
   \emph{Let $\hat\CL\in \CC^{24}\otimes (\CC^{24})^*$. The following are equivalent: 
   \begin{enumerate}
        \item[(a$'$)] $\hat\CL\in \rho(\ZZ\Aut(\Lambda))^*$;
    \item[(b$'$)] $\hat\CL\in \Lambda\otimes \Lambda^*$,
    \item[(c$'$)] $\hat\CL\in \rho(\ZZ\Aut(\Lambda))$.
    \end{enumerate} }
 In particular, the obvious implications (b)$\Rightarrow$(a) and (c)$\Rightarrow$(b) just correspond to the inclusions in eq.\ref{oneinside}. On the other hand, proving (a)$\Rightarrow$(b) amounts to showing the equality
\be \rho(\ZZ\Aut(\Lambda))^*=\Lambda\otimes \Lambda^*\ .
\ee  But if this is true, then because of self-duality of $\Lambda\otimes \Lambda^*$, one also obtains
\be \rho(\ZZ\Aut(\Lambda))=\Lambda\otimes \Lambda^*\ ,
\ee which amounts to the implication (b)$\Rightarrow$(c). Therefore, it is sufficient to prove the implication (a)$\Rightarrow$(b) of Theorem \ref{th:LeechEndo}, and all the other results follow immediately.

Let 
\be \Lambda(2n):=\{\lambda\in\Lambda\mid \lambda^2=2n\}\ee denote the set of Leech lattice vectors of length-squared $2n$. 
The following lemma is easy to prove.
\begin{lemma}\label{Lambda8gener} The lattice $\Lambda$ is generated by the vectors in $\Lambda(8)$.
\end{lemma}
\begin{proof} The lattice $\Lambda$ can be written as the disjoint union
    \be \Lambda_0 \cup (v+\Lambda_0)
    \ee where $\Lambda_0$ is the sublattice \be\label{Lambda0}\Lambda_0:=\{\frac{1}{\sqrt{8}}(x_1,\ldots,x_{24})\in\Lambda\mid x_i\in 2\ZZ\}\subset \Lambda\ ,\ee while $v=\frac{1}{\sqrt{8}}(y_1,\ldots,y_{24})\in \Lambda$ is any vector with all odd $y_i$. 
    It is shown in \cite{ConwaySloane}, Chapter 10, that  $\Lambda_0$ is generated by all the vectors $\lambda_C\in \Lambda$ of the form 
    \be\label{lambdaC}\lambda_C=\frac{2}{\sqrt{8}}\sum_{i\in C} e_i\ ,\ee where $C\subseteq\{1,\ldots, 24\}$ is a $\CG_{24}$-set.
 Recall that the cardinality $|C|$ of a $\CG_{24}$-set can be $|C|=0$ (empty set), $|C|=8$ (octads), $|C|=12$ (dodecads), $|C|=16$ (octads complements), or $|C|=24$ (the full set $\Omega=\{1,\ldots, 24\}$). In particular, when $|C|=16$, one has $\lambda_C\in \Lambda(8)$. Furthermore, any octad $O$ or dodecad $D$ can be obtained as the symmetric difference $C\triangle C'$ of suitable $\CG_{24}$-sets $C,C'$ of cardinality $|C|=|C'|=16$, with $|C\cap C'|=12$ or $|C\cap C'|=10$, respectively. Thus, one can obtain the element $\lambda_O$ or $\lambda_D$ as the sum $\lambda_C-\tilde\lambda_{C'}$, where
 \be \tilde\lambda_{C'}=\frac{2}{\sqrt{8}}(\sum_{i\in C\cap C'} e_i - \sum_{i\in C'\setminus (C\cap C')} e_i)\in \Lambda(8)\ .
 \ee Furthermore, the vector $\frac{1}{\sqrt{8}}(2,\ldots, 2)$ can be obtained from the sum $\lambda_O+\lambda_{C}$ for $O$ an octad and $C$ its complement $\{1,\ldots, 24\}\setminus O$. Thus, the lattice generated by vectors in $\Lambda(8)$ contains $\Lambda_0$. Finally, the vector $v$ can be chosen to be of length $8$, for example
     \be v=\frac{1}{\sqrt{8}}(-3,-3,-3,-3,-3,1,\ldots,1)\ ,
    \ee so that also $v+\Lambda_0$ is in the lattice generated by $\Lambda(8)$.
\end{proof}

The group $\Aut(\Lambda)$ contains a subgroup isomorphic to $\ZZ_2^{12}$ and generated by diagonal matrices $\epsilon_C$, where $C$ is a $\CG_{24}$-set and
\be \epsilon_C={\rm diag}(\nu_1,\ldots,\nu_{24})\ ,\qquad \nu_i=\begin{cases}
    -1 & \text{for }i\in C\\
    1 & \text{for }i\notin C\ .
\end{cases}
\ee Note  that, given a Leech lattice vector of the form $\lambda_C\equiv  ((\lambda_C)_1,\ldots,(\lambda_C)_{24})^t=\frac{2}{\sqrt{8}}\sum_{i\in C}e_i$, where $C$ is a $\CG_{24}$-set, then one has
\be \mathbf{1}_{24}-\epsilon_C=\diag(\sqrt{8}(\lambda_C)_1,\ldots,\sqrt{8}(\lambda_C)_{24})\ .
\ee
Recall that the vectors of the form $\lambda_C$, for $C$ a $\CG_{24}$-set, generate the sublattice $\Lambda_0\subset \Lambda$ in \eqref{Lambda0}. In particular, since this sublattice contains the vector $(\sqrt{8},0,0,\ldots,0)\in \Lambda_0$, this means that there is a suitable integral linear combination $\sum_C n_C \epsilon_C$ of elements $\epsilon_C\in \ZZ_2^{12}\subset \Aut(\Lambda)$ such that
\be\label{Amatr} \sum_C n_C \epsilon_C= {\rm diag}(8,0,\ldots,0)\ ,\qquad\qquad \text{for some }n_C\in\ZZ\ .
\ee

It is also known (see \cite{ConwaySloane}, Thm. 27, Chapter 10) that $Co_0$ acts transitively on the sets $\Lambda(4)$, $\Lambda(6)$, and $\Lambda(8)$. This means that, given any pair of vectors $x,y\in \Lambda$, $x\neq y$, with same length $x^2=y^2\le 8$, there exists $g\in Co_0$ such that $g(x)=y$. In particular, for any $\lambda\in\Lambda(8)$, there is an element $g_\lambda\in Co_0$ that maps the vector
\be (\sqrt{8},0,0,\ldots,0)^t\in \Lambda(8)\ ,
\ee to $\lambda$. Thus, for any $\lambda,\mu\in\Lambda(8)$, the group elements $g_\lambda$ and $g_\mu^{-1}$ are represented by matrices of the form
\be\label{gilambdagimu} g_\lambda=\frac{1}{\sqrt{8}}\begin{pmatrix}
    \lambda_1 & \ldots\\
    \lambda_2 & \ldots\\
    \vdots & \ddots\\
    \lambda_{24} & \ldots
\end{pmatrix}\qquad g^{-1}_\mu=(g_\mu)^t=\frac{1}{\sqrt{8}}\begin{pmatrix}
    \mu_1 & \mu_2 &\ldots &\mu_{24}\\
    \vdots & \vdots & \ddots &\vdots
\end{pmatrix}
\ee i.e. the first column of $g_\lambda$ is given by components of $\frac{1}{\sqrt{8}}\lambda\in \RR^{24}$ and the first row of $g_\mu^{-1}$ is given by the components of $\frac{1}{\sqrt{8}}\mu\in \RR^{24}$.

We are now ready to prove that (a)$\Rightarrow$(b) in Theorem \ref{th:LeechEndo}.
Suppose that $\hat\CL:\CC^{24}\to \CC^{24}$ is a $\CC$-linear map such that $\Tr(g\hat\CL)\in \ZZ$ for all $g\in \Aut(\Lambda)\cong Co_0$. 
For any pair of vectors $\lambda,\mu\in\Lambda(8)$, consider the $Co_0$ elements $g_\lambda$ and $g^{-1}_\mu$ in \eqref{gilambdagimu} and the integral linear combination $\sum_C n_C \epsilon_C$ of elements $\epsilon_C\in \Aut(\Lambda)$ in eq.\eqref{Amatr}. Then, we have
\be \ZZ\ni \sum_Cn_C  \Tr(g_\lambda\epsilon_C g^{-1}_\mu\hat\CL ) =\Tr(\sum_Cn_C\epsilon_C g^{-1}_\mu\hat\CL g_\lambda)= \mu\cdot \hat\CL(\lambda)\ ,\ee
 so that 
 \be\label{integr} \mu\cdot \hat\CL(\lambda)\in \ZZ\ ,
 \ee for all $\lambda,\mu\in \Lambda(8)$. Because, by lemma \ref{Lambda8gener}, every vector in $\Lambda$ can be written as an integral linear combination of vectors in $\Lambda(8)$, eq.\eqref{integr} must be satisfied for all $\lambda,\mu\in \Lambda$.  This implies that, for all $\lambda\in\Lambda$, $\hat\CL(\lambda)$ is in the dual $\Lambda^*$ of the Leech lattice
 \be \hat\CL(\lambda)\in \Lambda^*=\Lambda\ ,\qquad \forall\lambda\in\Lambda
 \ee where in the last step we used that $\Lambda$ is self-dual.

\section{Topological defects in orbifold CFTs}\label{s:orbifolds}

In section \ref{section:TDLs_Vfnat}, we showed that with every topological defect $\CL\in \cTop$ of $V^{f\natural}$ satisfying properties 1, 2, and 3 of section \ref{subsec:3.2}, is associated a linear operator $\rho(\CL)=\hat\CL_{\rvert V^{f\natural}(1/2)}$ that is an endomorphism of a copy of the Leech lattice $\Lambda\subset {}^\RR V^{f\natural}(1/2)\subset V^{f\natural}(1/2)$. Furthermore, the map $\rho:\cTop\to {\rm End}(\Lambda)$ is already surjective when we restrict to the subcategory of topological defects generated by the invertible ones, see section \ref{s:proof}. Based on these observations, it is natural to wonder whether the category $\cTop$ contains any non-trivial simple defect other than the invertible ones.  In this section, we will show that there is a large number of simple non-invertible TDLs satisfying properties 1, 2, and 3 of section \ref{subsec:3.2}. We will construct such defects using techniques from non-abelian orbifolds of vertex operator superalgebras. We stress that the methods described in this section are certainly not new (see for example \cite{Frohlich:2009gb,Bhardwaj:2017xup}); here, we simply provide a description of these ideas that is well-suited for applications to $V^{f\natural}$ and K3 sigma models. 

As we will show, all defects obtained in this way have integral quantum dimension. However, $\cTop$ also contains simple defects with non-integral (and, in fact, irrational) quantum dimension; some examples of dimension $\sqrt{2}$ and $\sqrt{3}$ will be discussed in a companion paper \cite{Angius:2025xxx}.

\subsection{VOAs with finite groups of symmetries}

In this section, we summarize some of the main properties of non-abelian orbifolds of VOAs; we refer to \cite{Coste:2000tq,Bhardwaj:2017xup} for more information.

Given a finite group $G$, we will use the following notation:
\begin{itemize}
	\item $[g]_G$ a $G$-conjugacy class with representative $g\in G$. We write simply $[g]$ when the group is clear.
	\item The centralizer of $g\in G$ in $G$:
	$$C_G(g):=\{h\in G\mid hg=gh\}\ .$$
	\item $Irr(G)$ is the set of (isomorphism classes of) irreducible representations of $G$. 
\end{itemize}
We will often use the orthogonality and completeness relation for finite group characters:
\be\label{orth1} \frac{1}{|C_G(h)|}\sum_{\rho\in Irr(G)}\Tr_\rho(g)\Tr_\rho(h)^*=\delta_{[g],[h]}\ ,\qquad g,h\in G\ .\ee
\be\label{orth2} \frac{1}{|G|}\sum_{g\in G}\Tr_\rho(g)\Tr_{\rho'}(g)^*=\delta_{\rho,\rho'}\ ,\qquad \rho,\rho'\in Irr(G)\ .\ee

Consider a holomorphic bosonic VOA $V$, such that $V$ is the only irreducible $V$-module up to isomorphism, and furnishes a modular invariant, chiral CFT. Suppose that $V$ admits a finite symmetry group of automorphisms $G\subseteq Aut(V)$. We will assume, for all groups that we consider, that the fixed point subalgebra $V^G$ is rational, $C_2$-cofinite, unitary, and has positive conformal weight. Then, as discussed in section \ref{s:topdefCFT}, one has invertible defects $\CL_g$, $g\in G$, that preserve the $G$-invariant subVOA $V^G$, and define a fusion category $Vec^\alpha_G$. We also assume that the symmetry is not anomalous, i.e. the topological junctions $V_g\otimes V_h\to V_{gh}$ can be chosen so that the associator $\alpha$ is trivial (see section \ref{s:topdefCFT}). The results of this section can be generalized to the case where $[\alpha]$ is non-trivial, but we will not consider that case in this article.

The $g$-twisted sectors $V_g$, $g\in G$, can be decomposed as direct sums of irreducible $V^G$-modules (see below). Up to isomorphism, the irreducible $V^G$-modules $M_{[g],\rho}$ obtained in this way can be labeled by a $G$-conjugacy class $[g]$ and by $\rho\in Irr(C_G(g))$ a representation of the centralizer $C_G(g)=\{h\in G\mid gh=hg\}$. Every irreducible $V^G$-module is isomorphic  to some $M_{[g],\rho}$.  We denote by
\be \chi_{[g],\rho}(\tau):=\Tr_{M_{[g],\rho}}(q^{L_0-\frac{c}{24}})\ ,
\ee the corresponding characters. 

The T- and S-matrices for the $\chi_{[g],\rho}$ characters are given by \cite{Coste:2000tq}
\be T_{([g],\rho),([h],\rho')}=\delta_{[g],[h]}\delta_{\rho,\rho'}\frac{\Tr_\rho(g)}{\Tr_\rho(e)}
\ee
\be \label{eq:Sgh} S_{([g],\rho),([h],\rho')}=\frac{1}{|C_G(g)||C_G(h)|}\sum_{k\in G(g,h)}\Tr_\rho(khk^{-1})^*\Tr_{\rho'}(k^{-1}gk)^*,
\ee
where, for any two $g,h\in G$, we let
\be G(g,h):=\{k\in G\mid gkhk^{-1}=khk^{-1}g\}
\ee be the (possibly empty) set of elements $k\in G$ for which $g$ and $khk^{-1}$ commute, so that $khk^{-1}\in C_G(g)$ and $k^{-1}gh\in C_G(h)$.

The VOA $V$ decomposes into irreducible modules of $G\times V^G$ as \cite{DLM2}
\be\label{Vdecomp} V=\bigoplus_{\rho\in Irr(G)} W_e^{\rho}\otimes   M_{[e],\rho}\ ,\ee  where $W_e^{\rho}$ is a finite dimensional vector space furnishing an irreducible representation $\rho:G\to GL(W_e^\rho)$. Physically, $W_e^{\rho}$ can be interpreted as the subspace of states that are highest weight vectors  for $V^G$-modules isomorphic to  $M_{[e],\rho}$ (or, equivalently, the subspace of $V^G$-primary operators in the $M_{[e],\rho}$ representation). We denote by $\{\Phi_{[e],\rho}^\alpha\}_{\alpha=1,\ldots,\dim\rho}$ a basis for the space $W_e^{\rho}$ of these $V^G$-primary operators.  Therefore, the $g$-twined  partition function can be written as
\be Z^g(\tau)=\sum_{\rho\in Irr(G)} \Tr_\rho(g)\chi_{[e],\rho}(\tau)\ .
\ee

When $g=e$, one has $G(e,h)=G$, and
\be S_{([e],\rho),([h],\rho')}=\frac{1}{|C_G(h)|}\Tr_\rho(h)^*\dim(\rho')
\ee so that the $g$-twisted partition function is
\begin{align} Z_g(\tau)&=\sum_{\rho\in Irr(G)} \Tr_\rho(g)\sum_{[h]}\sum_{\rho'\in Irr(C_G(h))}S_{([e],\rho),([h],\rho')}\chi_{[h],\rho'}(\tau)\\
	&=\sum_{\rho\in Irr(G)}\sum_{[h]}\sum_{\rho'\in Irr(C_G(h))}\Tr_\rho(g)\frac{1}{|C_G(h)|}\Tr_\rho(h)^*\dim(\rho')\chi_{[h],\rho'}(\tau)\\
	&= \sum_{\rho'\in Irr(C_G(g))}\dim(\rho')\chi_{[g],\rho'}(\tau),
\end{align} where we used \eqref{orth1}. In particular, one can check that the partition function $Z(\tau)$ is modular invariant, and this justifies the decomposition \eqref{Vdecomp}.

More generally, the twisted-twining partition function is
\be Z^h_g(\tau)=\sum_{\rho'\in Irr(C_G(g))}\Tr_{\rho'}(h)\chi_{[g],\rho'}(\tau)\ ,
\ee and this relation can be inverted using \eqref{orth1} to give
\be \chi_{[g],\rho}(\tau)=\frac{1}{|C_G(g)|}\sum_{h\in C_G(g)}\Tr_{\rho}(h)^* Z_g^h(\tau)\ .
\ee The S-matrix \eqref{eq:Sgh} is such that the usual S-transformation holds 
\be Z^h_g(-1/\tau)=Z_{h^{-1}}^g(\tau)\ .
\ee

These formulae imply that, for each $g\in G$, the $g$-twisted sector $V_g\equiv V_{\CL_g}$ decomposes as
\be V_g=\bigoplus_{\rho\in Irr(C_G(g))} W_g^{\rho}\otimes M_{[g],\rho}\ ,
\ee where $W_g^\rho$ is the vector space of all $V^G$-primary $\CL_g$-defect operators in $V_g$ of type $M_{[g],\rho}$, i.e. that belong to a irreducible $V^G$-module isomorphic to $M_{[g],\rho}$.  The operators $\hat\CL_h$, $h\in C_G(g)$, have a well defined action $\hat\CL_h: V_g\to V_g$; in particular,  they act on the space $W_g^\rho$ via the representation $\rho:C_G(g)\to GL(W_g^{\rho})$. 

More generally, the operator $\hat\CL_x$, for $x\in G$ not necessarily commuting with $g$, maps $V_g$ to $V_{xgx^{-1}}$, and in particular gives an isomorphism of $C_G(g)$ representations
\be \hat\CL_x:W_g^\rho\stackrel{\cong}{\longrightarrow}W_{g'}^{\rho'}
\ee  where $g'=xgx^{-1}$ and $\rho':C_{G(g')}\to GL(W_g')$ is such that $\rho'(xhx^{-1})=\rho(h)$. Indeed, it must map primaries  in $V_g$ in the $V^G$-module $M_{[g],\rho}$ to primaries in the isomorphic $V^G$-module $M_{[g'],\rho'}\cong M_{[g],\rho}$ in $V_{g'}$.\footnote{The spaces $V_{g}$ and $V_{xgx^{-1}}$ are isomorphic as $V^G$-modules (in particular, as modules for the Virasoro algebra), but not as twisted $V$-modules, i.e. there is no isomorphism of vector spaces $f:V_{g}\stackrel{\cong}{\to}V_{xgx^{-1}}$ such that $f^{-1} Y_{xgx^{-1}}(v,z)f=Y_g(v,z)$ for all $v\in V$.}

\subsection{Topological defects from orbifolds by subgroups}\label{s:orbdefects}

Let us suppose, as before, that the self-dual VOA $V$ has a finite symmetry group $G$, such that $V^G$ is strongly rational. Suppose that the group is anomaly free, so that we can choose the topological $3$-junctions $\mu_{g,h}:\CL_g\otimes\CL_h \to \CL_{gh}$ for all $g,h\in G$ in such a way that the $3$-cocycle $\alpha$ is trivial. Let $H\subseteq G$ be a (not necessarily normal) subgroup, and let us consider the orbifold theory $V/H$. By definition, $V/H$ is obtained by considering the space $\oplus_{h\in H} V_h$ and projecting on the $H$-invariant subspace
\be V/H\cong (\bigoplus_{h\in H} V_h)^H\ .
\ee The parent theory $V$ and the orbifold $V/H$ have  the $H$-invariant subalgebra $V^H\subset V$ in common. We know that the topological defects in $V/H$ preserving the subalgebra $V^H$ are given by Wilson lines for the group $H$ and form a fusion category isomorphic to $Rep(H)$ \cite{Frohlich:2009gb, Bhardwaj:2017xup}.

We are interested in understanding the larger category of topological defects of the orbifold $V/H$ that preserve the subalgebra $V^G\subseteq V^H$, but not necessarily $V^H$. Recall that the spectrum and all correlation functions in the orbifold theory $V/H$ can be completely determined in terms of correlation functions in $V$ in the presence of defects. Thus, it must be possible to determine the defects $\CL$ in the orbifold theory completely in terms of defects in $V$. In particular, each $V^G$-primary local or defect operator in $V$ should give rise to a $V^G$-primary local or defect operator in $V/H$. Thus, if we include all twisted sectors, the set of $V^G$-primary operators is the same for the two theories. In fact, the main difference between $V$ and $V/H$ is which subset of the set of $V^G$-primary operators is genuinely local (i.e. the lines they are attached to are identified with the trivial line), and which ones are attached to defect lines that are non-trivial (and, of course, to what kind of defects).

In order to determine the topological defects $\CL$ preserving $V^G$ in the orbifold theory $V/H$, our strategy is to analyse which sets of $V^G$-primaries can be consistently included in a given twisted space $V_\CL$. Suppose a defect space $V_\CL$ in the orbifold theory $V/H$ contains a $V^G$-primary operator $\Phi_{g,\rho}$ of type $M_{[g],\rho}$; in the parent theory $V$, this operator was contained in some $g$-twisted sector $V_g$. Then $V_\CL$ must also contain any operator that can be obtained from the OPE of $\Phi_{g,\rho}$ with the local operators in $V/H$. In this sense, we say that $V_\CL$ must be stable with respect to OPE with operators in $V/H$.
The orbifold VOA $V/H$ contains $H$-invariant $V_G$-primaries $\Phi_h$ in the $h$-twisted sector for each $h\in H\subseteq G$. The OPE of $\Phi_{g,\rho}\in V_{\CL}$ with $\Phi_{h}\in V/H$ yields operators in the $gh$-twisted and/or $hg$-twisted $V$-modules.   We conclude that $V_\CL$ must contain $g'$-twisted primary operators for all $g'$ in the double coset 
$$HgH:=\{h_1gh_2\in G\mid h_1,h_2\in H\}\ .$$ If $V_g$ is the $\CL_g$-twisted defect space in the original theory $V$, then
\be V_{HgH}:=\bigoplus_{g'\in HgH} V_{{g'}}
\ee is stable  with respect to OPE with $V/H$, and therefore can represent a defect space $V_\CL$ in the orbifold theory $V/H$. 

However, in general, a defect $\CL_{HgH}$ with twisted space $V_{HgH}$ will not be simple. To see this, let us first note that there is a well-defined action of the group $H$ on  $ V_{HgH}$: in particular, for each $g'\in HgH$, an element $h\in H$ maps the component $V_{g'}\subset V_{HgH}$ to $V_{hg'h^{-1}}$, where $hg'h^{-1}$ is still in the same coset $HgH$.  Thus, $V_{HgH}$ carries a representation of $H$, and we can decompose it into irreducible $H$-representations:
\be \label{eq:Hirreps} V_{HgH}=\bigoplus_{\rho\in Irr(H)} V^{\rho}_{HgH}\ .
\ee Now, each $V^{\rho}_{HgH}$ is stable with respect to the  OPE with operators in $V/H$ because all operators in $V/H$ are in the trivial $H$-representation.  Therefore,  each $V^{\rho}_{HgH}$ represents the defect space for a topological defect in $V/H$; we denote such defect by $\CL_{HgH,\rho}$. For example, for the trivial coset $HeH=H$, the defects $\CL_{HeH,\rho}$, $\rho\in Irr(H)$, include  the identity defect in $V/H$ and the simple $H$ Wilson lines.

It is easy to compute the $\CL_{HgH,\rho}$-twisted partition functions in terms of the twisted-twining partition functions $Z_g^h$ in the parent theory $V$. Indeed, by construction, one has
\be Z_{HgH,\rho}(V/H,\tau)=\frac{1}{|H|}\sum_{g'\in HgH}\sum_{\substack{h\in H\\ hg'=g'h}}Tr_{\rho}(h)^*Z_{g'}^h(V,\tau).
\ee
This formula uses the orthogonality of characters \eqref{orth2} to project the space $ V_{HgH}$ onto the subspace of states transforming in a representation $\rho$. Let us explain the restriction on commuting pairs $g'h=hg'$. When acting on a space $V_{HgH}$, which is the direct sum of subspaces $V_{g'}$ labeled by elements $g'\in HgH$, an element $h\in H$ admits a block decomposition, with each block labeled by a pair $g',g''\in HgH$. The blocks can be non-zero only for $g''=hg'h^{-1}$; on the other hand, only the diagonal blocks, i.e. for $hg'h^{-1}=g'$, contribute to the trace.  

The S-transform of $Z_{HgH,\rho}$ gives the $\CL_{HgH,\rho}$-twined partition function
\be\label{twinHgH} Z^{HgH,\rho}(V/H,\tau)=\frac{1}{|H|}\sum_{g'\in HgH}\sum_{\substack{h\in H\\ hg'=g'h}}Tr_{\rho}(h)^*Z^{g'}_{h^{-1}}(V,\tau)\ .
\ee 
From \eqref{twinHgH} one can read off the quantum dimension of the defect $\CL_{HgH,\rho}$, which is the coefficient of the vacuum character $\chi_{e,1}$ in the decomposition of \eqref{twinHgH}. Because the character $\chi_{e,1}$ appears only in the terms $Z_e^{g'}(\tau)$ and with coefficient $1$, one obtains
\be \langle\CL_{HgH,\rho}\rangle=\frac{|HgH|(\dim\rho)}{|H|}\ ,
\ee where $|HgH|$ is the number of elements in the coset $HgH$. Note that $|HgH|/|H|$ is always a positive integer, and it is $1$ if and only if $g$ normalizes $H$, i.e. if $gHg^{-1}=H$. Therefore, the defect $\langle\CL_{HgH,\rho}\rangle$ is invertible if and only if $\dim\rho=1$ and $gHg^{-1}=H$. In particular, if $H$ is a normal subgroup of $G$ (i.e. $gHg^{-1}=H$ for all $g\in G$), then there is at least one invertible symmetry element $\CL_{HgH,\rho=1}$ for each double coset $HgH$, corresponding to the trivial $H$-representation $\rho=1$. Furthermore, for $H$ a normal subgroup, the double cosets are in one-to-one correspondence with the left cosets, and the invertible defects $\CL_{HgH,1}$ generate the quotient group $G/H$. 

There are a number of subtleties in this construction:
\begin{itemize}
	\item  It may happen that two apparently different defects $\CL$, $\CL'$ constructed in this way turn out to be isomorphic, $\CL\cong \CL'$. This means that the corresponding operators $\hat\CL$ and $\hat\CL'$ act in the same way on all states in both $V/H$ and all  other defect spaces. A necessary (but not sufficient) condition for $\CL$ and $\CL'$ to be isomorphic is that the twisted spaces $V_{\CL}$ and $V_{\CL'}$ are isomorphic as $V^G$-modules, i.e. each irreducible $V^G$-module $M_{[g],\rho}$ has the same multiplicity in $V_{\CL}$ and in $V_{\CL'}$. In this case, the twisted partition functions are the same, and by S-transformation the twining partition functions are also the same. We stress, however, that having the same twining partition function does not necessarily mean that the action on $V/H$ is the same: an obvious counterexample is given by two distinct invertible defects in the same conjugacy class of the symmetry group. See the section \ref{s:orbexamples} for  examples.
	\item It may in principle happen that some of the defects $\CL_{HgH,\rho}$ are not simple. It should be possible to check this by computing the multiplicity of the identity defect in the tensor product of $V^{\rho}_{HgH}$ with its dual $V_{Hg^{-1}H}^{\rho^*}$.
	\item One consistency check that may help with the previous two subtleties is that the total dimension of the category of topological defects, i.e. the quantity
	\be \sum_{\CL\text{ simple}} \langle \CL\rangle^2
	\ee must be the same in the original theory and in the orbifold (see for example \cite{Bhardwaj:2017xup}). The sum here is over a set of representatives of isomorphism classes of simple defects. In particular, if we start with a category generated by invertible defects forming a group $G$ in the parent theory $V$,  then the square dimension of the category of  defects in the orbifold $V/H$ must  be equal to $|G|$.
\end{itemize}

In section \ref{s:orbexamples}, we will consider some examples.

	\subsubsection{Generalizations to SuperVOAs and non-holomorphic CFTs}\label{s:orbSVOAnonholom}
	
In section \ref{s:orbdefects}, we described the topological defects obtained from orbifolding (gauging) a non-abelian finite symmetry group of a holomorphic bosonic VOA $V$. We would like to generalize this construction both to supersymmetric  theories (holomorphic SVOAs) and to fully-fledged (non-holomorphic) SCFTs. 

In the first case, the main complication arises due to the presence of the fermion number $(-1)^F$. We will always consider the case of a (non-anomalous) symmetry group $G$ that commutes with $(-1)^F$, i.e. such that the group of automorphisms of the holomorphic SVOA $V$ contains a subgroup $G\times \langle(-1)^F\rangle$. Furthermore, we will assume that the group $G$ commutes with an $\CN=1$ superVirasoro algebra of central charge $c$, $SVir_c\subset V$. Even in this case, it could happen that the action of $G\times (-1)^F$ on the Ramond ($(-1)^F$-twisted) sector $V_{tw}$ is projective, such that the group acting linearly is a non-trivial extension $\tilde G$, which may not be  a direct product of $(-1)^F$ with another factor. Another possible complication is that the $g$-twisted sectors $V_g$ or $V_{g,tw}:= V_{(-1)^Fg}$ (NS and R sectors, resp.) do not have a well--defined fermion number, and that the expected modular transformations are modified in a non-trivial way (see the discussion in section \ref{s:TDLsuperVOA} and \cite{DongRenYang:2022} for more information).

Practically, we will  assume that these complications do not arise for the groups $G$ that we consider, i.e. we assume that all invertible defects $\CL_g$, $g\in G$, satisfy properties 1, 2, and 3 in section \ref{subsec:3.2}. This means that we assume that there is a finite direct  product group $G\times (-1)^F$ with a well-defined action on both  $V$ and $V_{tw}$. Furthermore, we will assume that all $g$-twisted sectors $V_g$ and $V_{g,tw}$ have a well-defined fermion number $(-1)^F$. Finally, we  assume that $G$ is non-anomalous, which implies that each $V_g$ and $V_{g,tw}$ sector carries a (genuine, non-projective) action of the centralizer $C_G(g)$ of $g$ in $G$, and that this action commutes with $(-1)^F$.

If all such assumptions hold, then the results of section \ref{s:orbdefects} generalize to the SVOA case with minor modifications. For example, for every commuting pair $g,h\in G$, $gh=hg$, one can define a $g$-twisted, $h$-twined vector-valued partition function $\CZ_g^h:=(Z_{g,\NS}^{h,+},Z_{g,\NS}^{h,-},Z_{g,\R}^{h,+},Z_{g,\R}^{h,-})^t$ with components
\begin{align}\label{eq:TwistTwineNS}
 	Z^{h,\pm}_{g,\NS}(V,\tau)&:=\Tr_{V_g}(q^{L_0-\frac{c}{24}}(\pm 1)^Fh)\\\label{eq:TwistTwineR}
 	Z^{h,\pm}_{g,\R}(V,\tau)&:=\Tr_{V_{g,tw}}(q^{L_0-\frac{c}{24}}(\pm 1)^Fh)\ ,
 \end{align}
satisfying the expected modular properties. Upon taking $g=e$ or $h=e$, this reduces to (\ref{eq:SymTwineNS})--(\ref{eq:SymTwineR}) or (\ref{eq:SymTwistNS})--(\ref{eq:SymTwistR}), respectively. In particular, one has
\be \CZ_g^h(-1/\tau)=\rho(S)\CZ_h^{g^{-1}}(\tau)\ ,
\ee where $\rho(S)$ is the matrix in eq.\ref{modularspin}. 

Thus, given a subgroup $H\subset G$, one has a well defined orbifold SVOA $V/H=(\oplus_{h\in H}V_h)^H$, and we can construct a set of topological defects $\CL_{HgH,\rho}$ preserving the $G$-invariant subVOA $V^G$ for each double coset $HgH$, $g\in G$, and $\rho\in Irr(H)$. The vector-valued $\CL_{HgH,\rho}$-twined and $\CL_{HgH,\rho}$-twisted partition functions $\CZ^{Hgh,\rho}$ and $\CZ_{HgH,\rho}$ are defined by the same formulae as in the bosonic case, namely
\be\label{stwinHgH} \CZ^{HgH,\rho}(V/H,\tau)=\frac{1}{|H|}\sum_{g'\in HgH}\sum_{\substack{h\in H\\ hg'=g'h}}\Tr_{\rho}(h)^*\CZ^{g'}_{h^{-1}}(V,\tau)\ ,
\ee and
\be\label{stwistHgH} \CZ_{HgH,\rho}(V/H,\tau)=\frac{1}{|H|}\sum_{g'\in HgH}\sum_{\substack{h\in H\\ hg'=g'h}}\Tr_{\rho}(h)^*\CZ_{g'}^h(V,\tau)\ .
\ee 
In particular, if we take $V= V^{f\natural}$, and $V^{f\natural}$ is self-orbifold by $H$, i.e. $V^{f\natural}/H\cong V^{f\natural}$, this construction provides a number of non-invertible simple TDLs in $\cTop$. In particular, our assumptions about the defects $\CL_g$, $g\in G$, in the original theory $V$, automatically imply that the non-invertible defects $\CL_{HgH,\rho}$ in the orbifold theory in $V/H$ are still `well-behaved' with respect to fermion number. This follows because the $\CL_{HgH,\rho}$-twisted sector can be defined in terms of modules for the $G$-invariant subalgebra $V^G$.

While the assumptions we make are in general quite strong, there are many examples for the SVOA $V=V^{f\natural}$ where one can reasonably expect these conditions to be satisfied. Indeed, suppose that $G$ is a subgroup of $\Aut_\tau(V^{f\natural})\cong Co_0$ that fixes pointwise a $4$-dimensional subspace $\Pi^\natural\subset {}^\RR V^{f\natural}_{tw}(1/2)$. This implies that the $G$-fixed subalgebra $V^G\subset V^{f\natural}$ (that we assume to be strongly rational) contains an $\widehat{so}(4)_1\cong (\widehat{su}(2)_1)^2$ current algebra, and that the group $Spin(4)$ generated by the current zero modes commutes with $G$. Moreover, it is easy to check that $(-1)^F$ acts on $V^{f\natural}$ and $V^{f\natural}_{tw}$ in the same way as the non-trivial central element in the kernel of the projection $Spin(4)\to SO(4)$. Because every $g$-twisted NS and R sector $V_g$ and $V_{g,tw}$ is an ordinary unitary module for $\widehat{so}(4)_1\subset V^G$, it is also a module for the group $Spin(4)$, so that a consistent definition of $(-1)^F$ is given by the same central element of $Spin(4)$. Furthermore, with this definition, $(-1)^F$ commutes with the action of $C_G(g)$ on $V_g$ and $V_{g,tw}$. Thus, all of our assumptions are satisfied in this case, and eqs.\eqref{stwinHgH} and \eqref{stwistHgH} hold. 

In order to make a connection with K3 NLSMs, one can choose a current zero mode $J_0^3$ in one of the $\widehat{su}(2)_1$ in $\widehat{so}(4)_1\cong (\widehat{su}(2)_1)^2$, and for every $g,h\in G$, $gh=hg$, define the $g$-twisted, $h$-twined genus
\be\label{ghtwisttwiningVf} \phi^h_g(V^{f\natural},\tau,z):=\Tr_{V^{f\natural}_{g,tw}}(h q^{L_0-\frac{c}{24}}y^{J_0^3}(-1)^F)\ .
\ee Similarly, given a subgroup $H\subset G$, one can define the twisted and twining genera $\phi_{HgH,\rho}(V^{f\natural}/H,\tau,z)$ and $\phi^{HgH,\rho}(V^{f\natural}/H,\tau,z)$ for each topological defect $\CL_{HgH,\rho}$ in the orbifold theory $V^{f\natural}/H$, with
\be\label{genustwinHgH} \phi^{HgH,\rho}(V^{f\natural}/H,\tau,z)=\frac{1}{|H|}\sum_{g'\in HgH}\sum_{\substack{h\in H\\ hg'=g'h}}\Tr_{\rho}(h)^*\phi^{g'}_{h^{-1}}(V^{f\natural},\tau,z)\ ,
\ee and
\be\label{genustwistHgH} \phi_{HgH,\rho}(V^{f\natural}/H,\tau,z)=\frac{1}{|H|}\sum_{g'\in HgH}\sum_{\substack{h\in H\\ hg'=g'h}}\Tr_{\rho}(h)^*\phi_{g'}^h(V^{f\natural},\tau,z)\ .
\ee When $V^{f\natural}/H\cong V^{f\natural}$, the defects $\CL_{HgH,\rho}$ are contained in the category $\cTop_{{\Pi'}^\natural}$, for a certain subset ${\Pi'}^\natural\subset {}^\RR V^{f\natural}_{tw}(1/2)$. Based on Conjecture \ref{conj:K3relation}, we expect (\ref{genustwinHgH}) and (\ref{genustwistHgH}) to be reproduced by a corresponding calculation for a defect for some K3 NLSM $\calC$.  Note that ${\Pi'}^\natural$ and $\Pi^\natural$ are in general distinct as subspaces of the real space $\Lambda\otimes \RR$. The reason is the following. The isomorphism $V\stackrel{\cong}{\longrightarrow} V/H$ is defined up to composition by $\Aut(V)\cong \Aut(V/H)\cong Spin(24)$ from the left or from the right. In general, it might not be possible to choose such an isomorphism so that it is the identity on the common subspace $V^H$. In particular, we can always choose it to act trivially on the supercurrent $\tau\in V^H$, but then it might map $\Pi^\natural$ to a different ${\Pi'}^\natural$, such that the two subspaces are not related by an automorphism in $\Aut_\tau(V)\cong Co_0$.  

Finally, we would like to apply the same reasoning to a generic supersymmetric non-linear sigma model $\calC$ on K3, which is a non-holomorphic (and, in general, non-rational) superconformal field theory. If $G$ is a non-anomalous group of symmetries preserving the $\CN=(4,4)$ superconformal algebra and the spectral flow, then the same argument used for $V^{f\natural}$ suggests that all NSNS and RR $g$-twisted sectors $\Hh_{g,\NSNS}$ and $\Hh_{g,\RR}$ have well-defined fermion number commuting with the (non-projective) action of $C_G(g)$. Therefore, we can define the 
$g$-twisted, $h$-twined partition functions $\CZ_g^h(\calC,\tau)=(Z_{g,\NS}^{h,+},Z_{g,\NS}^{h,-},Z_{g,\R}^{h,+},Z_{g,\R}^{h,-})^t$
with (non-holomorphic) components
\begin{align}\label{eq:TwistTwineNSK3}
 	Z^{h,\pm}_{g,\NS}(\calC,\tau)&:=\Tr_{\Hh^{K3}_{g,\NSNS}}(hq^{L_0-\frac{c}{24}}{\bar q}^{\bar L_0-\frac{\bar c}{24}}(-1)^{F+\bar F})\\\label{eq:TwistTwineRK3}
 	Z^{h,\pm}_{g,\R}(\calC,\tau)&:=\Tr_{\Hh^{K3}_{g,\HRR}}(hq^{L_0-\frac{c}{24}}{\bar q}^{\bar L_0-\frac{\bar c}{24}}(-1)^{F+\bar F})\ ,
 \end{align}
as well as (holomorphic)
$g$-twisted, $h$-twined genera 
\be\label{K3twisttwin} \phi^{h}_g(\calC,\tau,z):=\Tr_{\Hh^{K3}_{g,\HRR}}(h q^{L_0-\frac{c}{24}}{\bar q}^{\bar L_0-\frac{\bar c}{24}}y^{J_0^3}(-1)^{F+\bar F})\ .\ee
Furthermore, for a given subgroup $H\subset G$, the orbifold $\calC/H$ is a well-defined $\CN=(4,4)$ NLSM, following the assumption that $G$ is non-anomalous. Thus, we have a topological defect $\CL_{HgH,\rho}$ in the orbifold theory $\calC/H$, whose twisted and twining partition functions $\CZ_{HgH,\rho}(\calC/H,\tau)$ and $\CZ^{HgH,\rho}(\calC/H,\tau)$ and genera $\phi_{HgH,\rho}(\calC/H,\tau,z)$ and $\phi^{HgH,\rho}(\calC/H,\tau,z)$ are given again by, respectively \eqref{stwinHgH}, \eqref{stwistHgH}, \eqref{genustwinHgH} and \eqref{genustwistHgH}, upon replacing $V^{f\natural}$ with $\calC$. 

As discussed in section \ref{s:topdefK3}, for each subgroup $G_{\Pi^\natural}\subset \Aut_\tau(V^{f\natural})\cong Co_0$ fixing a subspace $\Pi^\natural\subset \Lambda\otimes \RR$, there exists a K3 model $\calC$ whose symmetry group preserving the $\CN=(4,4)$ algebra and the spectral flow is isomorphic to $G_\Pn$, with isomorphic $G_\Pn$  action on the $24$-dimensional spaces $V^{f\natural}_{tw}(1/2)$ and $\Hh_{\HRR,gr}$, and such that the $g$-twined genera match $\phi^g(V^{f\natural})=\phi^g(\calC)$ for all $g\in G_{\Pi^\natural}$. If one can prove that all $g$-twisted, $h$-twined genera are the same in the two theories, i.e. $\phi_g^h(V^{f\natural})=\phi_g^h(\calC)$ for all $g,h\in G_\Pn$, $gh=hg$, then we can conclude that, for any non-anomalous subgroup $H\subset G_\Pn$, the $\CL_{HgH,\rho}$ twisted and twining genera in the respective orbifold theories $V^{f\natural}/H$ and $\calC/H$ are also the same. We will explicitly show an example in section \ref{s:K3matching}.

\subsection{Examples}\label{s:orbexamples}

In this section, we present examples of the method described in section \ref{s:orbdefects} for constructing $V^{G}$--preserving topological defects in the orbifold theory $V/H$, where $H, G$, $H\subseteq G$ are finite symmetry groups of the VOA $V$. 
In particular, for $G=S_3$, $H=\ZZ_2$, we find a category of order six with the fusion rules of $Rep(S_3)$ (section \ref{s:S3Z2}), for $G=S_3$, $H=\ZZ_3$, we find a category of invertible defects which form the group $S_3$ (section \ref{s:S3Z3}), and finally, for $G=A_5$, $H=S_3$, we find a fusion category of order 60 with fusion rules given in equations (\ref{eq:fus1})-(\ref{eq:fus4}) (section \ref{s:A5S3}).
Note that the results we describe in this section yield $V^G$--preserving fusion categories of $V/H$ for any VOA $V$ with non-anomalous $G$--symmetry and strongly rational $V^G$.

\subsubsection{$G=S_3$, $H=\ZZ_2$.}\label{s:S3Z2}

The group $G=S_3$ has six elements, the identity $e$, the transpositions $g_1,g_2,g_3$, with $g_i$ fixing the point $i\in \{1,2,3\}$, and the order $3$ permutations $h=g_1g_2$ and $h^{-1}$. It has three conjugacy classes $[e]$, $[g_i]$, $[h]=[h^{-1}]$, determined by the order of their elements. The centralizer of each non-trivial element $k\in G$ is the cyclic group $C_G(k)= \langle k\rangle\cong \ZZ_n$ generated by the element itself. There are three irreducible representations of $G$, which we denote by $\rho_0$ (the trivial representation), $\rho_1$ (the $1$-dimensional sign representation), and $\rho_2$ ($2$-dimensional). The character table is:
\be\label{ctblS3} \begin{array}{c|ccc}
	& [e] & [g_i] & [h]\\
\hline	\rho_0 & 1 & 1 & 1\\
	\rho_1 & 1 & -1 & 1\\
	\rho_2 & 2 & 0 & -1\end{array}
\ee
 For each element $g_i$ of order $2$, we denote by $\rho_+$ and $\rho_-$ the two irreducible representations of $C_G(g_i)=\langle g_i\rangle\cong \ZZ_2$ with $\rho_\pm(g_i)=\pm 1$. Finally, we denote by $\sigma_1$, $\sigma_\omega$ and $\sigma_{\bar\omega}$ the three irreducible representations of $C_G(h)=\langle h\rangle\cong \ZZ_3$, where $\omega=e^{2\pi i/3}$ and $\sigma_\gamma(h)=\gamma$, $\gamma\in \{1,\omega,\bar\omega\}$.  

In the parent theory $V$, almost all spaces $W_g^{\rho}$ of $V^G$-primary operators are $1$-dimensional, except for $W_e^{\rho_2}$, which has dimension 2. We denote these primary operators as
\be \Phi_{e,0},\quad \Phi_{e,1}\quad \Phi^{\alpha}_{e,2},\quad \Phi_{g_i,+},\quad \Phi_{g_i,-},\quad \Phi_{h,\gamma},\quad \Phi_{h^{-1},\gamma}
\ee
where $i=1,2,3$, and the index $\alpha$ labels the two different primaries in the $V^G$-module of type $M_{[e],\rho_2}$. The space $\oplus_{g\in G}V_g$ of all local and defect operators of the original theory $V$ contains three $V^G$-primaries of type $M_{[g_i],+}$ and three of type $M_{[g_i],-}$   (denoted, respectively, as $\Phi_{g_i,+}$ and $\Phi_{g_i,-}$, $i=1,2,3$), while there are two primaries $\Phi_{h,\gamma}$ and $\Phi_{h^{-1},\bar\gamma}$ of type $M_{[h],\gamma}$.

We consider the orbifold $V/H$ by the subgroup $H=\langle g_1\rangle\cong \ZZ_2$. Note that $g_1$ has eigenvalues $\pm1$ on the $2$-dimensional space $W_e^{\rho_2}$, so it is convenient to take a basis  of $g_1$-eigenvectors $\Phi^{+}_{e,2},\Phi^{-}_{e,2}$. Then, the original VOA $V$ contains the $V^G$-primaries
\be \Phi_{e,0}, \Phi_{e,1}, \Phi^{+}_{e,2},\Phi^{-}_{e,2}\in V\ ,
\ee while the orbifold VOA $V/H$ contains the $V^G$-primaries
\be \Phi_{e,0}, \Phi^{+}_{e,2},\Phi_{g_1,+}\in V/H\ .
\ee  In particular, $\Phi_{e,0}, \Phi^{+}_{e,2}$ and their $V^G$-descendants generate the $H$-invariant subalgebra $V^H$. In $V/H$ there is a quantum symmetry Wilson line $\CL_{\rho_-}$, whose twisted Hilbert space $V_{\rho_-}$ contains the primaries
\be \Phi_{e,1}, \Phi^{-}_{e,2},\Phi_{g_1,-}\in V_{\rho_-}\ .
\ee It is a $\ZZ_2$ invertible defect, with $\hat\CL_{\rho_-}$ acting by $-1$ on $\Phi_{g_1,+}$ and trivially on $\Phi_{e,0}, \Phi^{+}_{e,2}$.  This is the only non-trivial simple topological defect preserving the full algebra $V^H$.

Let us now consider defects preserving $V^G\subseteq V^H$, but not necessarily $V^H$. There is a single non-trivial double coset in $H\backslash G/H$, namely
\be Hg_2H=\{g_2,\ g_1g_2=h,\ g_2g_1=h^{-1},\ g_1g_2g_1=g_3\}\ ,
\ee that contains all elements of $G$ that are not in $H$.  Thus, following the discussion around equation (\ref{eq:Hirreps}, we expect there to exist two non-invertible defects, $\CW_+:=\CL_{Hg_2H,+}$ and $\CW_-:=\CL_{Hg_2H,-}$,  corresponding to the two irreducible representations of $H\cong\ZZ_2$. The quantum dimension of these defects must be $2$, because their superposition $\CW_++\CW_-$ should correspond to the superposition $\CL_{g_2}+\CL_{g_3}+\CL_h+\CL_{h^{-1}}$ of four invertible defects in $V$. 

Let us now attempt to determine which $V^G$-primaries are contained in $V_{\CW_+}$ and $V_{\CW_-}$. Let us take $g_1$ to act on $\Phi_{h,\gamma}$ by
\be g_1(\Phi_{h,\gamma})=\Phi_{h^{-1},\bar\gamma}\ .
\ee Note that $\Phi_{h,\gamma}$ and $\Phi_{h^{-1},\bar\gamma}$ are in the same $V^G$-module $M_{[h],\bar\gamma}$. Our guess is
$$\Phi_{g_2,+}+\Phi_{g_3,+},\quad \Phi_{g_2,-}+\Phi_{g_3,-},\quad \Phi_{h,1}+\Phi_{h^{-1},1},\quad \Phi_{h,\omega}+\Phi_{h^{-1},\bar\omega},\quad \Phi_{h,\bar\omega}+\Phi_{h^{-1},\omega}\quad \in V_{\CW_+}
$$
$$ \Phi_{g_2,+}-\Phi_{g_3,+},\quad \Phi_{g_2,-}-\Phi_{g_3,-},\quad \Phi_{h,1}-\Phi_{h^{-1},1},\quad \Phi_{h,\omega}-\Phi_{h^{-1},\bar\omega} ,\quad\Phi_{h,\bar\omega}-\Phi_{h^{-1},\omega}\quad \in V_{\CW_-}.
$$
The two defect Hilbert spaces $\Hh_{\CW_+}$ and $\Hh_{\CW_-}$ have exactly the same $V^G$-modules and multiplicities, i.e. the twisted partition functions must be the same, $Z_{\CW_+}(\tau) =Z_{\CW_-}(\tau)$. As a consequence,  their S-transforms, i.e. the twining partition functions $Z^{\CW_+}(\tau) =Z^{\CW_-}(\tau)$, are also the same. It follows that $\hat \CW_+$ and $\hat \CW_-$ must act with the same eigenvalues on the three $V^G$-primaries in $V/H$. These observations suggest that the two defects are isomorphic,
$$ \CW_+\cong \CW_-\equiv \CW\ , 
$$ in the sense that there is no way to distinguish them when inserted in correlation functions.\footnote{In principle, one should also check that $\hat \CW_-$ and $\hat \CW_+$ act in the same way on all defect Hilbert spaces, and not only on $V/H$.} This conclusion is consistent with the fact that there are only three $V^G$-primaries in $V/H$, so the action of any $V^G$--preserving defect $\hat\CL$ depends on only three parameters. The fusion rules must be the ones of $Rep(S_3)$:
\be \CL_{\rho_-}^2=\CI\ ,\qquad \CL_{\rho_-}\CW=\CW\CL_{\rho_-}=\CW\ ,\qquad \CW^2=\CI+\CL_{\rho_-}+\CW\ .
\ee The rule $\CW\CL_{\rho_-}=\CL_{\rho_-}\CW=\CW$ implies that $\hat \CW$ must annihilate the primary $\Phi_{g_1,+}$. On the primaries that are invariant under $\CL_{\rho_-}$, the possible eigenvalues of $\hat \CW$ are $-1$ and $2$, so it must act as
\be \hat \CW(\Phi_{e,0})=2\Phi_{e,0}\ ,\qquad \hat \CW(\Phi_{e,2}^+)=-\Phi_{e,2}^+\ ,\qquad \hat \CW(\Phi_{g_1,+})=0\ .
\ee

Let us check that the partition functions make sense. The twining partition function $Z^\CW$ is
\be Z^\CW(\tau)=2\chi_{[e],0}(\tau)-\chi_{[e],2}(\tau)\ .
\ee The S-matrix of equation (\ref{eq:Sgh}) for the group $G=S_3$  is (see \cite{Coste:2000tq})
\be \frac{1}{6}\begin{pmatrix}
	1 & 1 & 2 & 3 & 3 & 2 & 2 & 2\\
    1 & 1 & 2 & -3 & -3 & 2 & 2 & 2\\
    2 & 2 & 4 & 0 & 0 & -2 & -2 & -2\\
	3 & -3 & 0 & 3 & -3 & 0 & 0 & 0\\
	3 & -3 & 0 & -3 & 3 & 0 & 0 & 0\\
    2 & 2 & -2 & 0 & 0 & 4 & -2 & -2\\
    2 & 2 & -2 & 0 & 0 & -2 & -2 & 4\\
    2 & 2 & -2 & 0 & 0 & -2 & 4 & -2\\
\end{pmatrix} \,
\ee
acting on the vector of characters $(\chi_{[e],0},\chi_{[e],1},\chi_{[e],2},\chi_{[g_i],+},\chi_{[g_i],-},\chi_{[h],1},\chi_{[h],\omega},\chi_{[h],\bar\omega})^t$.
Applying this S-transformation to $Z^\CW$ leads to the $\CW$-twisted partition function,
\be Z_\CW(\tau)=\chi_{[g_i],+}(\tau)+\chi_{[g_i],-}(\tau)+\chi_{[h],1}(\tau)+\chi_{[h],\omega}(\tau)+\chi_{[h],\bar\omega}(\tau),
\ee which has a character decomposition with integral multiplicities and matches the list of primaries above. The square dimension of the category of defects is
\be \langle \CI\rangle^2+\langle \CL_{\rho}\rangle^2+\langle \CW\rangle^2=1^2+1^2+2^2=6,
\ee and matches the dimension $|S_3|=6$ before the orbifold. This confirms that $\CW_+$ and $\CW_-$ are indeed isomorphic.

\subsubsection{$G=S_3$, $H=\ZZ_3$.}\label{s:S3Z3}

Let us consider the same group $G=S_3$ as in the previous example, but now take the orbifold $V/H$ by the normal subgroup $H=\langle h\rangle\cong \ZZ_3$, so that
$$ \Phi_{e,0},\Phi_{e,1},\Phi_{h,1},\Phi_{h^{-1},1}\in V/H.
$$
 The trivial coset $HeH$ provides the three invertible Wilson lines $\CL_{\sigma_1}\equiv\CI$, $\CL_{\sigma_{\omega}}$, $\CL_{\sigma_{\bar\omega}}$, generating the $\ZZ_3$ quantum symmetry group. 
 
 There is another coset,
\be
Hg_1H=\{g_1,g_2,g_3\},
\ee
which leads to three additional defects as follows.
The $\CL_{Hg_1H,\sigma_\gamma}$-twisted partition function, for each $\gamma \in \{1,\omega,\bar\omega\}$, is given by
\be
Z_{Hg_1H,\sigma_\gamma}(\tau)=\frac{1}{3}\sum_{g'\in Hg_1H}\sum_{\substack{h\in H\\ hg'=g'h}}\Tr_{\sigma_\gamma}(h)^*Z_{g'}^h(\tau)=\frac{1}{3}\sum_{g'\in Hg_1H}Z_{g'}(\tau)=Z_{[g_i]}(\tau)
	\ee  where in the last step we used the fact that the twisted partition functions depend only on the $G$-conjugacy class. This means that all three defects for this coset have the same twining partition function
		\be Z^{Hg_1H,\sigma_\gamma}(\tau)=Z^{[g_i]}(\tau)=\chi_{e,\sigma_0}(\tau)-\chi_{e,\sigma_1}(\tau)\ .
	\ee
	 In this case, though, we cannot conclude that the three operators $\hat\CL_{Hg_1H,\sigma_\gamma}$ for $\gamma\in \{1,\omega,\bar\omega\}$ act in the same way on the primaries of $V/H$. In particular, the primaries $\Phi_{h,1},\Phi_{h^{-1},1}\in V/H$ are in the same $V^G$-module, so that each $\hat\CL_{Hg_1H,\sigma_\gamma}$ acts on them by a $2\times 2$ matrix. The fact that the trace of these three matrices is the same (in fact, $0$) does not necessarily imply that they act in the same way.
	
	A simple comparison with the defects in the original theory shows that these defects have dimension $1$, and therefore they are invertible. Thus, together with $\CL_{\sigma_{\omega}}$, $\CL_{\sigma_{\bar\omega}}$ they generate a group of order $6$, and the only possibilities are $\ZZ_6$ and $S_3$. However, only $S_3$ is compatible with the twining genera; in particular, $\Phi_{h,1},\Phi_{h^{-1},1}\in V/H$ span the irreducible $2$-dimensional representation of $S_3$.

\subsubsection{$G=A_5$, $H=S_3$.}\label{s:A5S3}

The case $G=A_5$ and $H=S_3$ is a simple example where $H$ is not normal in $G$ and both groups are non-abelian. Let us consider $A_5$ to be the group of even permutations of $\{1,2,3,4,5\}$, and $S_3\subset A_5$ to be the subgroup preserving the decomposition $\{1,2,3,4,5\}=\{1,2,3\}\sqcup\{4,5\}$. The character table of $A_5$ is
\be\label{ctblA5} \begin{array}{ccccc}
	  1A & 2A & 3A & 5A & 5B\\
\hline	  1  & 1 & 1 & 1 & 1        \\
	  3 & -1 & 0 & -b_5 & -b_5^*\\
	  3 & -1 & 0 & -b_5^* & -b_5\\
	  4 & 0 &1 & -1 & -1\\
	  5 & 1 & -1 & 0 & 0
\end{array}
\ee where $b_5=\frac{-1-\sqrt{5}}{2}$ and $b_5^*=\frac{-1+\sqrt{5}}{2}$.
We indicate the irreducible $V^G$-modules by $M_{nX,\rho}$, where $nX\in \{1A,2A,3A,5A,5B\}$ is a $G$-conjugacy class, and $\rho$ is an irreducible representation of the centralizer $C_G(g)$ in $G$ of a representative $g$ of $nX$. We are mostly interested in the $G$-conjugacy classes that contain some representative in $H$. In particular, the three elements $g_1,g_2,g_3\in H$ of order $2$ belong to the $G$-class $2A$, and the two elements $h,h^{-1}\in H$ of order $3$ belong to the $G$-class $3A$. For $g=e$, the irreducible representations of $C_G(e)\equiv G$ are denoted by $\{1,3,\bar 3,4,5\}$ according to their dimension. The decomposition of the irreducible $G$-representations into $H$ representations can be determined from the character tables \eqref{ctblS3} and \eqref{ctblA5}:
\be\label{A5toS3} 3,\bar3\rightarrow \rho_1+\rho_2\ ,\qquad 4\rightarrow\rho_0+\rho_1+\rho_2\ ,\qquad 5\rightarrow\rho_0+2\rho_2 \ ,
\ee where $\rho_0$, $\rho_1$, and $\rho_2$ are the three irreducible representations of $H=S_3$, namely the trivial, the non-trivial $1$-dimensional, and the $2$-dimensional one, respectively.

When $g\in G$ is in class $2A$, the centralizer $C_G(g)$ is $\langle g,k\rangle\cong \ZZ_2\times\ZZ_2$, for a certain element $k\in G$ that depends on $g$. We denote by $\rho_{\epsilon,\epsilon'}$, $\epsilon,\epsilon'\in\{0,1\}$ the four irreducible representations of $\langle g,k\rangle$; they are all $1$-dimensional, and defined by $\rho_{\epsilon,\epsilon'}(g)=(-1)^\epsilon$ $\rho_{\epsilon,\epsilon'}(k)=(-1)^{\epsilon'}$. In particular, for $g_i\in H$ of order $2$, the centralizer in $H$ is $C_H(g_i)=C_G(g_i)\cap H=\langle g_i\rangle\cong \ZZ_2$, so that $\rho_{00}$ and $\rho_{01}$ become trivial as representations of $C_H(g_i)$. Finally, for $g\in G$ in class $3A$, the centralizer $C_G(g)=\langle g\rangle\cong \ZZ_3$, and we denote by $\sigma_\gamma$, $\gamma\in \{1,\omega,\bar\omega\}$ the three irreducible $1$-dimensional representations with $\sigma_{\gamma}(g)=\gamma$.

The $V^G$-primary operators in $V/H$ are
\be \Phi_{1A,1}\ ,\quad \Phi_{1A,4}\ ,\quad \Phi_{1A,5}\ ,\quad \Phi_{2A,00},\quad \Phi_{2A,01},\quad \Phi_{3A,1}\quad \in V/H\ .
\ee The notation is as follows. The operators $\Phi_{nX,\rho}$ are labeled by the $V^G$-module $M_{nX,\rho}$ to which they belong, where $nX\in \{1A,2A,3A,5A,5B\}$ is a $G$-conjugacy class and $\rho$ is a representation of the centralizer in $G$ of an element of class $nX$. The $G$-classes $1A$, $2A$, and $3A$ are the only classes that contain elements in $H$. The fields $\Phi_{1A,4}\in W_{1A}^4$ and $\Phi_{1A,5}\in W_{1A}^5$ are the unique $H$-invariant primary operators in $V^G$-modules of type $M_{1A,4}$ and $M_{1A,5}$; the fact that there is exactly one such operator in $W_{1A}^4$ and $W_{1A}^5$ and no $H$-invariant operator in $W_{1A}^3$ or $W_{1A}^{\bar 3}$ follows from the multiplicities of the $H$-trivial representation $\rho_0$ in the decomposition \eqref{A5toS3}.
 The operators $\Phi_{2A,\epsilon\epsilon'}$, $\epsilon,\epsilon'\in\{0,1\}$  are the only $V^G$-primary operators in the $M_{2A,\rho_{\epsilon\epsilon'}}$ modules. In particular, $\rho_{00}$ and $\rho_{01}$ denote the two representations where the restriction of the centralizer to $H$ has a trivial action.

There are three distinct double cosets in $H\backslash G/H$, with representatives the identity ($6$ elements), the permutation $(3,4,5)$ ($36$ elements) and the permutation $(2,4)(3,5)$ ($18$ elements). As usual, the identity coset gives rise to the $Rep(S^3)$ category of Wilson lines, with twining partition functions
\begin{align} Z^{\rho_0}(\tau)=&\chi_{1A,1}(\tau)+\chi_{1A,4}(\tau)+\chi_{1A,5}(\tau)+\chi_{2A,00}(\tau)+\chi_{2A,01}(\tau)+\chi_{3A,1}(\tau)\\ Z^{\rho_1}(\tau)=&\chi_{1A,1}(\tau)+\chi_{1A,4}(\tau)+\chi_{1A,5}(\tau)-\chi_{2A,00}(\tau)-\chi_{2A,01}(\tau)+\chi_{3A,1}(\tau)\\
Z^{\rho_2}(\tau)=&2\chi_{1A,1}(\tau)+2\chi_{1A,4}(\tau)+2\chi_{1A,5}(\tau)-\chi_{3A,1}(\tau)\ .
\end{align}

The $H(3,4,5)H$ coset contains $6$ elements in class $2A$, $18$ in class $3A$, and $6$ in each of the classes $5A$ and $5B$. For any given coset representative $g'\in H(3,4,5)H$,  the only element of $H$ commuting with $g'$ is the identity. From formula \eqref{twinHgH}, it follows that the corresponding defect partition functions are integral integral multiples of 
\begin{align*} Z^{H(3,4,5)H,\rho_0}(\tau)&=\frac{1}{|H|}\sum_{g'\in H(3,4,5)H} Z_e^{g'}=\frac{1}{6}(6Z_e^{2A}+18Z_e^{3A}+6Z_e^{5A}+6Z_e^{5B})
\\&=6\chi_{1A,1}(\tau)+\chi_{1A,4}(\tau)-2\chi_{1A,5}(\tau)\ .
\end{align*} We denote by \be X:=\CL_{H(3,4,5)H,\rho_0}\ee this simple defect of dimension $\langle X\rangle=6$. 

The $H(2,4)(3,5)H$ coset consists of $6$ elements in each of the classes $2A$, $5A$, and $5B$. Each element of order $2$ in this coset commutes with exactly one element of order $2$ in $H$. For each $g_i\in H$ of order $2$, $i=1,2,3$, the centralizer $C_G(g_i)\cong \ZZ_2\times \ZZ_2$  contains two elements $k_i,g_ik_i\in G$ of order two that are not in $H$. Thus, the six elements of order $2$ in the coset  $H(2,4)(3,5)H$ are exactly $k_i,g_ik_i$, $i=1,2,3$. The elements of order $5$ in the coset commute only with the identity element in $H$. Therefore, using \eqref{twinHgH}, the twining partition functions are
\begin{align} Z^{H(2,4)(3,5)H,\rho_0}(\tau)&=\frac{1}{6}[6Z_e^{2A}+6Z_e^{5A}+6Z_e^{5B}+\sum_{i=1}^3(Z^{g_i}_{k_i}+Z^{g_i}_{g_ik_i})]\\
&=3\chi_{1A,1}(\tau)-2\chi_{1A,4}(\tau)+\chi_{1A,5}(\tau)+\chi_{2A,00}(\tau)+\chi_{2A,01}(\tau)
\end{align}
\begin{align} Z^{H(2,4)(3,5)H,\rho_1}(\tau)&=\frac{1}{6}[6Z_e^{2A}+6Z_e^{5A}+6Z_e^{5B}-\sum_{i=1}^3(Z^{g_i}_{k_i}+Z^{g_i}_{g_ik_i})]\\ &=3\chi_{1A,1}(\tau)-2\chi_{1A,4}(\tau)+\chi_{1A,5}(\tau)-\chi_{2A,00}(\tau)-\chi_{2A,01}(\tau)
\end{align}
\be Z^{H(2,4)(3,5)H,\rho_2}(\tau)=\frac{2}{6}(6Z_e^{2A}+6Z_e^{5A}+6Z_e^{5B})=Z^{H(2,4)(3,5)H,\rho_0}(\tau)+Z^{H(2,4)(3,5)H,\rho_1}(\tau)
\ee
We denote  by
\be Y_+:=\CL_{H(2,4)(3,5)H,\rho_0}\ ,\qquad Y_-:=\CL_{H(2,4)(3,5)H,\rho_1}\ ,
\ee the two defects of dimension $\langle Y_\pm\rangle=3$. Clearly, the defect $\CL_{H(2,4)(3,5)H,\rho_2}=Y_++Y_-$ is not simple.

All in all, we have six defects: $\CL_{\rho_0}\equiv\CI$, $\CL_{\rho_1}$, $\CL_{\rho_2}$, $Y_\pm$, and $X$. Because each $V^G$-module appears with multiplicity at most $1$ in the decomposition of $V/H$,  the twining partition functions $Z^\CL$ are sufficient to uniquely determine  the eigenvalues of each $\hat\CL$ acting on $V/H$:
\be\label{eq:6defects}
\begin{array}{c|cccccc}
	 &\Phi_{1A,1} & \Phi_{1A,4}& \Phi_{1A,5}& \Phi_{2A,00} &\Phi_{2A,01}& \Phi_{3A,1}
	\\
	\hline
	 \CI & 1 & 1 & 1 & 1 & 1 & 1 \\
	\CL_{\rho_1} & 1 & 1 & 1 & -1 & -1 & 1\\
	\CL_{\rho_2} & 2 & 2 & 2 & 0 & 0 & -1\\
	Y_+ & 3 & -2 & 1 & 1 & -1 & 0\\
	Y_- & 3 & -2 & 1 & -1 & 1 & 0\\
	X & 6 & 1 & -2 & 0 & 0 & 0
	\end{array}
	\ee
From these eigenvalues, we can determine the fusion ring:
\be \label{eq:fus1}\CL_{\rho_1}^2=\CI\ ,\qquad \CL_{\rho_1}\CL_{\rho_2}=\CL_{\rho_2}\CL_{\rho_1}=\CL_{\rho_2}\ ,\qquad \CL_{\rho_2}^2=\CI+\CL_{\rho_1}+\CL_{\rho_2}\ ,
\ee
\be \label{eq:fus2} Y_{\pm}^2=\CI+\CL_{\rho_2}+X\ ,\qquad \CL_{\rho_1}Y_{\pm}=Y_{\pm}\CL_{\rho_1}=Y_{\mp}\ ,\qquad Y_+Y_-=Y_-Y_+=\CL_{\rho_1}+\CL_{\rho_2}+X\ ,
\ee
\be \label{eq:fus3} \CL_{\rho_2}Y_{\pm}=Y_{\pm}\CL_{\rho_2}=Y_++Y_-\ ,\qquad X^2=\CI+\CL_{\rho_1}+2\CL_{\rho_2}+2Y_++2Y_-+3X\ee 
\be \label{eq:fus4} X\CL_{\rho_1}=\CL_{\rho_1}X=X\ ,\qquad \CL_{\rho_2}X=X\CL_{\rho_2}=2X\ ,\qquad Y_{\pm}X=XY_{\pm}=Y_++Y_-+2X\ .
\ee
The fusion ring is commutative and  the six defects are simple and unoriented. As a consistency check, note that the dimension of the category is
\be
\langle \CI\rangle^2+\langle \CL_{\rho_1}\rangle^2+\langle \CL_{\rho_2}\rangle^2+\langle Y_+\rangle^2+\langle Y_-\rangle^2+\langle X\rangle^2=1^2+1^2+2^2+3^2+3^2+6^2=60\ ,
\ee which matches with the dimension $|G|=60$ of the category before the orbifold.

\section{Defects and twining functions in K3 sigma models}
\label{s:K3defects}
In this section we provide some evidence for Conjecture \ref{conj:K3relation} by constructing a set of non-invertible defects in K3 sigma models and by comparing them with non-invertible defects  in $\cTop$.  

Our main example of a K3 NLSM will be the GTVW model $\calC_{\GTVW}$ studied in \cite{Gaberdiel:2013psa,Harvey:2020jvu}, and some of its orbifolds. The GTVW model is a rational SCFT, whose bosonic chiral and antichiral algebras $\CA$ and $\bar\CA$ are isomorphic to $(\widehat{su}(2)_1)^6$. Recall that $\widehat{su}(2)_1$ has two unitary irreducible representations of respective conformal weight $0$ and $1/4$, that we denote by $[0]$ and $[1]$, respectively. The irreducible representations form a group $\ZZ_2\cong\{0,1\}$ with respect to fusion. We will denote representations of $\CA\otimes \bar\CA$  by
\be [a_1,\ldots,a_6;b_1,\ldots,b_6]\ ,
\ee with $a_i,b_i\in \ZZ_2=\{0,1\}$. The even (bosonic) and odd (fermionic) NSNS sectors of $\calC_{\GTVW}$ contain, respectively, the following $\CA\otimes\bar\CA$-representations 
\be (\NSNS)_+:\qquad \{[a_1,\ldots,a_6;b_1,\ldots,b_6]\in \ZZ_2^{12}\mid a_i\equiv b_i,\ \sum_i a_i\equiv 0\mod 2\}
\ee
\be (\NSNS)_-:\qquad \{[a_1,\ldots,a_6;b_1,\ldots,b_6]\in \ZZ_2^{12}\mid a_i\equiv b_i+1,\ \sum_i a_i\equiv 0\mod 2\}
\ee while the even and odd RR sectors contain
\be (\HRR)_+:\qquad \{[a_1,\ldots,a_6;b_1,\ldots,b_6]\in \ZZ_2^{12}\mid a_i\equiv b_i,\ \sum_i a_i\equiv 1\mod 2\}
\ee
\be (\HRR)_-:\qquad \{[a_1,\ldots,a_6;b_1,\ldots,b_6]\in \ZZ_2^{12}\mid a_i\equiv b_i+1,\ \sum_i a_i\equiv 1\mod 2\}\ .
\ee
The full group of symmetries of the SCFT is $(SU(2)^6\times SU(2)^6)\rtimes S_6$, where the two $SU(2)^6$ factors  are generated by the zero modes of, respectively, the holomorphic and anti-holomorphic currents,\footnote{Strictly speaking, there is a $\ZZ_2^6$ subgroup of the center $\ZZ_2^6\times\ZZ_2^6\subset SU(2)^6\times SU(2)^6 $ that acts trivially on all fields of the theory, so the group acting faithfully is a quotient of $(SU(2)^6\times SU(2)^6)\rtimes S_6$.} while $S_6$ is the group acting diagonally on the chiral and on the anti-chiral algebra by permutation of the six $\widehat{su}(2)_1$ factors.

The GTVW model contains many copies of the $\CN=(4,4)$ superconformal algebra. We consider the standard choice described in \cite{Gaberdiel:2013psa,Harvey:2020jvu}, where the holomorphic and antiholomorphic $\CN=4$ SCAs contain the first $\widehat{su}(2)_1$ factor in $\CA$ and $\bar\CA$, respectively. With this choice, the spectral flow generators are the four ground fields in the representation $[1,0,0,0,0,0;1,0,0,0,0,0]$ in the even RR sector.

The group of symmetries fixing the $\CN=(4,4)$ SCA and the spectral flow is a finite subgroup $G_{\GTVW}$ of $(SU(2)^6\times SU(2)^6)\rtimes S_6$ isomorphic to $\ZZ_2^8\rtimes M_{20}$ \cite{Gaberdiel:2013psa}, where $M_{20}\cong \ZZ_2^4\rtimes A_5$ is one of the Mathieu groups. The normal subgroup $\ZZ_2^8\rtimes \ZZ_2^4$ of $G_{\GTVW}$ is contained in $SU(2)^6\times SU(2)^6$, and the quotient $G_{\GTVW}/(\ZZ_2^8\rtimes \ZZ_2^4)\cong A_5$ acts on the $(\widehat{su}(2)_1)^6$ chiral and antichiral algebras by even permutations of the last five $\widehat{su}(2)_1$ factors. The quotient $G_{\GTVW}/(\ZZ_2^8\rtimes \ZZ_2^4)\cong A_5$ can be lifted to a subgroup of $G_{\GTVW}$ isomorphic to $A_5$ (i.e. $G_{\GTVW}$ is a split extension of $A_5$ by $\ZZ_2^8\rtimes \ZZ_2^4$). In section \ref{s:S3orbGTVW}, the results of sections \ref{s:orbSVOAnonholom} and \ref{s:A5S3} will be applied to this $A_5$ subgroup of $G_{\GTVW}$, and in section \ref{s:K3matching}, we demonstrate an explicit matching of defect twining genera with a corresponding category of defects in $V^{f\natural}$.

\subsection{The $S_3$ orbifold of the GTVW model}\label{s:S3orbGTVW}

Let $\calC\equiv \mathcal{C}_{\GTVW}$ be the K3 NLSM in \cite{Gaberdiel:2013psa}. We consider the subgroup
\be G:=\langle g_2,g_3,g_5\rangle \cong A_5
\ee
of $\ZZ_2^8\rtimes M_{20}$ generated by elements $g_2$, $g_3$, and $g_5$ acting on the $\widehat{su}(2)_1$ factors of both the holomorphic and antiholomorphic  $(\widehat{su}(2)_1)^6$ algebra by permutations
\be g_2\equiv (34)(56)\ ,\qquad  g_3\equiv (265)\ ,\qquad g_5\equiv (25463)\ .
\ee More precisely, $g_2$ is a pure permutation, while $g_3$ and $g_5$ are obtained by composing the permutation above followed by a $SU(2)^6\times SU(2)^6$ transformation, acting diagonally both in for the holomorphic and antiholomorphic sector by $(1,1,\Omega,\Omega^\dag,1,1)\in SU(2)^6$. Here, 
\be \Omega=\begin{pmatrix}
 \frac{1-i}{2} & \frac{1+i}{2}\\
 -\frac{1-i}{2} & \frac{1+i}{2}
 \end{pmatrix}\ ,\qquad \Omega^{\dag}=\begin{pmatrix}
 \frac{1+i}{2} & \frac{1+i}{2}\\
 \frac{1-i}{2} & \frac{1-i}{2}
 \end{pmatrix}
 \ee satisfy $\Omega^3=-1$, so that $\Omega^6=1$ and $\Omega^\dag=-\Omega^2$ (see \cite{Harvey:2020jvu,Angius:2024evd}).

The group $G$ is non--anomalous so the orbifold by any subgroup is a consistent K3 NLSM.\footnote{As will be clear from the calculation of the $g$-twisted partition functions below, none of the NSNS $g$-twisted sectors contains holomorphic fields of spin $1/2$, and this excludes that the orbifold of $\calC_{\GTVW}$ by any subgroup of $G$ is a sigma model on $T^4$. }  In particular, we consider the orbifold $\mathcal{C}'=\mathcal{C}_{\GTVW}/H$, where $H\subset G$ is the subgroup 
\be H:=\langle g_2,g_3\rangle\cong S_3\ , \ee
generated by $g_2$ and $g_3$.
Our goal is to show that the orbifold theory $\mathcal{C}'$ admits a set of non--invertible defects corresponding to the double cosets $H\backslash G/H$ as described in section \ref{s:A5S3}. In the following section, we will prove that the corresponding twining genera coincide with the ones obtained by the analogous construction for $V^{f\natural}$.

In order to apply the general formulae of section \ref{s:orbifolds}, we need to compute both the $g$-twisted $h$-twined partition functions $\CZ_g^h$ and elliptic genera $\phi_g^h$ for all commuting pairs $g,h\in G$. It is convenient to define the flavoured NSNS $g$-twisted  $h$-twined partition functions which keep track of the eigenvalues of $J_0^3$ and $\bar J_0^3$,  i.e. the zero modes of a holomorphic and an anti-holomorphic current in the $\CN=4$ superconformal algebra,
\be Z^{h,\pm}_{g,\NS}(\tau,z,\bar z)=\Tr_{\Hh_{g,\NSNS}}(hq^{L_0-\frac{c}{24}}\bar q^{\bar L_0-\frac{\bar c}{24}}  e^{2\pi i zJ_0^3} e^{-2\pi i\bar z\bar J_0^3}(\pm 1)^{F+\bar F})\ ,\ee
where 
$\Hh_{g,\NSNS}$ is the  $\NSNS$ $g$-twisted sector. In the following, we will often oscillate between considering $z$ and $\bar z$ as complex conjugate of each other or as independent complex variables; hopefully, this should be clear from the context.   The corresponding RR sector flavoured partition functions can be  obtained from the NSNS ones using the spectral flow isomorphism that commutes with $G$. The components of the `unflavoured' partition functions $\CZ_g^h(\calC,\tau)=(Z_{g,\NS}^{h,+}(\tau),Z_{g,\NS}^{h,-}(\tau),Z_{g,\R}^{h,+}(\tau),Z_{g,\R}^{h,-}(\tau))^t$ in eqs.\eqref{eq:TwistTwineNSK3} and \eqref{eq:TwistTwineRK3} can be obtained from the flavoured ones  by setting $z=\bar z=0$, while the elliptic genera of eq.\eqref{K3twisttwin} are given by setting $\bar z=0$ in the fourth component $\phi_g^h(\calC,\tau,z)=Z_{g,\R}^{h,-}(\tau,z,0)$.

Let us start with the untwisted ($g=e$) $h$-twined partition functions.  They only depend on the $G$-conjugacy class $[h]\in \{1A,2A, 3A, 5A,5B\}$, and from the explicit description of the action \cite{Gaberdiel:2013psa,Harvey:2020jvu,Angius:2024evd} it follows that:
\begin{equation}
	\begin{split}
		 Z^{e,\pm}_{e,\NS}=&  \frac{1}{\vert \eta (\tau) \vert^{12}} \left[ \Big\vert \theta_3 (2 \tau, 2z) \theta_3 (2 \tau)^5 \pm \theta_2 (2 \tau, 2z) \theta_2 (2 \tau)^5 \Big\vert^2 +\right. \\ 
		& \left. + 5 \Big\vert \theta_3 (2 \tau, 2 z) \theta_3 (2 \tau) \theta_2 (2 \tau)^4 \pm \theta_2 (2 \tau, 2z) \theta_2 (2 \tau) \theta_3 (2 \tau)^4 \Big\vert^2 + \right. \\
		& \left. +  10 \Big\vert \theta_3 (2 \tau, 2 z) \theta_3 (2 \tau)^3 \theta_2 (2 \tau)^2 \pm \theta_2 (2 \tau, 2z) \theta_2 (2 \tau)^3 \theta_3 (2 \tau)^2 \Big\vert^2 \right]\ , \\
	\end{split}
	\label{Zee_theta}
\end{equation}
\begin{equation}
	\begin{split}
		Z^{g_2,\pm}_{e,\NS}  =&   \frac{1}{\vert \eta (\tau) \vert^4 \vert \eta (2 \tau) \vert^4}  \left\lbrace \Big\vert \theta_3 (2 \tau , 2z) \theta_3 (2 \tau) \theta_3 (4 \tau)^2 \pm \theta_2 (2 \tau, 2 z) \theta_2 (2 \tau) \theta_2 (4 \tau)^2 \Big\vert^2 + \right. \\
		& \left.+ \Big\vert \theta_3 (2 \tau, 2 z) \theta_3 (2 \tau) \theta_2 (4 \tau)^2 \pm \theta_2 (2 \tau,  2z) \theta_2 (2 \tau) \theta_3 (4 \tau)^2   \Big\vert^2 + \right. \\
		& \left. +2 \Big\vert \theta_3 (2 \tau,  2z) \theta_3 (2 \tau) \theta_3 (4 \tau) \theta_2 (4 \tau) \pm \theta_2 (2 \tau, 2 z) \theta_2 (2 \tau) \theta_2 (4 \tau) \theta_3 (4 \tau) \Big\vert^2  \right\rbrace\ , \\
	\end{split}
	\label{g2_twining_theta}
\end{equation}

\begin{equation}
		\begin{split}
		Z^{g_3,\pm}_{e,\NS} =  & \frac{1}{\vert \eta (3 \tau) \vert^2 \vert \eta (\tau) \vert^{6}} \left\lbrace  \Big\vert \theta_3 (2 \tau, 2z) \theta_3 (6 \tau) \theta_3 \left( 2 \tau, \frac{1}{3} \right)^2 \pm \theta_2 (2 \tau, 2z) \theta_2 (6 \tau) \theta_2 \left( 2 \tau, \frac{1}{3} \right)^2 \Big\vert^2 +\right. \\
		& \left.+2 \Big\vert \theta_3 (2 \tau, 2z) \theta_2 (6 \tau) \theta_2 \left( 2 \tau, \frac{1}{3} \right) \theta_3 \left( 2 \tau, \frac{1}{3} \right) \pm\theta_2 (2 \tau, 2z) \theta_3 (6 \tau) \theta_3 \left( 2 \tau, \frac{1}{3} \right) \theta_2 \left( 2 \tau, \frac{1}{3} \right)  \Big\vert^2 +\right. \\ 
		& \left. +  \Big\vert \theta_3 (2 \tau, 2z) \theta_3 (6 \tau) \theta_2 \left( 2 \tau, \frac{1}{3} \right)^2 \pm \theta_2 (2 \tau, 2z) \theta_2 (6 \tau) \theta_3 \left( 2 \tau, \frac{1}{3} \right)^2 \Big\vert^2 \right\rbrace\ , 
	\end{split}\label{g3_twined_theta}
\end{equation}
\begin{equation}
		Z^{g_5,\pm}_{e,\NS} =  \frac{ \Big\vert \theta_3 (2 \tau , 2z) \theta_3 (10 \tau)  \pm \theta_2 (2 \tau, 2 z) \theta_2 (10 \tau) \Big\vert^2 }{ \vert \eta ( \tau) \vert^2\vert \eta (5 \tau) \vert^2}  \ .
	\label{g5_twining_theta}
\end{equation}
Note that all of the twining partition functions $Z^{g,\pm}_{e,\NS}$, $g\in G$ have the form\footnote{This is true because the elements of $G\cong A_5$ act in the same way on the holomorphic and anti-holomorphic fields. Note, however, that for generic elements of the symmetry group $\ZZ_2^8\rtimes M_{20}$ of $\calC_{\GTVW}$ this formula may not hold.}
\be\label{gtwinNSNS} Z^{g,\pm}_{e,\NS} (\tau,z,\bar z)=\sum_i\left\vert \ch_0(\tau,z)f^g_{i,b}(\tau)\pm \ch_1(\tau,z)f^g_{i,f}(\tau)\right\vert^2\ ,\ee where
\be \ch_0(\tau,z)=\frac{\theta_3 (2 \tau, 2z)}{\eta(\tau)}\ ,\qquad \ch_1(\tau,z)=\frac{\theta_2 (2 \tau, 2z)}{\eta(\tau)}
\ee  are the characters of the two irreducible unitary $\widehat{su}(2)_1$ representations $[0]$ and $[1]$
and $f^g_{1,b}(\tau)$, $f^g_{1,f}(\tau)$, $f^g_{2,b}(\tau)$, $f^g_{2,f}(\tau)$, $\ldots$ are suitable functions. 
Noting that the S-transformation of the $\widehat{su}(2)_1$ characters is given by the matrix
\be\nonumber
{1\over \sqrt 2} \begin{pmatrix} 1 & 1 \\ 1 & -1\end{pmatrix}
\ee
acting on $(\ch_0, \ch_1)^t$, we get the following $g$-twisted NSNS characters $Z^{e,\pm}_{g,\NS}$
\be \label{gtwistNSNS} Z^{e,\pm}_{g,\NS}(\tau,z,\bar z)=\sum_i\left\vert \frac{\theta_3 (2 \tau, 2z)}{\eta(\tau)}f_{g,i,b}(\tau)\pm \frac{\theta_2 (2 \tau, 2z)}{\eta(\tau)}f_{g,i,f}(\tau)\right\vert^2\ ,
\ee where
\begin{align*}
	f_{g,i,b}(\tau)&=\frac{f^g_{i,b}(-1/\tau)+f^g_{i,f}(-1/\tau)}{\sqrt{2}}, & f_{g,i,f}(\tau)&=\frac{f^g_{i,b}(-1/\tau)-f^g_{i,f}(-1/\tau)}{\sqrt{2}}.
\end{align*} Here, we used the fact that the S-transformation of the NSNS partition functions without $(-1)^{F+\bar F}$ insertion (i.e., the $(\NSNS,+)$  spin structure on the torus) is still a partition function with the same torus spin structure $(\NSNS,+)$. Furthermore, in these theories one has that $(-1)^{F+\bar F}=(-1)^{J_0^3+\bar J_0^3}$, and $(-1)^{J_0^3}$ acts by $(-1)^r$ on the $\widehat{su}(2)_1$ representation $[r]$, $r\in \{0,1\}$.

More generally, the $g$-twisted $g^k$-twined partition function $Z^{g^k,\pm}_{g,\NS}$ is given by\footnote{Note that the partition function $\CZ_e^e(\calC,\tau)$ of a non-linear sigma model on K3 transforms in a different $SL(2,\ZZ)$-representation than the one on $V^{f\natural}$, and in particular $\rho(T)$ in \eqref{modularspin} should be replaced by $\left(\begin{smallmatrix}
    0 & 1 & 0 & 0\\ 1 & 0 & 0 & 0\\ 0 & 0 & 1 & 0\\ 0 & 0& 0& 1
\end{smallmatrix}\right)$. The reason is that the factor $e^{2\pi i\frac{c-\bar c}{24}}$ equals $1$ for K3 models and $-1$ for $V^{f\natural}$. }
\be\label{ggk} Z^{g^k,\pm}_{g,\NS}(\tau,z,\bar z)=\sum_i\left\vert \frac{\theta_3 (2 \tau, 2z)}{\eta(\tau)}f_{g,i,b}(\tau+k)\pm \frac{\theta_2 (2 \tau, 2z)}{\eta(\tau)}f_{g,i,f}(\tau+k)\right\vert^2\ .
\ee These formulae provide the NSNS partition functions $Z^{h,\pm}_{g,\NS}$ for almost all commuting pairs $g,h\in G$. The only exceptions are the cases where $g$ is in class 2A, for example $g=g_2$, where the centralizer of $g_2$ is $C_G(g_2)=\langle g_2,k_2\rangle \cong\ZZ_2\times\ZZ_2$, for some suitable $k_2\in G$ in class 2A, so that one can have $Z^{h,\pm}_{g_2,\NS}$ with $h=k_2$ or $h=k_2g_2$. It turns out that there is a single twisted-twining partition function to consider, namely
\be\label{Zg2p2}
	Z^{k_2,\pm}_{g_2,\NS}(\tau,z,\bar z)=Z^{g_2k_2,\pm}_{g_2,\NS}(\tau,z,\bar z)=\frac{|\theta_3(2\tau,2z)\theta_3(2\tau)\pm\theta_2(2\tau,2z)\theta_2(2\tau)|^2}{|\eta(\tau)|^6}
	\left(|\theta_3(2\tau)|^2+|\theta_2(2\tau)|^2\right)\ .
\ee The computation is in appendix \ref{a:Zg2p2}.

The twisted--twining partition functions in the RR sector are obtained from the NSNS expressions \eqref{gtwinNSNS}, \eqref{gtwistNSNS}, \eqref{ggk}, and \eqref{Zg2p2} by replacing the $z$-dependent factors according to the rules
\be \theta_3 (2 \tau , 2z) \to \pm \theta_2 (2 \tau , 2z)\ ,\qquad \pm\theta_2 (2 \tau , 2z) \to + \theta_3 (2 \tau , 2z)\ .
\ee This is because spectral flow between the NSNS and RR sectors simply exchanges the $\widehat{su}(2)_1$ representations $[0]\leftrightarrow [1]$ on both the holomorphic and the antiholomorphic side. 
For example, 
\be\label{ggkRR} Z^{g^k,\pm}_{g,\R}(\tau,z,\bar z)=\sum_i\left\vert \frac{\theta_3 (2 \tau, 2z)}{\eta(\tau)}f_{g,i,f}(\tau+k)\pm \frac{\theta_2 (2 \tau, 2z)}{\eta(\tau)}f_{g,i,b}(\tau+k)\right\vert^2\ ,
\ee
and
\be\label{g2p2RR}
Z^{k_2,\pm}_{g_2,\R}(\tau,z,\bar z)=\frac{|\theta_3(2\tau,2z)\theta_2(2\tau)\pm \theta_2(2\tau,2z)\theta_3(2\tau)|^2}{|\eta(\tau)|^6}
\left(|\theta_3(2\tau)|^2+|\theta_2(2\tau)|^2\right)\ .
\ee Eqs.\eqref{ggk},\eqref{Zg2p2},\eqref{ggkRR}, and \eqref{g2p2RR} provide all components of the vector-valued partition functions $\CZ_g^h(\tau)$ (by setting $z=\bar z=0$) and the twisted--twining genera $\phi_g^h(\tau,z)$ (upon setting $\bar z=0$ in the component with  the $(\HRR,-)$ spin structure) for all nontrivial commuting pairs $g,h\in G$.

Let us now consider the orbifold K3 sigma model $\calC':=\calC/H$. Using eqs.\eqref{stwinHgH} and \eqref{stwistHgH}, we can compute  the $\CL$-twisted and $\CL$-twined partition functions \be \CZ'_{\CL}(\tau):= \CZ_{\CL}(\calC',\tau)\ ,\qquad  \CZ'^{\CL}(\tau):= \CZ^{\CL}(\calC',\tau)\ ,\ee  for all topological defects $\CL$ of $\calC'$ that preserve the fields in the $G$-invariant subspace $(\Hh_{\calC})^G\subset \Hh_{\calC'}$.  In particular, the partition function $\CZ'(\tau):= \CZ(\calC',\tau)=\CZ^e_e(\calC',\tau)$ of $\calC'$ is given by 
\begin{align*} \CZ'=\frac{1}{6}(\CZ^e_e+3\CZ_e^{g_2}+2\CZ_e^{g_3})+\frac{1}{2} (\CZ_{g_2}^{e}+\CZ_{g_2}^{g_2})+\frac{1}{3}(\CZ_{g_3}^{e}+\CZ_{g_3}^{g_3}+\CZ_{g_3}^{g_3^2})\ ,
\end{align*}
where here and in the discussion below we suppress the dependence on $\tau$ for simplicity.
Using the formulas of section \ref{s:A5S3}, we get the following $\CL$-twined partition functions for the five non-trivial simple defects listed in the table in \eqref{eq:6defects}
\begin{subequations}\label{GTVWtwin}
\begin{align}\label{twin1}
    \CZ'^{\CL_{\rho_1}}&=\frac{1}{6}(\CZ^e_e+3\CZ_e^{g_2}+2\CZ_e^{g_3})-\frac{1}{2} (\CZ_{g_2}^{e}+\CZ_{g_2}^{g_2})+\frac{1}{3}(\CZ_{g_3}^{e}+\CZ_{g_3}^{g_3}+\CZ_{g_3}^{g_3^2})\\
    \label{twin2}\CZ'^{\CL_{\rho_2}}&=\frac{1}{3}(\CZ^e_e+3\CZ_e^{g_2}+2\CZ_e^{g_3})-\frac{1}{3}(\CZ_{g_3}^{e}+\CZ_{g_3}^{g_3}+\CZ_{g_3}^{g_3^2})\\
    \label{twin3}\CZ'^{X}&=\CZ^{g_2}_e+3\CZ_e^{g_3}+2\CZ^{g_5}_{e}\\
    \label{twin4}\CZ'^{Y_{\pm}}&=\CZ^{g_2}_e+2\CZ_e^{g_5}\pm \CZ_{g_2}^{k_2}\ ,
\end{align}\end{subequations} while the $\CL$-twisted partition functions are
\begin{subequations}\label{GTVWtwist}\begin{align}
    \label{twist1}\CZ'_{\CL_{\rho_1}}&=\frac{1}{6}(\CZ^e_e+3\CZ^e_{g_2}+2\CZ^e_{g_3})-\frac{1}{2} (\CZ^{g_2}_{e}+\CZ_{g_2}^{g_2})+\frac{1}{3}(\CZ^{g_3}_{e}+\CZ^{g_3}_{g_3}+\CZ^{g_3}_{g_3^2})\\
    &=\frac{1}{6}(\CZ^e_e-3\CZ_e^{g_2}+2\CZ_e^{g_3})+\frac{1}{2} (\CZ_{g_2}^{e}-\CZ_{g_2}^{g_2})+\frac{1}{3}(\CZ_{g_3}^{e}+\CZ^{g_3}_{g_3}+\CZ_{g_3}^{g_3^2})\notag\\
    \label{twist2}\CZ'_{\CL_{\rho_2}}&=\frac{1}{3}(\CZ^e_e+3\CZ^e_{g_2}+2\CZ^e_{g_3})-\frac{1}{3}(\CZ^{g_3}_{e}+\CZ^{g_3}_{g_3}+\CZ^{g_3}_{g_3^2})\\
    &=\frac{1}{3}(\CZ^e_e-\CZ_e^{g_3})+\CZ^e_{g_2}+\frac{1}{3}(2\CZ_{g_3}^{e}-\CZ^{g_3}_{g_3}-\CZ_{g_3}^{g_3^2})\notag\\
    \label{twist3}\CZ'_{X}&=\CZ_{g_2}^e+3\CZ^e_{g_3}+2\CZ_{g_5}^{e}\\
    \label{twist4}\CZ'_{Y_{\pm}}&=\CZ_{g_2}^e+2\CZ^e_{g_5}\pm \CZ^{g_2}_{k_2}=\CZ_{g_2}^e\pm \CZ^{k_2}_{g_2}+2\CZ^e_{g_5}\ .
\end{align}\end{subequations} where we used the identities $\CZ_{g_3^2}^{g_3}=\CZ_{g_3}^{g_3^2}$ and $\CZ^{g_2}_{k_2}=\CZ^{g_2}_{k_2}$, which follow because $g_3,g_3^2$ are in the same $G$-conjugacy class 3A and $g_2,k_2$ are in the same $G$-conjugacy class 2A.
As a consistency check, one can show that all the $\CL$-twisted partition functions $\CZ'_{\CL}$ have integral Fourier coefficients, and that such coefficients are non-negative for the $(\NSNS,+)$ and $(\HRR,+)$ spin structures, as expected for the partition function of the $\CL$-twisted sector $\Hh'_{\CL}$. This is obvious for $\CZ'_{X}$, being a linear combination of partition functions for twisted sectors of $\calC$ with non-negative integral coefficients. As for the other cases, we note that all linear combinations
\be \frac{1}{6}(\CZ^e_e-3\CZ_e^{g_2}+2\CZ_e^{g_3})\qquad \frac{1}{2} (\CZ_{g_2}^{e}-\CZ_{g_2}^{g_2})\qquad \frac{1}{3}(\CZ^{g_3}_{e}+\CZ^{g_3}_{g_3}+\CZ^{g_3}_{g_3^2})
\ee
\be \frac{1}{3}(\CZ^e_e-\CZ_e^{g_3})\qquad \frac{1}{2}(\CZ_{g_2}^e\pm \CZ^{k_2}_{g_2})
\ee
\be \frac{1}{3}(2\CZ_{g_3}^{e}-\CZ^{g_3}_{g_3}-\CZ_{g_3}^{g_3^2})=\frac{1}{3}(\CZ_{g_3}^{e}+\omega\CZ^{g_3}_{g_3}+\omega^2\CZ_{g_3}^{g_3^2})+\frac{1}{3}(\CZ_{g_3}^{e}+\omega^2\CZ^{g_3}_{g_3}+\omega\CZ_{g_3}^{g_3^2})
\ee where $\omega=e^{\frac{2\pi i}{3}}$, correspond to partition functions for some subspace of states in a $g$-twisted sector, transforming in a particular representation of $C_G(g)$.

Let us now consider the $\CL$-twined and $\CL$-twisted genera $\phi'^{\CL}$ and $\phi'_\CL$ in $\calC'$. The formulae have exactly the same structures as for $\CZ'^\CL$ and $\CZ'_\CL$, but there are some simplifications. In particular, by setting $\bar z=0$ in the expression \eqref{g2p2RR} for $Z^{k_2,-}_{g_2,\R}(\tau,z,\bar z)$ one gets
\be \phi_{g_2}^{k_2}(\tau,z)=Z^{k_2,-}_{g_2,\R}(\tau,z,0)=0\ ,
\ee because of an overall factor $\lim_{\bar z\to 0}(\overline{\theta_2(2\tau,2z)\theta_3(2\tau)}-\overline{\theta_3(2\tau,2z)\theta_2(2\tau)})=0$. Furthermore, we have the identities
\be\label{simpleg2} \frac{1}{2}(\phi^e_{g_2}(\tau,z)+\phi_{g_2}^{g_2}(\tau,z))=\frac{1}{2}(\phi_e^e(\tau,z)-\phi_e^{g_2}(\tau,z))
\ee and 
\be\label{simpleg3} \frac{1}{3}(\phi_{g_3}(\tau,z)+\phi_{g_3}^{g_3}(\tau,z)+\phi_{g_3}^{g_3^2}(\tau,z))=\frac{1}{3}(\phi_e^e(\tau,z)-\phi_e^{g_3}(\tau,z))\ .
\ee  These equations hold because the orbifold $\calC/\langle g\rangle$ of $\calC$ by a symmetry $g=g_2$ or $g=g_3$ is still a consistent K3 NLSM, and therefore the elliptic genus of the orbifold theory must coincide with the elliptic genus of K3, i.e  \be \frac{1}{N}\sum_{j,k=0}^{N-1}\phi_{g^j}^{g^k}=\phi_e^e(\tau,z)\ ,\qquad N=2,3\ .\ee

Therefore, the $\CL$-twined genera ${\phi'}^{\CL}$ in $\calC'$ are simply given by 
\be\label{twinCprime1} {\phi'}^{\CI}=\phi\ ,\qquad  {\phi'}^{\CL_{\rho_1}}=\phi_e^{g_2}\ ,\qquad {\phi'}^{\CL_{\rho_2}}=\phi_e^{g_2}+\phi_e^{g_3}\ ,
\ee
\be\label{twinCprime2} {\phi'}^{Y_+}={\phi'}^{Y_-}=\phi_e^{g_2}+2\phi_e^{g_5}\ ,\qquad {\phi'}^{X}=\phi_e^{g_2}+3\phi_e^{g_3}+2\phi_e^{g_5}\ ,
\ee
where we have suppressed the $(\tau,z)$ dependence for simplicity.
Let us comment on this result. The first line consists of the twining genera for the three Wilson lines of the orbifold group $H\cong S_3$, i.e. the (non-invertible) quantum symmetry. Because the identity and $\CL_{\rho_1}$ are invertible defects,  their $\CL$-twined genera coincide precisely with the twining genera for elements in class 1A and 2A of the Conway group. The $\CL$-twined genera for non-invertible defects $\CL$ are given in all cases by non-negative integral linear combinations of the standard $Co_0$ twining genera described in section \ref{s:Conwaytwining}. This fact ensures that the $\CL$-twisted genera $\phi'_\CL$ automatically satisfy the expected integrality conditions. 

\subsection{Comparison with non-invertible defects in $V^{f\natural}$}\label{s:K3matching}

Let us now consider the SVOA $V^{f\natural}$, and let us consider the groups of symmetries $H\cong S_3$ and $G\cong A_5$ that preserve the $\CN=1$ supercurrent and such that
\be H\subset G\subset \ZZ_2^{8}\rtimes M_{20} \subset Co_0\cong \Aut_\tau(V^{f\natural})\ .
\ee  Here, the subgroup $\ZZ_2^8\rtimes M_{20}\subset Co_0$ is isomorphic to the group $G_\GTVW$ of symmetries of the GTVW model that preserve the $\CN=(4,4)$ superconformal algebra and the spectral flow as described above. As a subgroup of $Co_0$, it is unique up to conjugation \cite{HohnMason2016}.

The group $G$ is anomaly free, so that the orbifold by any of its subgroups, in particular by $H$, is consistent. In particular, one easily checks that $V':=V^{f\natural}/H$ is isomorphic to $V^{f\natural}$
\be V'\cong V^{f\natural}\ .
\ee While the parent and the orbifold theory are isomorphic, it is convenient to continue using different symbols $V,V'$ to denote $V^{f\natural}$ and $V^{f\natural}/H$, respectively. 
If we denote  the twisted-twining partition functions in the parent theory $V$ as
\be \CZ_g^h(\tau):=\CZ_g^h(V,\tau)\ ,\qquad g,h\in G\ ,
\ee  and the $\CL$-twisted and $\CL$-twined partition functions in the orbifold theory $V'=V/H$ as,
\be \CZ'_\CL(\tau):=\CZ_\CL (V',\tau)\ ,\qquad \CZ'^\CL(\tau):=\CZ^\CL (V',\tau)\ ,
\ee  where $\CL$ is any topological defect preserving the $G$-invariant subalgebra $V^G\subset V'$, then formulae \eqref{GTVWtwin} and \eqref{GTVWtwist} hold in this case as well.

Therefore, one just needs to compute  $\CZ_g^h(V, \tau)$ for all commuting pairs $g,h\in G$. The $g$-twined partition functions $\CZ^g_e(V,\tau)$ for all $g\in G$ depend only on the $Co_0$-conjugacy class of $g$, which in turn is fully determined by its set of eigenvalues on the $24$-dimensional representation of $Co_0$, or in terms of the Frame shape (see section \ref{s:Conwaytwining}).  In particular, the classes $1A$, $2A$, $3A$ and $5AB$ of $G$ correspond, respectively, to Frame shapes $1^{24}$, $1^82^8$, $1^63^6$, $1^45^4$ (for both classes of order $5$), and for all $g\in G$, $\CZ_e^g(V,\tau)$ is straightforwardly computed given this data using equations \eqref{eq:SymTwineNS} and \eqref{eq:SymTwineR}. 
By applying modular transformations to all $g$-twined functions, one can obtain all twisted-twining partition functions of the form $\CZ_g^{g^k}$, for all $g\in G$. The only missing ingredient is the $g_2$-twisted $k_2$-twined partition function $\CZ_{g_2}^{k_2}$ for $g_2$ and $k_2$ in class 2A of $G$, with $\langle g_2,k_2\rangle\cong \ZZ_2\times\ZZ_2$. 

Let us now prove that all components of $\CZ_{g_2}^{k_2}(V,\tau)$ vanish, i.e.
\be \CZ_{g_2}^{k_2}(V,\tau)=(0,0,0,0)^t\ .
\ee This is most easily proved by considering the corresponding symmetries in the theory $F(24)$ of $24$ free fermions, which is related to $V^{f\natural}$ by exchanging the NS odd and the Ramond even sector of the two theories. The $24$ free fermions generating $F(24)$ correspond to the Ramond ground states in $V^{f\natural}$, so that the action of $G$ on such fields is the the same as on ${}^\RR V^{f\natural}(1/2)$. 
Consider two commuting elements $g_2,k_2$ in class 2A of $G$. Then, the product $g_2k_2$ has also order two, and therefore it is in the same $G$-conjugacy class. Thus, when acting on the $24$-dimensional representation of $Co_0$, $g_2$, $k_2$ and $g_2k_2$ have all Frame shape $1^82^8$.
This means that if we simultaneously diagonalize $g_2$ and $k_2$ in the space of $24$ free fermions, the pairs of $(g_2,k_2)$-eigenvalues have multiplicities 
\be 12\times (1,1)\ ,\qquad 4\times (1,-1)\ ,\qquad 4\times (-1,1)\ ,\qquad 4\times (-1,-1)\ .
\ee Now, a free fermion $\psi(z)$ of $F(24)$ that is simultaneous $(g_2,k_2)$-eigenvector with eigenvalues $(g_2,k_2)=(-1,-1)$, when acting on the $g_2$-twisted NS sector of $F(24)$,  has integral mode expansion, and in particular has a zero mode $\psi_0$ with $k_2(\psi_0)=-\psi_0$. Similarly, a fermion with $(g_2,k_2)=(-1,1)$, when acting on the $g_2$-twisted NS sector of $F(24)$, has a zero mode $\psi_0$ with $(-1)^Fk_2(\psi_0)=-\psi_0$. The existence of such zero modes implies that the $g_2$-twisted $k_2$-twined partition functions with either $(\NS,+)$ or $(\NS,-)$ spin structure vanish. Similarly, in the $g_2$-twisted Ramond sector of $F(24)$, a fermion with eigenvalues $(g_2,k_2)=(1,-1)$ has a zero mode $\psi_0$ with $k_2(\psi_0)=-\psi_0$, while a fermion with $(g_2,k_2)=(1,1)$ has a zero mode $\psi_0$ with $(-1)^Fk_2(\psi_0)=-\psi_0$. It follows that  the Ramond $g_2$-twisted $k_2$-twined partition functions of $F(24)$ also vanish. Because all components of the $F(24)$ partition functions vanish, the same must be true for the corresponding partition functions in $V^{f\natural}$.

Furthermore, the analogues of the identities \eqref{simpleg2} and \eqref{simpleg3} hold for the partition functions of $V$, namely
\be \frac{1}{2}(\CZ_{g_2}+\CZ_{g_2}^{g_2})=\frac{1}{2}(\CZ-\CZ^{g_2})
\ee and 
\be \frac{1}{3}(\CZ_{g_3}+\CZ_{g_3}^{g_3}+\CZ_{g_3}^{g_3^2})=\frac{1}{3}(\CZ-\CZ^{g_3})\ .
\ee  Indeed, because the orbifold $V^{f\natural}/\langle g\rangle$ of $V^{f\natural}$ by either $g=g_2$ or $g=g_3$ is isomorphic to $V^{f\natural}$ itself, then the orbifold partition function must be equal to the partition function of the parent theory, i.e.
\be \frac{1}{N}\sum_{j,k=0}^{N-1}\CZ_{g^j}^{g^k}(\tau)=\CZ_e^e(\tau)\ ,\qquad N=2,3\ ,\ee
from which the previous identities follow. 

Thus, the twining partition functions ${\CZ'}^\CL$ are given by
\be \CZ'=\CZ\ ,\qquad {\CZ'}^{\CL_{\rho_1}}=\CZ^{g_2}\ ,\qquad {\CZ'}^{\CL_{\rho_2}}=\CZ^{g_2}+\CZ^{g_3}\ ,
\ee
\be {\CZ'}^{Y_\pm}=\CZ^{g_2}+2\CZ^{g_5}\ ,\qquad {\CZ'}^{X}=\CZ^{g_2}+3\CZ^{g_3}+2\CZ^{g_5}\ .
\ee
Similar formulas hold for the $\CL$-twisted and $\CL$-twined genera $\phi'^\CL(V',\tau,z)$ and $\phi'_\CL(V',\tau,z)$ of $V'$. In particular, it is still true that
\be \phi_{g_2}^{k_2}(V^{f\natural},\tau,z)=0\ ,
\ee because the fermion zero modes on the $g_2$-twisted Ramond and NS sector can always be chosen so as to be neutral with respect to the $\widehat{su}(2)_1$ current $J_0^3$. 

Furthermore, it is known  that for every $g\in G$, the (untwisted) $g$-twined elliptic genus in the GTVW  model $\calC$ equals the corresponding $g$-twined genus in the model $V^{f\natural}$
\be \phi^g_e(\calC,\tau,z)=\phi^g_e(V^{f\natural},\tau,z)\ ,\qquad \forall g\in G\ .
\ee
This is clear because all conjugacy classes in $G$ correspond to geometric symmetries of K3 surfaces, whose twining functions are always reproduced by the construction in \cite{Duncan:2015xoa}; see \cite{Cheng:2016org} for a discussion.
Comparing with \eqref{twinCprime1} and \eqref{twinCprime2}, we find
\be \phi(V',\tau,z)=\phi(V,\tau,z)=\phi(\calC',\tau,z)\ ,\qquad {\phi}^{\CL_{\rho_1}}(V',\tau,z)=\phi^{g_2}_e(V,\tau,z)=\phi^{\CL_{\rho_1}}(\calC',\tau,z)\ ,\ee
\be\phi^{\CL_{\rho_2}}(V',\tau,z)=\phi^{g_2}_e(V,\tau,z)+\phi^{g_3}_e(V,\tau,z)=\phi^{\CL_{\rho_2}}(\calC',\tau,z)\ ,
\ee
\be \phi^{Y_\pm}(V',\tau,z)=\phi_e^{g_2}(V,\tau,z)+2\phi_e^{g_5}(V,\tau,z)=\phi^{Y_\pm}(\calC',\tau,z)\ ,\ee
\be\phi^{X}(V',\tau,z)=\phi_e^{g_2}(V,\tau,z)+3\phi_e^{g_3}(V,\tau,z)+2\phi_e^{g_5}(V,\tau,z)=\phi^{X}(\calC',\tau,z)\ .
\ee
We conclude that the $\CL$-twined elliptic genera for all topological defects $\CL$ in the K3 model $\calC'=\calC/H$ preserving the subspace $\Hh^G$ coincide with the corresponding genera in  $V'=V/H\cong V^{f\natural}$ preserving the subalgebra $V^G$. This result provides further evidence for Conjecture \ref{conj:K3relation} in section \ref{s:topdefK3}.

\section{Discussion}\label{s:discussion}

In this article, we described properties of a tensor category $\cTop$ of topological defects of the SVOA $V^{f\natural}$ preserving the $\CN=1$ superconformal algebra, and that are `well-behaved' with respect to the fermion number (see section \ref{subsec:3.2}). In particular, we proved that if $\cTop$ contains all invertible defects $\CL_g$, $g\in \Aut_\tau(V^{f\natural})\cong Co_0$, then there is a map $\cTop\ni \CL\mapsto \rho(\CL)\in {\rm End}(\Lambda)$ assigning with each $\CL$ a $\ZZ$-linear map $\rho(\CL)$ from the Leech lattice $\Lambda$ to itself, that is a ring automorphism
\be \rho(\CL_1+\CL_2)=\rho(\CL_1)+\rho(\CL_2)\ ,\qquad \rho(\CL_1\CL_2)=\rho(\CL_1)\rho(\CL_2)\ ,
\ee and such that $\rho(\CL^*)=\rho(\CL)^t$. The map is surjective on ${\rm End}(\Lambda)$, but not injective. Moreover, in Conjecture \ref{conj:K3relation}, we proposed  a correspondence between four--plane--preserving TDLs in $\cTop$ and $\CN=(4,4)$--preserving TDLs in K3 NLSMs. We now conclude with a discussion of various aspects and consequences of our results:

\begin{enumerate}
    \item Our main theorem, Theorem \ref{th:main}, is based on \eqref{Cardy}, namely  the condition  that $\Tr_{V^{f\natural}_{tw}(1/2)}(\hat\CL g)$ is integral for all $g\in \Aut_\tau(V^{f\natural})\cong Co_0$. One can be apply a similar argument to any SVOA $V$ with $\CN=1$ superconformal symmetry  (for example, for $V=V^{fE_8}$)\footnote{In principle, one could apply this argument also to $V=F(24)$ for some choice of the $\CN=1$ superconformal algebra. However, $F(24)$ contains no Ramond ground states, so the analog of \eqref{Cardy} is an empty condition in this case.}: the analog of \eqref{Cardy} should hold for any topological defect $\CL$ in $V$ preserving the $\CN=1$ supercurrent $\tau(z)$, assuming that the invertible symmetries $g\in \Aut_\tau(V)$ are known and that they are themselves `well-behaved' with respect to the fermion number. Suppose that we can prove that the SVOA $V$ admits a certain non-invertible TDL $\CL_1$ preserving the $\CN=1$ supercurrent. Then, we can put further constraints on any other potential defect $\CL_2$  by requiring eq.\eqref{Cardy} to be satisfied not only by $\hat\CL_2$, but also by the product $\hat\CL_1\hat\CL_2$, or more generally by $\hat\CL_1^k\hat\CL_2^l$ for all powers $k,l=0,1,2,\ldots$. In the case of $V=V^{f\natural}$, these additional requirements do not lead to any new constraints. Indeed, Theorem \ref{th:LeechEndo} tells us that, for $V=V^{f\natural}$, \eqref{Cardy} is equivalent to $\hat\CL$ being a Leech lattice endomorphisms. Now, if $\hat\CL_1$ and $\hat\CL_2$ are in ${\rm End}(\Lambda)$, then all  the products $\hat\CL_1^k\hat\CL_2^l$ are also in ${\rm End}(\Lambda)$, and therefore they automatically satisfy \eqref{Cardy}. 
    \item Theorems \ref{th:LeechEndo} and \ref{th:LeechEndo}$'$ depend on special properties of the Leech lattice $\Lambda$ and its group of automorphisms $\Aut(\Lambda)\cong Co_0$, and analogous theorems might not hold for other even unimodular lattices. Consider for example the $24$-dimensional even self-dual lattice $N$ whose root system is $A_1^{24}$ \cite{ConwaySloane}. It can be described as the subset of vectors $\frac{1}{\sqrt{2}}(x_1,\ldots,x_{24})\in \RR^{24}$ where $x_1,\ldots,x_{24}\in \ZZ$ and such that the set of indices $i\in \{1,\ldots, 24\}$ for which $x_i$ is odd forms a $\CG_{24}$-set, where $\CG_{24}$ is the Golay code. The group of automorphisms of $N$ is $\Aut(N)\cong \ZZ_2^{24}\rtimes M_{24}$, where $M_{24}$ is the Mathieu group. With respect to the standard basis of $\RR^{24}$, the elements of $\ZZ_2^{24}$ are represented by diagonal matrices $\epsilon=\diag(\pm 1,\ldots, \pm 1)$ for all possible choices of signs. The general element $g\in \ZZ_2^{24}\rtimes M_{24}$ can be written as $g=\epsilon\pi$, where $\epsilon\in \ZZ_2^{24}$ and $\pi\in M_{24}$ is represented by a permutation matrix. In particular, all elements of $\Aut(N)$ are represented by integral matrices where each row and each column have a single non-zero entry equal to $\pm 1$. Take a $24\times 24$ matrix $M$ having, for example, two rows with all entries equal to $1/2$ and the other rows with all zero entries. It is clear that, for all $g\in \ZZ_2^{24}\rtimes M_{24}$,  $\Tr(gM)=\pm \frac{1}{2}\pm \frac{1}{2}\in \{-1,0,1\}$ is integral, so that the analogue of condition (a) in Theorem \ref{th:LeechEndo} is satisfied. On the other hand, $M$ cannot be an integral linear combination of elements $g_i\in \ZZ_2^{24}\rtimes M_{24}$, because its entries are not integral. Furthermore, it is easy to see that $M$ does not map vectors of $N$ to vectors of $N$. Therefore, the analogues of Theorems \ref{th:LeechEndo} and \ref{th:LeechEndo}$'$ fail for the group $\Aut(N)$.
    \item In this article, we considered only topological defects $\CL$ for which the $\CL$-twisted NS and R sectors $V_\CL$ and $V_{tw,\CL}$ have discrete $L_0$ spectrum. It may be possible to relax this condition and include twisted spaces with continuous spectrum. In this case, in order to apply our methods, it is sufficient to require that in each twisted Ramond sector: (1) there is a gap $\Delta>0$ between the conformal weight $h=\frac{c}{24}$ of the ground states and the weight $h\ge \Delta+\frac{c}{24}$ of the continuous part of the spectrum; and that (2) there is a perfect cancellation between the bosonic and fermionic densities of states at $h\ge \Delta+\frac{c}{24}$.
    \item Though we have provided both a conditional argument (appendix \ref{a:conjcons}) and a construction of concrete non-invertible TDLs (section \ref{s:K3defects}), both of which illustrate the consistency of Conjecture \ref{conj:K3relation}, it is crucial to gather more evidence in order to establish a definitive answer to Questions \ref{q1}, \ref{q2}, and \ref{q3}. So far, both of these pieces of evidence concern defects of integral quantum dimension. In the forthcoming work \cite{Angius:2025xxx}, we will explicitly construct defects in $V^{f\natural}$ of irrational quantum dimension, some of which can be identified with corresponding defects in a K3 NLSM.
    More generally, one can investigate  non-invertible TDLs which have been constructed in other K3 NLSMs, for example \cite{Cordova:2023qei, Angius:2024evd,Arias-Tamargo:2025xdd,Caldararu:2025eoj}, and attempt to find corresponding defects in $V^{f\natural}$ whose twining functions are the same (some of which will be considered in \cite{Angius:2025xxx}).  
    \item We stressed that a category $\cTop$ of topological defects of $V^{f\natural}$ preserving the $\CN=1$ superVirasoro algebra is expected to contain infinitely many simple objects (in this sense, it is a tensor category rather than a fusion category). In particular, the fact that $V^{f\natural}$ can be obtained as a $\ZZ_2$ orbifold of $V^{fE_8}$ implies that there are continuous families of topological defects that are (generically) simple. Such continuous families of defects have been observed long ago in CFTs of a single free boson on an orbifold $S^1/\ZZ_2$ \cite{Fuchs:2007tx,Becker:2017zai,Chang_2019, Thorngren:2021yso}, and the argument easily generalizes to the case of $V^{fE_8}/\ZZ_2$. 
    Indeed, the group of automorphisms $\Aut_\tau(V^{f E_8})$ preserving the $\CN=1$ supercurrent $\tau$ contains a normal subgroup $U(1)^8$, with elements $g(\vec\alpha):=e^{2\pi i \vec\alpha\cdot \vec J_0}$, where $\vec\alpha\in E_8\otimes\RR/E_8$ and $\vec J_0=(J_0^,\ldots, J_0^8)$ are the zero modes of the $8$ weight $1$ currents that are superdescendants of the $8$ free fermion fields. Let $\calR$ be the generator of the $\ZZ_2$ symmetry acting by $\calR \vec J \calR=-\vec J$ and such that $V^{f\natural}=V^{f E_8}/\langle \calR\rangle$. 
    While the continuous invertible symmetries $\CL_{g(\vec\alpha)}$ are (generically) projected out in the orbifold procedure, the superpositions $\CL_{g(\vec\alpha)}+\CL_{g(-\vec\alpha)}$ survive as (generically) simple dimension $2$ topological defects in $V^{f\natural}\cong V^{f E_8}/\langle \calR\rangle$ that preserve the supercurrent $\tau$. The associated lattice endomorphisms $\rho(\CL_{g(\vec\alpha)}+\CL_{g(-\vec\alpha)})\in {\rm End}(\Lambda)$ obviously cannot vary continuously with $\vec\alpha$, so they must be constant. The fact that $\rho(\CL_{g(\vec\alpha)}+\CL_{g(-\vec\alpha)})$ is independent of $\vec\alpha$ can also be verified by a direct calculation. This observation provides another proof that the ring homomorphism $\rho:\cTop\to {\rm End}(\Lambda)$ is not injective. We emphasize that continuous families of simple non-invertible defects, preserving the $\CN=(4,4)$ superconformal algebra and spectral flow, also appear in various K3 sigma models, and in particular  all torus orbifolds \cite{Angius:2024evd}. The existence of such continuous families on both sides of the correspondence is another non-trivial piece of evidence in favor of Conjecture \ref{conj:K3relation}.
    \item Let $\CL\in \cTop$ be a topological defect of $V^{f\natural}$ preserving a subspace $\Sigma\subseteq {}^\RR V^{f\natural}_{tw}(1/2)$ with $\dim \Sigma=d\ge 4$, i.e. such that $\hat\CL|\psi\rangle = \langle \CL\rangle |\psi\rangle$ for all $|\psi\rangle\in \Sigma$. Then, $\CL$ also preserves a subalgebra $\widehat{so}(d)_1$ of $\widehat{so}(24)_1$, so that it commutes with the group $Spin(d)\subset \Aut(V^{f\natural})$ generated by the zero modes of its currents. Furthermore, $\CL\in \cTop_\Pn$ for all $4$-dimensional subspaces $\Pn\subseteq \Sigma$. Note that the family of $4$-dimensional subspaces $\Pn\subseteq \Sigma$ is a Grassmannian of dimension $4(d-4)$. For $d>4$, the $\CL$-twined genus $\phi^\CL(V^{f\natural},\tau,z)$ does not depend on $\Pn\subseteq \Sigma$, because all $\widehat{su}(2)_1\subset \widehat{so}(4)_1\subset \widehat{so}(d)_1$ algebras that are preserved by $\CL$ are related to each other by conjugation in $Spin(d)$.\footnote{When the $\CL$-fixed subspace $\Sigma$ has dimension exactly $4$, the two $\widehat{su}(2)_1$ in $\widehat{so}(4)_1$ might not be related by any automorphism commuting with $\CL$, so that there might be two different $\CL$-twined genera. This phenomenon has been observed for invertible symmetries, see \cite{Duncan:2015xoa} and section \ref{s:Conwaytwining}.} Suppose that, for at least one choice $\Pn\subseteq \Sigma$, there is a corresponding K3 sigma model $\calC$ as predicted by Conjecture \ref{conj:K3relation}, and let $\CL':=F(\CL)\in \cTop^{K3}_{\calC}$ be the defect of $\calC$ corresponding to $\CL$. Thus, $\CL'$ fixes a $d$-dimensional subspace of $\HRR$ ground fields in $\Hh_{\HRR,gr}^{K3}$, $d\ge 4$, that contains  in particular the space $\Pi$ of spectral flow generators. Because the space of exactly marginal operators of the K3 model is related to $\Pi^\perp\cap \Hh_{\HRR,gr}^{K3}$ by spectral flow, one can conclude that $\CL'$ is a topological defect not only for $\calC$, but for a whole $4(d-4)$-dimensional family of K3 sigma models that are related to $\calC$ by exactly marginal deformations. Furthermore, for $d>4$, the $\CL'$-twined genus $\phi^{\CL'}(\calC,\tau,z)$ is invariant under such deformations.
\item Our methods allowed us to put constraints on the Grothendieck ring of the category of $\CN=1$ preserving topological defects in $V^{f\natural}$. On the other hand, given a potentially consistent fusion (sub)ring of topological defects of $V^{f\natural}$, i.e. such that there is a homomorphism to the ring ${\rm End}(\Lambda)$, it would be desirable to determine whether one or many corresponding tensor categories exist. This amounts to finding all possible solutions to the pentagon identities for the corresponding fusion matrices \cite{Etingof:2015}. Inequivalent solutions might lead to different modular properties for the twisted-twining partition functions or elliptic genera, and to different spin selection rules in the twisted sectors \cite{Chang_2019}.
\item If $g\in \Aut(V)$ is an automorphism of a holomorphic SVOA $V$ of finite order $N$, then the $g$-twined partition function $\CZ^g$, if well-defined, is expected to be a (vector-valued) modular form of level $N$, possibly up to multiplier. This means that it is invariant under the subgroup $\Gamma(N)$ of $SL_2(\ZZ)$ matrices that are equal to the identity mod $N$. Analogous modular properties are expected for the $g$-twined genus \eqref{eq:D-MCtwining} in $V^{f\natural}$, for those invertible symmetries $g$ for which it is well-defined. When $\CL\in \cTop$ is a generic topological defect preserving a possibly non-rational subSVOA of $V^{f\natural}$, one cannot expect the $\CL$-twined partition function (or genus) to be invariant under a finite index subgroup of $SL_2(\ZZ)$. This phenomenon is related to the fact that the tensor category generated by $\CL$ might contain infinitely many simple objects. The simplest example is the case of an invertible symmetry $g$ of infinite order, for example the generic element of a Lie group, for which all powers $g^n$, $n\in\ZZ$, are distinct automorphisms. Nevertheless, one can expect $\cTop$ to contain many non-invertible simple defects $\CL$ for which the preserved subSVOA is rational. It would be very interesting to understand what kind of weak Jacobi forms can be obtained as $\CL$-twining genera for such $\CL$.
\end{enumerate}

\bigskip 

\noindent {\bf Acknowledgements.} 

\noindent We thank Sebastiano Carpi, Giulio Codogni, Tiziano Gaudio, Luca Giorgetti, and Shu-heng Shao for useful comments and discussions. RV would like to thank the organizers and the participants of the Conference in Vertex Alegebras and Related Topics in Taipei, December 2024, for many stimulating discussions.
RA is supported by the ERC Starting Grant QGuide-101042568- StG 2021. RV acknowledges support from CARIPARO Foundation Grant under grant n. 68079. SG and RV were partially supported by the PRIN Project n. 2022ABPBEY.

\appendix

\section{Consistency of Conjecture \ref{conj:K3relation}}\label{a:conjcons}

As in section \ref{s:topdefK3}, let $\Pi^\natural\subset {}^\RR V^{f\natural}_{tw}(1/2)$ be a $4$-dimensional subspace, and let $\cTop_{\Pi^\natural}$ be the subcategory of $\cTop$ preserving the space of fields in $\Pi^\natural$. We will  consider the case where the category $\cTop_\Pn$ only contains  defects with integral quantum dimension, $\langle \CL\rangle\in \NN$ for all $\CL\in \cTop_\Pn$. By Corollary \ref{th:integralqdim}, this always happens when the $4$-dimensional subspace $\Pn\subset \Lambda\otimes\RR$  contains a non-zero lattice vector $0\neq \lambda\in \Lambda\cap \Pn$. Our goal is to prove that, in this case, one can find a K3 NLSM $\calC$ and an isomorphism $\varphi:V^{f\natural}_{tw}(1/2)\stackrel{\cong}{\longrightarrow} \Hh^{K3}_{\HRR,gr}$ such that for all $\CL\in \cTop_\Pn$, the endomorphism $\varphi\circ(\rho(\CL))\circ\varphi^{-1}:\Hh^{K3}_{\HRR,gr}\to \Hh^{K3}_{\HRR,gr}$ is compatible with all known properties of topological defects in $\calC$. 

Let us denote by
\be \Lambda_\Pn:=\Lambda\cap \Pn^\perp
\ee the sublattice of the Leech lattice orthogonal to $\Pn$. It is clear that $\Lambda_\Pn$ is a primitive sublattice (i.e. $\Lambda/\Lambda_\Pn$ is torsion free) with $\rk \Lambda_\Pn\le 20$. Let us show that every $\CL\in \cTop_\Pn$ preserves all vectors $v\in V^{f\natural}_{tw}(1/2)\cong \Lambda\otimes\CC$ that are orthogonal to $\Lambda_\Pn$, i.e. that \be\label{orthortho}\hat\CL(v)=\langle \CL\rangle v\ ,\qquad \forall v\in \Lambda_\Pn^\perp\subset \Lambda\otimes\CC\ .\ee This is obviously true when $\rk \Lambda_\Pn=20$, because in that case $\Lambda_\Pn^\perp=\Pn$; so, the proof is non-trivial only for $\rk\Lambda_\Pn<20$. Fix $\CL\in \cTop$ and consider 
\be (\hat\CL^*-\langle\CL\rangle)\Lambda\ .
\ee Because $\hat\CL^*\in {\rm End}(\Lambda)$ and $\langle \CL\rangle\in \ZZ$, this is a sublattice of the Leech lattice
\be\label{inclLambda} (\hat\CL^*-\langle\CL\rangle)\Lambda\subseteq \Lambda\ .
\ee On the other hand, for all $w\in \Pn\subset \Lambda\otimes\RR$ and for all $\lambda\in\Lambda$, one has 
\be (w,(\hat\CL^*-\langle\CL\rangle)\lambda)=((\hat\CL-\langle\CL\rangle)w,\lambda)=0\ ,\qquad \forall w\in\Pn,\ \lambda\in\Lambda,
\ee where we used that $\hat\CL(w)=\langle\CL\rangle w$ for all $w\in \Pn$. This implies that
\be (\hat\CL^*-\langle\CL\rangle)\Lambda\subseteq \Lambda\cap\Pn^\perp =\Lambda_\Pn\ .
\ee Passing to the orthogonal subspaces, we get the opposite inclusion
\be \Lambda_\Pn^\perp\subseteq ((\hat\CL^*-\langle\CL\rangle)\Lambda)^\perp\ , \qquad \forall \CL\in\cTop_\Pn\ .
\ee On the other hand, for all $v\in ((\hat\CL^*-\langle\CL\rangle)\Lambda)^\perp$ and for all $\lambda\in \Lambda$, we have
\be 0=\langle v,(\hat\CL^*-\langle\CL\rangle)\lambda\rangle=\langle(\hat\CL-\langle\CL\rangle)v,\lambda\rangle\ ,\qquad \forall \lambda\in \Lambda\ ,
\ee and this implies $(\hat\CL-\langle\CL\rangle)v=0$. This is true, in particular, for all $v\in \Lambda_\Pn^\perp$, and \eqref{orthortho} holds. As a consequence of \eqref{orthortho}, the category $\cTop_\Pn$ is a (not necessarily proper) subcategory of all $\cTop_{\Pn'}$ for all $4$-dimensional real subspaces $\Pn'\subset \Lambda_\Pn^\perp\cap \Lambda\otimes\RR$. 

Let us now show that the lattice $\Lambda_\Pn(-1)$, obtained from $\Lambda_\Pn\subset\Lambda$ by changing the sign of the quadratic form, can be primitively embedded in an even unimodular lattice of signature $(4,20)$, that we interpret as the lattice of D-brane charges of a K3 sigma model.  Choose a primitive sublattice $L\subset \Lambda$ of $\rk L=4$ and orthogonal to $\Lambda_\Pn$, i.e. $L\otimes \RR\subseteq \Lambda_\Pn^\perp$. Theorem 1.12.4 of \cite{nikulin} implies that all even lattices of rank $\le 4$ can be primitively and isometrically embedded in the $E_8$ lattice. Choose such an embedding $\iota: L\to E_8$ and consider the lattice $E_8\oplus_\perp \Lambda(-1)$, where  $\Lambda(-1)$ is the Leech lattice with reversed sign of the quadratic form. This is an indefinite even unimodular lattice of signature $(8,24)$. If we denote by $\mu\oplus\lambda$ the vectors in this lattice, with $\mu\in E_8$, $\lambda\in\Lambda$, then it is clear that
\be L_0:=\{\iota(\lambda)\oplus \lambda\mid \lambda\in L\subset \Lambda\}\ ,
\ee is a $4$-dimensional primitive isotropic sublattice of $E_8\oplus_\perp \Lambda(-1)$, i.e. such that all its vectors have zero square norm.  Thus, the lattice $L_0^\perp \cap (E_8\oplus_\perp \Lambda(-1))$ is a $28$-dimensional degenerate lattice that contains $L_0$ itself, while the quotient
\be \Gamma:= (L_0^\perp \cap (E_8\oplus_\perp \Lambda(-1)))/L_0\ ,
\ee is an even non-degenerate lattice of signature $(4,20)$. Furthermore, it is self-dual, see \cite{ConwaySloane} chapt.26.\footnote{Indeed, if $w_1,\ldots,w_4$ is a basis of $L_0$, then by self-duality, we can choose $w_1^*,\ldots,w_4^*\in E_8\oplus_\perp \Lambda(-1)$ such that $(w_i,w_j^*)=\delta_{ij}$, and without loss of generality, upon adding $w_i^*$ a suitable integral linear combination of $w^1,\ldots,w^4$, we can assume that $(w_i^*,w_j^*)=0$. Thus, it is clear that $w_1,\ldots,w_4,w_1^*,\ldots,w_4^*$ span an even self-dual sublattice of signature $(4,4)$ inside $(E_8\oplus_\perp \Lambda(-1))$. Now, for every element in the quotient $\Gamma$, there is a (unique) lift to $L_0^\perp \cap (E_8\oplus_\perp \Lambda(-1))$ that is orthogonal also to $w_1^*,\ldots,w_4^*$. This argument shows that $\Gamma$ is isomorphic to the orthogonal complement in the unimodular $E_8\oplus_\perp \Lambda(-1)$ of the even unimodular lattice generated by $w_1,\ldots,w_4,w_1^*,\ldots,w_4^*$, and therefore $\Gamma$ is itself unimodular.}

Being an even unimodular lattice of signature $(4,20)$, $\Gamma$ can be identified with the lattice of RR D-brane charges of a K3 NLSM. Furthermore, the lattice $\Lambda_\Pn\subseteq L^\perp\cap \Lambda$ can be primitively embedded in $\Gamma$ via
\be\label{LambdaEmbed} \Lambda_\Pn(-1) \to  (L_0^\perp \cap (E_8\oplus_\perp \Lambda(-1)))\to \Gamma\ ,
\ee  and we denote by $\Gamma_{\Pi_{K3}}\subset \Gamma$ its image, which is a negative definite primitive sublattice, and by $\Gamma^{\Pi_{K3}}:=(\Gamma_{\Pi_{K3}})^\perp\cap \Gamma$ its orthogonal complement. Then, we take $\Pi_{K3}$ any $4$-dimensional positive definite subspace in $\Gamma^{\Pi_{K3}}\otimes\RR$ that is not orthogonal to any vector in $\Gamma^{\Pi_{K3}}\otimes \bar\QQ$.

The choice of the subspace $\Pi_{K3}\subset \Gamma\otimes \RR$ uniquely determines a non-linear sigma model $\calC$ on K3, such that  $\Gamma$ is the lattice of D-brane charges, $\Gamma\otimes \RR\cong \Hh_{\HRR,gr}^{K3}$ is the space of RR ground fields and $\Pi_{K3}$ is the $4$-dimensional subspace of spectral flow operators, i.e the RR ground fields that are charged under both the holomorphic and antiholomorphic $\widehat{su}(2)_1$ subalgebras of the $\CN=(4,4)$ superconformal algebra. Then we can define the isomorphism of Hilbert spaces $\varphi:V^{f\natural}_{tw}(1/2)\to \Hh_{\HRR,gr}^{K3}$ in such a way that the lattice $\Lambda_\Pn\subset V^{f\natural}_{tw}(1/2)$ is mapped to $\varphi(\Lambda_\Pn)=\Gamma_{\Pi_{K3}}\subset \Hh_{\HRR,gr}^{K3}$ via the embedding \eqref{LambdaEmbed}, and the orthogonal subspace $\Lambda_\Pn^\perp \subset V^{f\natural}_{tw}(1/2)$ is mapped to $\Gamma_{\Pi_{K3}}^\perp\subset \Hh_{\HRR,gr}^{K3}$ in such a way that $\varphi(\Pn)=\Pi_{K3}$. Let $\CL\in \cTop_\Pn$ be a topological defect of $V^{f\natural}$, and let $\rho(\CL):\Lambda\to \Lambda$ be the corresponding lattice endomorphism. Then we can extend $\rho(\CL)$ to an endomorphism $\tilde\rho(\CL)$ of $E_8\oplus_\perp \Lambda(-1)$ by defining $\tilde\rho(\CL)_{\vert E_8}=\langle \CL\rangle{\rm id}_{E_8}$. This implies that the isotropic sublattice $L_0\subset E_8\oplus_\perp \Lambda(-1)$ is preserved, i.e $\tilde\rho(\CL)_{\vert L_0}=\langle \CL\rangle{\rm id}_{L_0}$, so that $\tilde\rho(\CL)$ induces a well-defined lattice endomorphism of the quotient $(L_0^\perp \cap (E_8\oplus_\perp \Lambda(-1)))/L_0=\Gamma$, whose $\CC$-linear extension is exactly $\varphi\circ\rho(\CL)\circ\varphi^{-1}$. In particular, the restriction of $\varphi\circ\rho(\CL)\circ\varphi^{-1}$ to $\Gamma_{\Pi_{K3}}$ is an endomorphism of $\Gamma_{\Pi_{K3}}$, while the restriction to its orthogonal subspace $\Gamma^{\Pi_{K3}}\otimes \RR$, and in particular the restriction to $\Pi_{K3}$, is $\langle\CL\rangle$ times the identity.

For invertible defects, this construction gives an isomorphism between the subgroup of $\Aut_\tau(V^{f\natural})=\Aut(\Lambda)$ fixing $\Pn$ and the group of symmetries of the K3 sigma model $\calC$ preserving the $\CN=(4,4)$ superconformal symmetry and the spectral flow \cite{Gaberdiel:2011fg}. 
For a generic non-invertible defect $\CL\in \cTop_\Pn$, we cannot prove that there exists a defect $\CL'\in \cTop^{K3}_\calC$ of $\calC$ such that the restriction of $\hat\CL'$  to $\Hh_{\HRR,gr}$ coincides with $\varphi\circ\rho(\CL)\circ \varphi^{-1}$; neither can we prove that for all $\CL'\in \cTop^{K3}_\calC$, the restriction of $\hat\CL'$  to $\Hh^{K3}_{R-R,gr}$ is reproduced by $\varphi\circ\rho(\CL)\circ \varphi^{-1}$ for some $\CL\in \cTop_\Pn$. However, the results of \cite{Angius:2024evd} imply that all defects $\CL'\in \cTop_\calC^{K3}$ have integral quantum dimension for our choice of $\Pi_{K3}$, and that for every $\CL\in \cTop_\Pn$ the maps $\varphi\circ\rho(\CL)\circ\varphi^{-1}$
 satisfy all known consistency conditions for a map $\hat\CL':\Hh_{\HRR,gr}^{K3}\to\Hh_{\HRR,gr}$ for any $\CL'\in \cTop_\calC^{K3}$.

 If we drop the assumption that the quantum dimension $\langle \CL\rangle$ of every $\CL\in \cTop_{\Pi^\natural}$ is integral, then the inclusion \eqref{inclLambda} is not true, and as a consequence $\frac{1}{\langle\CL\rangle}\hat\CL$ may in general act non-trivially on $\Lambda_\Pn^\perp$.

\section{The twisted-twining partition function $Z_{g_2}^{k_2}$ in the GTVW model}\label{a:Zg2p2}

In this section, we compute the twisted-twining partition function $Z_{g_2}^{k_2}$ of eq.\eqref{Zg2p2}. We can take $g_2$ and $k_2$ to be the permutations $g_2=(34)(56)$ and $k_2=(35)(46)$ of the six $\widehat{su}(2)_1$ factors (note that the group $G\cong A_5$ acts by the same permutation on the chiral and antichiral algebra). It is easy to prove that the group $\langle g_2,k_2\rangle\cong \ZZ_2\times \ZZ_2$ is the centralizer of $g_2$ in $G\cong A_5$.   

To compute $Z_{g_2}^{k_2}$, we  use the following strategy. It is easy to check that the subalgebra of $(\widehat{su}(2)_1)^6$ that is invariant under $g_2=(34)(56)$ contains the algebra $
(\widehat{su}(2)_1)^2((\widehat{su}(2)_2\otimes Vir_{c=1/2})^2$.\footnote{This follows from the observation that the subalgebra of $(\widehat{su}(2)_1)^2$ that is invariant under a symmetry that exchanges the two $\widehat{su}(2)_1$ factors is isomorphic to $\widehat{su}(2)_2\otimes Vir_{c=1/2}$. Here, $\widehat{su}(2)_2$ is generated by the three invariant (`diagonal') combinations of currents of $(\widehat{su}(2)_1)^2$, and $Vir_{c=1/2}$ is generated by the difference $T_{(\widehat{su}(2)_1)^2}-T_{\widehat{su}(2)_2}$ of the Sugawara stress tensors. The standard coset identification $\frac{(\widehat{su}(2)_1)^2}{\widehat{su}(2)_2}\cong Vir_{c=1/2}$ then implies that this is the full invariant subalgebra.} Crucially, the K3 model $\calC_{\GTVW}$ is  rational with respect to this algebra, so that we can write the $g_2$-twisted sector as a finite sum of irreducible modules with respect to this algebra and its antichiral counterpart. 

We denote by $\ch_{k,j}$, $j=0,\ldots,k$, the $\widehat{su}(2)_k$ characters, with $\ch_{j}:= \ch_{1,j}$, and by $\chi_0$, $\chi_{1/16}$ and $\chi_{1/2}$ the $Vir_{c=1/2}$ ones.  Then, the decomposition of the $(\widehat{su}(2)_1)^2$ characters into $\widehat{su}(2)_2\otimes Vir_{c=1/2}$ characters is as follows:
\begin{align}
\ch_{0}(\tau)^2&=\ch_{2,0}(\tau)\chi_0(\tau)+\ch_{2,2}(\tau)\chi_{1/2}(\tau)\ ,\\
\ch_{1}(\tau)^2&=\ch_{2,2}(\tau)\chi_0(\tau)+\ch_{2,0}(\tau)\chi_{1/2}(\tau)\ ,\\
\ch_{0}(\tau)\ch_{1}(\tau)&=\ch_{2,1}(\tau)\chi_{1/16}(\tau)\ .
\end{align} 

In the following, we will repeatedly use this well-known trick. Suppose the chiral algebra of a CFT contains the product $\CA\otimes \CA$ of two copies of the same algebra $\CA$, and suppose we want to compute the twining partition function for an automorphism $\phi$ that exchanges these two copies. Under this automorphism, a representation $M\otimes M'$ of $\CA\otimes\CA$ gets mapped to $M'\otimes M$, where $M, M'$ are representations of $\CA$. Therefore, the automorphism $\phi$ is `off-diagonal' when acting on $(M\otimes M')\oplus (M'\otimes M)$, so that 
\be \Tr_{(M\otimes M')\oplus (M'\otimes M)}(\phi q^{L_0-\frac{c}{24}})=0\ ,\qquad M\neq M'\ .
\ee
On the other hand, a representation $M\otimes M$ is stabilised by $\phi$, and the only states in $M\otimes M$ contributing to the $\phi$-twined partition function are the ones of the form $|\psi\rangle\otimes |\psi\rangle$, for each $|\psi\rangle\in M$, that are fixed by the automorphism and whose $L_0$-eigenvalue is twice the eigenvalue of $|\psi\rangle$. As a consequence, the $\phi$-twined character for the representation $M\otimes M$ is
\be\label{permutetwo} \Tr_{M\otimes M}(\phi q^{L_0-\frac{c}{24}})=\ch_{M}^{\CA}(2\tau)\ ,
\ee where $\ch_{M}^{\CA}(\tau)$ denotes the character of the representation $M$ of the algebra $\CA$.

When we apply this argument to the case where $\CA=\widehat{su}(2)_1$ and $\phi$ is an automorphism exchanging the two copies of $\CA$ in $(\widehat{su}(2)_1)^2$, we find that the  $([0]\otimes[1]) + ([1]\otimes [0])$ representations give vanishing contribution to the twining partition function, while the twining character of a $[r]\otimes [r]$ representation, $r=0,1$, is $\ch_r(2\tau)$. On the other hand, in terms of representations of the subalgebra $\widehat{su}(2)_2\times Vir_{c=1/2}$, the automorphism changes the sign of the  $\ch_{2,2}(\tau)\chi_{1/2}(\tau)$ and $\ch_{2,0}(\tau)\chi_{1/2}(\tau)$ terms. This observation gives the relations
\begin{align}
	\ch_{0}(2\tau)&=\ch_{2,0}(\tau)\chi_0(\tau)-\ch_{2,2}(\tau)\chi_{1/2}(\tau)\ ,\\
	\ch_{1}(2\tau)&=\ch_{2,2}(\tau)\chi_0(\tau)-\ch_{2,0}(\tau)\chi_{1/2}(\tau)\ ,
	\end{align}
that can be verified by explicit calculations.

These observations imply that the (unflavoured) $g_2$-twined partition function $Z^{g_2,+}_{e,\NS}(\tau)=Z^{g_2,+}_{e,\NS}(\tau,z=0,\bar z=0)$ (NSNS with no fermion number) of the $\calC_{\GTVW}$ model can be re-written as

\begin{equation}
\begin{split}
Z^{g_2,+}_{e,\NS} =&   \frac{1}{\vert \eta (\tau) \vert^4 }  \left\lbrace \Big\vert \theta_3 (2 \tau )^2(\ch_{2,0}\chi_0-\ch_{2,2}\chi_{1/2})^2 + \theta_2 (2 \tau)^2 (\ch_{2,2}\chi_0-\ch_{2,0}\chi_{1/2})^2 \Big\vert^2 + \right. \\
& \left.+ \Big\vert \theta_3 (2 \tau )^2(\ch_{2,2}\chi_0-\ch_{2,0}\chi_{1/2})^2 + \theta_2 (2 \tau)^2(\ch_{2,0}\chi_0-\ch_{2,2}\chi_{1/2})^2   \Big\vert^2 + \right. \\
& \left. +2 \Big\vert (\theta_3 (2 \tau )^2+\theta_2 (2 \tau)^2) (\ch_{2,0}\chi_0-\ch_{2,2}\chi_{1/2}) (\ch_{2,2}\chi_0-\ch_{2,0}\chi_{1/2})\Big\vert^2  \right\rbrace \\
\end{split}
\label{g2_twining_theta2}
\end{equation} where we didn't write the argument of the characters when it is simply $\tau$.
The $g_2$-twisted partition function $Z_{g_2,\NS}^{e,+}$ is obtained by a modular S-transformation of $Z^{g_2,+}_{e,\NS}$. Recall that the three $\widehat{su}(2)_2$ characters $\ch_{2,0},\ch_{2,1},\ch_{2,2}$ and the $Vir_{c=1/2}$ characters $\chi_0$, $\chi_{1/16}$, $\chi_{1/2}$ transform with the same S-matrix
\be \frac{1}{2}\begin{pmatrix}
	1 & \sqrt{2} & 1\\
	\sqrt{2} & 0 & -\sqrt{2}\\
	1 & -\sqrt{2} & 1\\
\end{pmatrix} \ .
\ee
This leads to the following S-transformations
\be \ch_{2,0}\chi_0-\ch_{2,2}\chi_{1/2}\quad\stackrel{S}{\longrightarrow}\quad \frac{\sqrt{2}}{2}\left (\ch_{2,1}(\chi_0+\chi_{1/2})+(\ch_{2,0}+\ch_{2,2})\chi_{1/16}\right)
\ee
\be \ch_{2,2}\chi_0-\ch_{2,0}\chi_{1/2}\quad\stackrel{S}{\longrightarrow}\quad \frac{\sqrt{2}}{2}\left(-\ch_{2,1}(\chi_0+\chi_{1/2})+(\ch_{2,0}+\ch_{2,2})\chi_{1/16}\right)
\ee
Comparison with the S-transformations of $\ch_{0}(2\tau)$ and $\ch_{1}(2\tau)$ gives the identities
\be (\ch_{2,0}(\tau)+\ch_{2,2}(\tau))\chi_{1/16}(\tau)=\ch_{0}(\tau/2)=\frac{\theta_3(\tau)}{\eta(\tau/2)},\label{chtotheta1}
\ee
\be\label{chtotheta2} \ch_{2,1}(\tau)(\chi_{0}(\tau)+\chi_{1/2}(\tau))=\ch_{1}(\tau/2)=\frac{\theta_2(\tau)}{\eta(\tau/2)}.
\ee
Therefore, the $g_2$-twisted partition function is
\begin{equation}
	\begin{split}
		Z_{g_2,\NS}^{e,+}&(\tau) =   \frac{1}{16\vert \eta (\tau) \vert^4 }  \left\lbrace \Big\vert \theta_3 \left(\frac{\tau}{2} \right) ^2\left(\ch_{2,1}(\chi_0+\chi_{1/2})+(\ch_{2,0}+\ch_{2,2})\chi_{1/16}\right )^2\right. \\
		&\qquad\qquad\qquad\qquad\left.+ \theta_2 \left(\frac{\tau}{2}\right)^2 \left(-\ch_{2,1}(\chi_0+\chi_{1/2})+(\ch_{2,0}+\ch_{2,2})\chi_{1/16}\right)^2 \Big\vert^2 + \right. \\
		& \left.+ \Big\vert \theta_3 \left(\frac{\tau}{2}\right)^2 \left(-\ch_{2,1}(\chi_0+\chi_{1/2})+(\ch_{2,0}+\ch_{2,2})\chi_{1/16}\right)^2 \right.\\
        &\qquad\qquad\qquad\qquad + \theta_2 \left(\frac{\tau}{2}\right)^2  \left(\ch_{2,1}(\chi_0+\chi_{1/2})+(\ch_{2,0}+\ch_{2,2})\chi_{1/16}\right)^2   \Big\vert^2 +  \\
		& +2 \Big\vert \left(\theta_3 \left(\frac{\tau}{2}\right)^2+\theta_2 \left(\frac{\tau}{2}\right)^2\right) \left(\ch_{2,1}(\chi_0+\chi_{1/2})+(\ch_{2,0}+\ch_{2,2})\chi_{1/16}\right) \\
        &\qquad\qquad\qquad\qquad\cdot \left. \left(-\ch_{2,1}(\chi_0+\chi_{1/2})+(\ch_{2,0}+\ch_{2,2})\chi_{1/16}\right)\Big\vert^2  \right\rbrace \\
	\end{split}
	\label{g2_twist_theta2}
\end{equation}
Let us denote by $$[a_1,a_2,x_1,y_1,x_2,y_2;\bar a_1,\bar a_2,\bar x_1,\bar y_1,\bar x_2,\bar y_2]$$ the modules of a holomorphic times antiholomorphic copy of the $(\widehat{su}(2)_1)^2((\widehat{su}(2)_2\otimes Vir_{c=1/2})^2$ algebra. In particular, $a_1,a_2\in \{0,1\}$ denote the $\widehat{su}(2)_1$-modules of conformal weights $0$, $1/4$, $x_1,x_2\in \{0,1,2\}$ denote the $\widehat{su}(2)_2$-modules of conformal weights $0$, $3/16$, $1/2$, and $y_1,y_2\in \{0,1,2\}$ denote the $Vir_{c=1/2}$ modules of weights $0$, $1/16$, $1/2$. Analogously, the indices $\bar a_1,\bar a_2,\bar x_1,\bar y_1,\bar x_2,\bar y_2$ denote the representations with respect to the antiholomorphic algebra. 

The symmetry $k_2$ acts by simultaneously exchanging the two $\widehat{su}(2)_2$ and the two $Vir_{c=1/2}$ factors of both the holomorphic and antiholomorphic algebra. This means that it exchanges the representations
$$[a_1,a_2,x_1,y_1,x_2,y_2;\bar a_1,\bar a_2,\bar x_1,\bar y_1,\bar x_2,\bar y_2]\leftrightarrow [a_1,a_2,x_2,y_2,x_1,y_1;\bar a_1,\bar a_2,\bar x_2,\bar y_2,\bar x_1,\bar y_1]\ .$$
Thus, the only modules that give a non-zero contribution to the $g_2$-twisted $k_2$-twined partition function are the ones that satisfy
\be x_1=x_2\ ,\qquad y_1=y_2\ ,\qquad \bar x_1=\bar x_2\ ,\qquad \bar y_1=\bar y_2\ .
\ee Such modules correspond to contributions to the $g_2$-twisted function $Z_{g_2,\NS}^{e,+}(\tau)$ containing terms of the form $\ch_{2,1}(\tau)^2(\chi_0(\tau)+\chi_{1/2}(\tau))^2$ or $(\ch_{2,0}(\tau)+\ch_{2,2}(\tau))^2\chi_{1/16}(\tau)^2$ in both the holomorphic and antiholomorphic factors. On the other hand, we can ignore contributions containing factors of the form $(\ch_{2,0}(\tau)+\ch_{2,2}(\tau))\chi_{1/16}(\tau)\ch_{2,1}(\tau)(\chi_0(\tau)+\chi_{1/2}(\tau))$ on either  the holomorphic or antiholomorphic side because they correspond to representations that do not contribute to the $k_2$-twined function. We can therefore rearrange the $g_2$-twisted partition function $Z_{g_2,\NS}^{e,+}(\tau)$, keeping track only of the relevant contributions, to obtain
\begin{equation}
	\begin{split}
		Z_{g_2,\NS}^{e,+}(\tau) =&   \frac{1}{16\vert \eta (\tau) \vert^4 }  \left\lbrace 2\left\vert \left(\theta_3 \left(\frac{\tau}{2} \right) ^2+\theta_2 \left(\frac{\tau}{2} \right) ^2\right)\left(\ch_{2,1}^2(\chi_0+\chi_{1/2})^2+(\ch_{2,0}+\ch_{2,2})^2\chi_{1/16}^2\right)\right\vert^2 \right.\\
		& \left. +2 \left\vert \left(\theta_3 \left(\frac{\tau}{2}\right)^2+\theta_2 \left(\frac{\tau}{2}\right)^2\right) \left (-\ch_{2,1}^2(\chi_0+\chi_{1/2})^2+(\ch_{2,0}+\ch_{2,2})^2\chi_{1/16}^2\right) \right\vert^2  +\ldots\right\rbrace \\
        =&   \frac{1}{4\vert \eta (\tau) \vert^4 }   \left\vert \theta_3 \left(\frac{\tau}{2} \right) ^2+\theta_2 \left(\frac{\tau}{2} \right) ^2\right\vert^2\left(\big\vert \ch_{2,1}^2(\chi_0+\chi_{1/2})^2\big\vert^2+\big\vert(\ch_{2,0}+\ch_{2,2})^2\chi_{1/16}^2\big\vert^2\right)+\ldots\\
        =&     \big\vert \ch_0(\tau)^2 +\ch_1(\tau)^2 \big\vert^2\left(\big\vert \ch_{2,1}^2(\chi_0+\chi_{1/2})^2\big\vert^2+\big\vert(\ch_{2,0}+\ch_{2,2})^2\chi_{1/16}^2\big\vert^2\right)+\ldots
	\end{split}
	\label{g2_twist_theta2_simple}
\end{equation} where $\ldots$ denote  factors from representations that will not contribute upon $k_2$ twining, and in the last line we used the identities
\be \frac{\theta_3(\tau/2)}{\eta(\tau)}=\ch_0(\tau)+\ch_1(\tau)\ ,\qquad \frac{\theta_2(\tau/2)}{\eta(\tau)}=\ch_0(\tau)-\ch_1(\tau)\ .
\ee
From the last expression, it is easy to re-introduce the dependence on $z,\bar z$ and also to determine the corresponding function in the $(\NSNS,-)$ spin structure:
\begin{equation}\begin{split}\label{g2twistrelevant} Z_{g_2,\NS}^{e,\pm}(\tau,z,\bar z) =&\big\vert \ch_0(\tau,z)\ch_0(\tau) \pm\ch_1(\tau,z)\ch_1(\tau)^2 \big\vert^2\\&\qquad\cdot\left(\big\vert \ch_{2,1}^2(\chi_0+\chi_{1/2})^2\big\vert^2+\big\vert(\ch_{2,0}+\ch_{2,2})^2\chi_{1/16}^2\big\vert^2\right)+\ldots
\end{split}\end{equation} 
Finally, by \eqref{permutetwo},  the $k_2$-twined partition function can be obtained  by the replacements
\begin{align} \ch_{2,1}^2(\tau)(\chi_0(\tau)+\chi_{1/2}(\tau))^2&\quad\to\quad \ch_{2,1}(2\tau)(\chi_0(2\tau)+\chi_{1/2}(2\tau)) \\(\ch_{2,0}(\tau)+\ch_{2,2}(\tau))^2\chi_{1/16}^2(\tau)&\quad\to\quad (\ch_{2,0}(2\tau)+\ch_{2,2}(2\tau))\chi_{1/16}(2\tau)
\end{align} in \eqref{g2twistrelevant} and by deleting all other contributions in the $\ldots$. For our final result we get
\begin{align*}
	{Z}_{g_2,\NS}^{k_2,\pm}(\tau,z,\bar z)=&\Big\vert (\ch_{1,0}(\tau,z)\ch_{1,0}(\tau)\pm \ch_{1,1}(\tau,z)\ch_{1,1}(\tau))\Big\vert^2\cdot
	\\&\qquad\cdot\left[\Big\vert(\ch_{2,0}(2\tau)+\ch_{2,2}(2\tau))\chi_{1/16}(2\tau)\Big\vert^2+\Big\vert \ch_{2,1}(2\tau)(\chi_{0}(2\tau)+\chi_{1/2}(2\tau))\Big\vert^2\right]\\
    =&\frac{|\theta_3(2\tau,2z)\theta_3(2\tau)\pm \theta_2(2\tau,2z)\theta_2(2\tau)|^2}{|\eta(\tau)|^6}
	\left(|\theta_3(2\tau)|^2+|\theta_2(2\tau)|^2\right)\ ,
\end{align*}
where in the last line we used \eqref{chtotheta1} and \eqref{chtotheta2}.

Eq.\eqref{g2twistrelevant} also implies that the representations that contribute to the $g_2$-twisted, $k_2$-twined partition functions are $g_2$-invariant. Indeed, the action of $g_2$ on the $g_2$-twisted sector $\Hh_{g_2,\NSNS}^{K3}$ is given by
\be (g_2)_{\rvert \Hh_{g_2,\NSNS}^{K3}}=(e^{2\pi i(L_0-\bar L_0)}(-1)^{F+\bar F})_{\rvert \Hh_{g_2,\NSNS}^{K3}}\ ,
\ee  and one can check from the formula above that $e^{2\pi i(L_0-\bar L_0)}(-1)^{F+\bar F}=1$ for all relevant contributions. As a consequence the $g_2$-twisted, $k_2g_2$-twined  partition function is given by the same expression
\be {Z}_{g_2,\NS}^{g_2k_2,\pm}(\tau,z,\bar z)={Z}_{g_2,\NS}^{k_2,\pm}(\tau,z,\bar z)\ ,
\ee and by spectral flow the same must be true in the Ramond sector.

 \printbibliography

@article{Anagiannis:2020hkk,
    author = "Anagiannis, Vassilis and Cheng, Miranda C. N. and Duncan, John and Volpato, Roberto",
    title = "{Vertex operator superalgebra/sigma model correspondences: The four-torus case}",
    eprint = "2009.00186",
    archivePrefix = "arXiv",
    primaryClass = "hep-th",
    doi = "10.1093/ptep/ptab095",
    journal = "PTEP",
    volume = "2021",
    number = "8",
    pages = "08B102",
    year = "2021"
}

@article{Duncan:2014vfa,
    author = "Duncan, John F. R. and Griffin, Michael J. and Ono, Ken",
    title = "{Moonshine}",
    eprint = "1411.6571",
    archivePrefix = "arXiv",
    primaryClass = "math.RT",
    journal = "Res. Math. Sci.",
    volume = "2",
    pages = "11",
    year = "2015"
}

@article{Anagiannis:2018jqf,
    author = "Anagiannis, Vassilis and Cheng, Miranda C. N.",
    title = "{TASI Lectures on Moonshine}",
    eprint = "1807.00723",
    archivePrefix = "arXiv",
    primaryClass = "hep-th",
    doi = "10.22323/1.305.0010",
    journal = "PoS",
    volume = "TASI2017",
    pages = "010",
    year = "2018"
}

@article{Harrison:2022zee,
    author = "Harrison, Sarah M. and Harvey, Jeffrey A. and Paquette, Natalie M.",
    title = "{Snowmass White Paper: Moonshine}",
    eprint = "2201.13321",
    archivePrefix = "arXiv",
    primaryClass = "hep-th",
    month = "1",
    year = "2022"
}

@article{Cheng:2014owa,
    author = "Cheng, Miranda C. N. and Dong, Xi and Duncan, John F. R. and Harrison, Sarah and Kachru, Shamit and Wrase, Timm",
    title = "{Mock Modular Mathieu Moonshine Modules}",
    eprint = "1406.5502",
    archivePrefix = "arXiv",
    primaryClass = "hep-th",
    reportNumber = "SU-ITP-14-17",
    doi = "10.1186/s40687-015-0034-9",
    journal = "Res. Math. Sci.",
    volume = "2",
    pages = "13",
    year = "2015"
}

@article{Harvey:2020jvu,
    author = "Harvey, Jeffrey A. and Moore, Gregory W.",
    title = "{Moonshine, superconformal symmetry, and quantum error correction}",
    eprint = "2003.13700",
    archivePrefix = "arXiv",
    primaryClass = "hep-th",
    doi = "10.1007/JHEP05(2020)146",
    journal = "JHEP",
    volume = "05",
    pages = "146",
    year = "2020"
}

@incollection {Johnson-Freyd:2023uvb,
    AUTHOR = {Furet, Alissa and Johnson-Freyd, Theo},
     TITLE = {Ground-state degeneracy of twisted sectors of {C}onway
              moonshine {SCFT}},
 BOOKTITLE = {Quantum symmetries: tensor categories, {TQFT}s, and vertex
              algebras},
    SERIES = {Contemp. Math.},
    VOLUME = {813},
     PAGES = {99--116},
 PUBLISHER = {Amer. Math. Soc., Providence, RI},
YEAR = {2025},
      ISBN = {978-1-4704-7361-7},
   MRCLASS = {20D08 (17B69 81Q60 81T40)},
  MRNUMBER = {4886200},
       DOI = {10.1090/conm/813/16284},
    eprint = "2305.05081",
    archivePrefix = "arXiv",
    primaryClass = "math-ph",
}

@article{Duncan:2014eha,
    author = "Duncan, John F. R. and Mack-Crane, Sander",
    title = "{The Moonshine Module for Conway\textquoteright{}s Group}",
    eprint = "1409.3829",
    archivePrefix = "arXiv",
    primaryClass = "math.RT",
    doi = "10.1017/fms.2015.7",
    journal = "SIGMA",
    volume = "3",
    pages = "e10",
    year = "2015"
}

@article{Verlinde:1988sn,
    author = "E.P. Verlinde",
    title = "{Fusion Rules and Modular Transformations in 2D Conformal Field Theory}",
    doi = "10.1016/0550-3213(88)90603-7",
    journal = "Nucl. Phys. B",
    volume = "300",
    pages = "360-376",
    year = "1988",
}

@article{Petkova:2000ip,
    author = "V.B. Petkova and J-B. Zuber",
    title = "{Generalized twisted partition functions}",
    eprint = "hep-th/0011021",
    archivePrefix = "arXiv",
    primaryClass = "hep-th",
    doi = "10.1016/S0370-2693(01)00276-3",
    journal = "Phys. Lett. B",
    volume = "504",
    pages = "157-164",
    year = "2001",
}

@article{Cordova:2023qei,
    author = "Cordova, Clay and Rizi, Giovanni",
    title = "{Non-invertible symmetry in Calabi-Yau conformal field theories}",
    eprint = "2312.17308",
    archivePrefix = "arXiv",
    primaryClass = "hep-th",
    doi = "10.1007/JHEP01(2025)045",
    journal = "JHEP",
    volume = "01",
    pages = "045",
    year = "2025"
}

@article{Bhardwaj:2017xup,
    author = "L. Bhardwaj and Y. Tachikawa",
    title = "{On finite symmetries and their gauging in two dimensions}",
    eprint = "1704.02330",
    archivePrefix = "arXiv",
    primaryClass = "hep-th",
    doi = "10.1007/JHEP03(2018)189",
    journal = "JHEP",
    volume = "03",
    pages = "189",
    year = "2018",
}

@article{Chang_2019,
   author={C-M. Chang and Y-H. Lin and S-H. Shao and Y. Wang and X. Yin},
   title={Topological defect lines and renormalization group flows in two dimensions},
    eprint = "1802.04445",
    archivePrefix = "arXiv",
    primaryClass = "hep-th",
    doi = "10.1007/JHEP01(2019)026",
    journal = "JHEP",
    volume = "01",
    year = "2019",
}

@article{Cheng:2015kha,
    author = "Cheng, Miranda C. N. and Duncan, John F. R. and Harrison, Sarah M. and Kachru, Shamit",
    title = "{Equivariant K3 Invariants}",
    eprint = "1508.02047",
    archivePrefix = "arXiv",
    primaryClass = "hep-th",
    doi = "10.4310/CNTP.2017.v11.n1.a2",
    journal = "Commun. Num. Theor. Phys.",
    volume = "11",
    pages = "41--72",
    year = "2017"
}

@article{Thorngren:2021yso,
    author = "Thorngren, Ryan and Wang, Yifan",
    title = "{Fusion category symmetry. Part II. Categoriosities at c = 1 and beyond}",
    eprint = "2106.12577",
    archivePrefix = "arXiv",
    primaryClass = "hep-th",
    doi = "10.1007/JHEP07(2024)051",
    journal = "JHEP",
    volume = "07",
    pages = "051",
    year = "2024"
}

@article{Carqueville:2023jhb,
   author={N. Carqueville and M. Del Zotto and I. Runkel },
   title={Topological defects},
    eprint = "2311.02449",
    archivePrefix = "arXiv",
    primaryClass = "math-ph",
    year = "2023",
}

@inproceedings{Frohlich:2009gb,
    author = "Froehlich, Jurg and Fuchs, Jurgen and Runkel, Ingo and Schweigert, Christoph",
    title = "{Defect lines, dualities, and generalised orbifolds}",
    booktitle = "{16th International Congress on Mathematical Physics}",
    eprint = "0909.5013",
    archivePrefix = "arXiv",
    primaryClass = "math-ph",
    reportNumber = "KCL-MTH-09-10, ZMP-HH-09-20",
    doi = "10.1142/9789814304634_0056",
    month = "9",
    year = "2009"
}

@article{Cheng:2015rby,
    author = "Cheng, Miranda C. N. and Ferrari, Francesca and Harrison, Sarah M. and Paquette, Natalie M.",
    title = "{Landau-Ginzburg Orbifolds and Symmetries of K3 CFTs}",
    eprint = "1512.04942",
    archivePrefix = "arXiv",
    primaryClass = "hep-th",
    doi = "10.1007/JHEP01(2017)046",
    journal = "JHEP",
    volume = "01",
    pages = "046",
    year = "2017"
}

@article{Cheng:2015fha,
    author = "Cheng, Miranda C. N. and Harrison, Sarah M. and Kachru, Shamit and Whalen, Daniel",
    title = "{Exceptional Algebra and Sporadic Groups at c=12}",
    eprint = "1503.07219",
    archivePrefix = "arXiv",
    primaryClass = "hep-th",
    month = "3",
    year = "2015"
}

@article{Creutzig:2017fuk,
    author = "Creutzig, Thomas and Duncan, John F. R. and Riedler, Wolfgang",
    title = "{Self-Dual Vertex Operator Superalgebras and Superconformal Field Theory}",
    eprint = "1704.03678",
    archivePrefix = "arXiv",
    primaryClass = "math-ph",
    doi = "10.1088/1751-8121/aa9af5",
    journal = "J. Phys. A",
    volume = "51",
    number = "3",
    pages = "034001",
    year = "2018"
}

@article{Albert:2022gcs,
    author = "Albert, Jan and Kaidi, Justin and Lin, Ying-Hsuan",
    title = "{Topological modularity of supermoonshine}",
    eprint = "2210.14923",
    archivePrefix = "arXiv",
    primaryClass = "hep-th",
    doi = "10.1093/ptep/ptad034",
    journal = "PTEP",
    volume = "2023",
    number = "3",
    pages = "033B06",
    year = "2023"
}

@article{HohnMason2016,
    author = "G. Hoehn and G. Mason" ,
    title ="{The 290 fixed-point sublattices of the {L}eech lattice}" ,
    eprint = "1505.06420",
    archivePrefix = "arXiv",
    primaryClass = "math.GR",
    doi="10.1016/j.jalgebra.2015.08.028",
    journal ="J. Algebra" ,
    volume="448",
    pages= "618-637",
    year = "2016",
}

@article{Gaberdiel:2013psa,
    author = "M.R. Gaberdiel and A. Taormina and R. Volpato and K. Wendland" ,
    title ="{A K3 sigma model with $\mathbb{Z}^8_2$ : $\mathbb{M}_{20}$ symmetry}" ,
    eprint = "1309.4127",
    archivePrefix = "arXiv",
    primaryClass = "hep-th",
    doi="10.1007/JHEP02(2014)022",
    journal ="JHEP" ,
    volume="02,022",
    year = "2014",
}

@article{Chang:2022hud,
    author = "C.M. Chang and J. Cheng and F. Xu" ,
    title ="{Topological defect lines in two dimensional fermionic CFTs}" ,
    eprint = "2208.02757",
    archivePrefix = "arXiv",
    primaryClass = "hep-th",
    doi="10.21468/SciPostPhys.15.5.216",
    journal ="SciPost Phys." ,
    volume="15",
    number= "5, 216",
    year = "2023",
}

@article{Runkel:2022fzi,
    author = "Runkel, Ingo and Szegedy, L\'or\'ant and Watts, G\'erard M. T.",
    title = "{Parity and spin CFT with boundaries and defects}",
    eprint = "2210.01057",
    archivePrefix = "arXiv",
    primaryClass = "hep-th",
    doi = "10.21468/SciPostPhys.15.5.207",
    journal = "SciPost Phys.",
    volume = "15",
    number = "5",
    pages = "207",
    year = "2023"
}

@article{Fuchs:2007tx,
    author = "J. Fuchs and M.R. Gaberdiel and I. Runkel and C. Schweigert" ,
    title ="{Topological defects for the free boson CFT}" ,
    eprint = "0705.3129",
    archivePrefix = "arXiv",
    primaryClass = "hep-th",
    doi="10.1088/1751-8113/40/37/016",
    journal ="J. Phys. A" ,
    volume="40, 11403",
    year = "2007",
}

@article{Bachas:2012bj,
    author = "C. Bachas and I. Brunner and D. Roggenkamp" ,
    title ="{A worldsheet extension of $O(d,d:Z)$}" ,
    eprint = "1205.4647",
    archivePrefix = "arXiv",
    primaryClass = "hep-th",
    doi="10.1007/JHEP10(2012)039",
    journal ="JHEP" ,
    volume="10, 039",
    year = "2012",
}

@article{Gaberdiel:2011fg,
    author = "M.R. Gaberdiel and S. Hohenegger and R. Volpato" ,
    title ="{Symmetries of K3 sigma models}" ,
    eprint = "1106.4315",
    archivePrefix = "arXiv",
    primaryClass = "hep-th",
    doi="10.4310/CNTP.2012.v6.n1.a1",
    journal ="Commun. Num. Theor. Phys." ,
    volume="6",
    pages="1-50",
    year = "2012",
}

@article{Eguchi:2010ej,
    author = "T. Eguchi and H. Ooguri and Y. Tachikawa" ,
    title ="{Notes on the K3 Surface and the Mathieu group $M_{24}$}" ,
    eprint = "1004.0956",
    archivePrefix = "arXiv",
    primaryClass = "hep-th",
    doi="10.1080/10586458.2011.544585",
    journal ="Experimental Mathematics" ,
    volume="20",
    pages="91-96",
    year = "2011",
}

@article{Eguchi:1988vra,
    author = "Eguchi, Tohru and Ooguri, Hirosi and Taormina, Anne and Yang, Sung-Kil",
    title = "{Superconformal Algebras and String Compactification on Manifolds with SU(N) Holonomy}",
    reportNumber = "UT-536-TOKYO",
    doi = "10.1016/0550-3213(89)90454-9",
    journal = "Nucl. Phys. B",
    volume = "315",
    pages = "193--221",
    year = "1989"
}

@article{Becker:2017zai,
    author = "M. Becker and Y. Cabrera and D. Robbins" ,
    title ="{Conformal interfaces between free boson orbifold theories}" ,
    eprint = "1706.03802",
    archivePrefix = "arXiv",
    primaryClass = "hep-th",
    doi="10.1007/JHEP09(2017)148",
    journal ="JHEP" ,
    volume="09",
    pages="148",
    year = "2017",
}

@article{Coste:2000tq,
    author = "A. Coste and T. Gannon and P. Ruelle" ,
    title ="{Finite group modular data}" ,
    eprint = "hep-th/0001158",
    archivePrefix = "arXiv",
    primaryClass = "hep-th",
    journal ="Nucl. Phys. B" ,
    volume="581",
    pages="679",
    year = "2000",
}

@article{Paquette:2017gmb,
    author = "N.M. Paquette and R. Volpato and M. Zimet " ,
    title ="{No More Walls! A Tale of Modularity, Symmetry, and Wall Crossing for 1/4 BPS Dyons}" ,
    eprint = "1702.05095",
    archivePrefix = "arXiv",
    primaryClass = "hep-th",
    doi="10.1007/JHEP05(2017)047",
    journal ="JHEP" ,
    volume="05",
    pages="047",
    year = "2017",
}

@article{nikulin,
    author = "V.V. Nikulin" ,
    title ="{Integer symmetric bilinear forms and some of their geometric applications}" ,
    journal ="Izv. Akad. Nauk SSSR Ser.Mat." ,
    volume="43, 111",
    year = "1979",
}

@article{ConwaySloane,
    author = "J.H. Conway and N.J.A. Sloane" ,
    title ="{Sphere packings, lattices and groups}" ,
    journal ="Grundlehren der Mathematischen Wissenschaften" ,
    volume="290",
    year = "1999",
}

@article{Duncan:2006,
  title={Super-moonshine for Conway's largest sporadic group},
  author={Duncan, John F},
  journal={Duke Math. J.},
  volume={136},
  number={1},
  pages={255--315},
  year={2007}
}

@article{Duncan:2015xoa,
    author = "J.F.R. Duncan and S. Mack-Crane" ,
    title ="{Derived Equivalences of K3 Surfaces and Twined Elliptic Genera}" ,
JOURNAL = {Res. Math. Sci.},
  FJOURNAL = {Research in the Mathematical Sciences},
    VOLUME = {3},
      YEAR = {2016},
    eprint = "1506.06198",
    archivePrefix = "arXiv",
    primaryClass = "math.RT",
    year = "2015",
DOI = {10.1186/s40687-015-0050-9},
}

@article{Cheng:2016org,
    author = "M.C.N. Cheng and S.M. Harrison and R. Volpato and M. Zimet" ,
    title ="{K3 String Theory, Lattices and Moonshine}" ,
    eprint = "1612.04404",
    archivePrefix = "arXiv",
    primaryClass = "hep-th",
   JOURNAL = {Res. Math. Sci.},
  FJOURNAL = {Research in the Mathematical Sciences},
    VOLUME = {5},
      YEAR = {2018},
    NUMBER = {3},
       DOI = {10.1007/s40687-018-0150-4},
}

@article{Taormina:2017zlm,
    author = "A. Taormina and K. Wendland" ,
    title ="{The Conway Moonshine Module is a reflected K3 theory}" ,
    eprint = "1704.03813",
    archivePrefix = "arXiv",
    primaryClass = "hep-th",
    doi="10.4310/ATMP.2020.v24.n5.a6",
    journal ="Adv. Theor. Math. Phys." ,
    volume="24",
    number="5",
    pages="1247-1323",
    year = "2020",
}

@article{Harrison:2020wxl,
    author = "S.M. Harrison, N.M. Paquette, D. Persson and R. Volpato" ,
    title ="{Fun with $F_{24}$}" ,
    eprint = "2009.14710",
    archivePrefix = "arXiv",
    primaryClass = "hep-th",
    doi="10.1007/JHEP02(2021)039",
    journal ="JHEP" ,
    volume="02",
    pages="039",
    year = "2021",
}

@article{Harrison:2021gnp,
    author = "S.M. Harrison, N.M. Paquette, D. Persson and R. Volpato" ,
    title ="{BPS Algebras in 2D String Theory}" ,
    eprint = "2107.03507",
    archivePrefix = "arXiv",
    primaryClass = "hep-th",
    doi="10.1007/s00023-022-01189-7",
    journal ="Annales Henri Poincare" ,
    volume="23",
    number="10",
    pages="3667-3752",
    year = "2022",
}

@article{Harrison:2018joy,
    author = "Harrison, Sarah M. and Paquette, Natalie M. and Volpato, Roberto",
    title = "{A Borcherds\textendash{}Kac\textendash{}Moody Superalgebra with Conway Symmetry}",
    eprint = "1803.10798",
    archivePrefix = "arXiv",
    primaryClass = "hep-th",
    doi = "10.1007/s00220-019-03518-0",
    journal = "Commun. Math. Phys.",
    volume = "370",
    number = "2",
    pages = "539--590",
    year = "2019"
}

@article{Angius:2024evd,
    author = "Angius, Roberta and Giaccari, Stefano and Volpato, Roberto",
    title = "{Topological defects in K3 sigma models}",
    eprint = "2402.08719",
    archivePrefix = "arXiv",
    primaryClass = "hep-th",
    doi = "10.1007/JHEP07(2024)111",
    journal = "JHEP",
    volume = "07",
    pages = "111",
    year = "2024"
}

@article{Angius:2025xxx,
    author = "R. Angius and S. Giaccari and S.Harrison and R. Volpato",
    year="in preparation",
}

@article{Cheng:2014zpa,
    author = "M.C.N. Cheng and S. Harrison" ,
    title ="{Umbral Moonshine and K3 Surfaces}" ,
    eprint = "arXiv:1406.0619",
    archivePrefix = "arXiv",
    primaryClass = "hep-th",
    doi="10.1007/s00220-015-2398-5",
    journal ="Commun. Math. Phys." ,
    volume="339, 01",
    pages="221-261",
    year = "2015",
}

@article{Moller:2024xtt,
    author = {M\"oller, Sven and Rayhaun, Brandon C.},
    title = "{Equivalence Relations on Vertex Operator Algebras, II: Witt Equivalence and Orbifolds}",
    eprint = "2410.18166",
    archivePrefix = "arXiv",
    primaryClass = "hep-th",
    month = "10",
    year = "2024"
}

@book{Etingof:2015,
    AUTHOR = {Etingof, Pavel and Gelaki, Shlomo and Nikshych, Dmitri and
              Ostrik, Victor},
     TITLE = {Tensor categories},
    SERIES = {Mathematical Surveys and Monographs},
    VOLUME = {205},
 PUBLISHER = {American Mathematical Society, Providence, RI},
      YEAR = {2015},
     PAGES = {xvi+343},
      ISBN = {978-1-4704-2024-6},
   MRCLASS = {18D10 (16T05)},
  MRNUMBER = {3242743},
MRREVIEWER = {Julien\ Bichon},
       DOI = {10.1090/surv/205},
       URL = {https://doi.org/10.1090/surv/205},
}

@article{Volpato:2024goy,
    author = "Volpato, Roberto",
    title = "{Vertex algebras, topological defects, and Moonshine}",
    eprint = "2412.21141",
    archivePrefix = "arXiv",
    primaryClass = "hep-th",
    month = "12",
    year = "2024"
}

@article{Frenkel:1988flm,
    author = "Frenkel, Igor and Lepowsky, James and Meurman, Arne",
    title = "{Vertex operator algebras and the Monster,}",
    journal = "Pure and Applied Mathematics",
    volume = "134",
    publisher = "Academic Press, Inc., Boston, MA",
    year = "1988"
}

@article {Huang1,
    AUTHOR = {Huang, Yi-Zhi},
     TITLE = {Rigidity and modularity of vertex tensor categories},
   JOURNAL = {Commun. Contemp. Math.},
  FJOURNAL = {Communications in Contemporary Mathematics},
    VOLUME = {10},
      YEAR = {2008},
     PAGES = {871--911},
      ISSN = {0219-1997,1793-6683},
   MRCLASS = {17B69 (18D10 81R10 81T40)},
  MRNUMBER = {2468370},
MRREVIEWER = {Pavel\ S.\ Kolesnikov},
       DOI = {10.1142/S0219199708003083},
       URL = {https://doi.org/10.1142/S0219199708003083},
}

@article {Huang2,
    AUTHOR = {Huang, Yi-Zhi},
     TITLE = {Vertex operator algebras and the {V}erlinde conjecture},
   JOURNAL = {Commun. Contemp. Math.},
  FJOURNAL = {Communications in Contemporary Mathematics},
    VOLUME = {10},
      YEAR = {2008},
    NUMBER = {1},
     PAGES = {103--154},
      ISSN = {0219-1997,1793-6683},
   MRCLASS = {17B69 (81R10 81T40)},
  MRNUMBER = {2387861},
MRREVIEWER = {Markus\ Rosellen},
       DOI = {10.1142/S0219199708002727},
       URL = {https://doi.org/10.1142/S0219199708002727},
}

@article {DongRenYang:2022,
    AUTHOR = {Dong, Chongying and Ren, Li and Yang, Meiling},
     TITLE = {Super orbifold theory},
   JOURNAL = {Adv. Math.},
  FJOURNAL = {Advances in Mathematics},
    VOLUME = {405},
      YEAR = {2022},
     PAGES = {Paper No. 108481, 34},
      ISSN = {0001-8708,1090-2082},
   MRCLASS = {17B69},
  MRNUMBER = {4436341},
MRREVIEWER = {Stanislav\ Z.\ Pakuliak},
       DOI = {10.1016/j.aim.2022.108481},
       URL = {https://doi.org/10.1016/j.aim.2022.108481},
}

@article {DongNgRen:2021,
    AUTHOR = {Dong, Chongying and Ng, Siu-Hung and Ren, Li},
     TITLE = {Vertex operator superalgebras and the 16-fold way},
   JOURNAL = {Trans. Amer. Math. Soc.},
  FJOURNAL = {Transactions of the American Mathematical Society},
    VOLUME = {374},
      YEAR = {2021},
    NUMBER = {11},
     PAGES = {7779--7810},
      ISSN = {0002-9947,1088-6850},
   MRCLASS = {17B69},
  MRNUMBER = {4328683},
MRREVIEWER = {Matthew\ Krauel},
       DOI = {10.1090/tran/8454},
       URL = {https://doi.org/10.1090/tran/8454},
}

@article{Gaudio:2024zxu,
    author = "Gaudio, Tiziano",
    title = "{Unitarity and strong graded locality of holomorphic vertex operator superalgebras with central charge at most 24}",
    eprint = "2410.07099",
    archivePrefix = "arXiv",
    primaryClass = "math.QA",
     doi="https://doi.org/10.1007/s00023-025-01542-6",
    journal ="Annales Henri Poincare" ,
    year = "2025",
}

@article{Carpi:2023onx,
    author = "Carpi, Sebastiano and Gaudio, Tiziano and Hillier, Robin",
    title = "{From vertex operator superalgebras to graded-local conformal nets and back}",
    eprint = "2304.14263",
    archivePrefix = "arXiv",
    primaryClass = "math.OA",
    month = "4",
    year = "2023"
}

@article{Dixon:1986jc,
    author = "Dixon, Lance J. and Harvey, Jeffrey A. and Vafa, C. and Witten, Edward",
    title = "{Strings on Orbifolds. 2.}",
    reportNumber = "PRINT-86-0246 (PRINCETON)",
    doi = "10.1016/0550-3213(86)90287-7",
    journal = "Nucl. Phys. B",
    volume = "274",
    pages = "285--314",
    year = "1986"
}

@article{Narain:1986qm,
    author = "Narain, K. S. and Sarmadi, M. H. and Vafa, C.",
    title = "{Asymmetric Orbifolds}",
    reportNumber = "HUTP-86-A089",
    doi = "10.1016/0550-3213(87)90228-8",
    journal = "Nucl. Phys. B",
    volume = "288",
    pages = "551",
    year = "1987"
}

@article{Dijkgraaf:1989pz,
    author = "Dijkgraaf, Robbert and Witten, Edward",
    title = "{Topological Gauge Theories and Group Cohomology}",
    reportNumber = "THU-89-9, IASSNS-HEP-89-33",
    doi = "10.1007/BF02096988",
    journal = "Commun. Math. Phys.",
    volume = "129",
    pages = "393",
    year = "1990"
}

@article{Roche:1990hs,
    author = "Roche, P. and Pasquier, V. and Dijkgraaf, R.",
    title = "{QuasiHopf algebras, group cohomology and orbifold models}",
    journal = "Nucl. Phys. B Proc. Suppl.",
    volume = "18",
    pages = "60--72",
    year = "1990"
}

@article {DongLepowsky,
    AUTHOR = {Dong, Chongying and Lepowsky, James},
     TITLE = {The algebraic structure of relative twisted vertex operators},
   JOURNAL = {J. Pure Appl. Algebra},
  FJOURNAL = {Journal of Pure and Applied Algebra},
    VOLUME = {110},
      YEAR = {1996},
    NUMBER = {3},
     PAGES = {259--295},
      ISSN = {0022-4049,1873-1376},
   MRCLASS = {17B69 (81R10 81T40)},
  MRNUMBER = {1393116},
MRREVIEWER = {Mirko\ Primc},
       DOI = {10.1016/0022-4049(95)00095-X},
       URL = {https://doi.org/10.1016/0022-4049(95)00095-X},
}

@article {DLM1,
    AUTHOR = {Dong, Chongying and Li, Haisheng and Mason, Geoffrey},
     TITLE = {Simple currents and extensions of vertex operator algebras},
   JOURNAL = {Comm. Math. Phys.},
  FJOURNAL = {Communications in Mathematical Physics},
    VOLUME = {180},
      YEAR = {1996},
    NUMBER = {3},
     PAGES = {671--707},
      ISSN = {0010-3616,1432-0916},
   MRCLASS = {17B69 (81R10)},
  MRNUMBER = {1408523},
MRREVIEWER = {Mirko\ Primc},
       URL = {http://projecteuclid.org/euclid.cmp/1104287460},
}

@article {DLM2,
    AUTHOR = {Dong, Chongying and Li, Haisheng and Mason, Geoffrey},
     TITLE = {Compact automorphism groups of vertex operator algebras},
   JOURNAL = {Internat. Math. Res. Notices},
  FJOURNAL = {International Mathematics Research Notices},
      YEAR = {1996},
    NUMBER = {18},
     PAGES = {913--921},
      ISSN = {1073-7928,1687-0247},
   MRCLASS = {17B69},
  MRNUMBER = {1420556},
MRREVIEWER = {J\"urgen\ Schulze},
       DOI = {10.1155/S1073792896000566},
       URL = {https://doi.org/10.1155/S1073792896000566},
}

@article {DLM3,
    AUTHOR = {Dong, Chongying and Li, Haisheng and Mason, Geoffrey},
     TITLE = {Twisted representations of vertex operator algebras},
   JOURNAL = {Math. Ann.},
  FJOURNAL = {Mathematische Annalen},
    VOLUME = {310},
      YEAR = {1998},
    NUMBER = {3},
     PAGES = {571--600},
      ISSN = {0025-5831,1432-1807},
   MRCLASS = {17B69},
  MRNUMBER = {1615132},
MRREVIEWER = {Mirko\ Primc},
       DOI = {10.1007/s002080050161},
       URL = {https://doi.org/10.1007/s002080050161},
}

@article {DLM4,
    AUTHOR = {Dong, Chongying and Li, Haisheng and Mason, Geoffrey},
     TITLE = {Twisted representations of vertex operator algebras and
              associative algebras},
   JOURNAL = {Internat. Math. Res. Notices},
  FJOURNAL = {International Mathematics Research Notices},
      YEAR = {1998},
    NUMBER = {8},
     PAGES = {389--397},
      ISSN = {1073-7928,1687-0247},
   MRCLASS = {17B69},
  MRNUMBER = {1628239},
MRREVIEWER = {Dra\v zen\ Adamovi\'c},
       DOI = {10.1155/S1073792898000269},
       URL = {https://doi.org/10.1155/S1073792898000269},
}

@article {DLM5,
    AUTHOR = {Dong, Chongying and Li, Haisheng and Mason, Geoffrey},
     TITLE = {Modular-invariance of trace functions in orbifold theory and
              generalized {M}oonshine},
   JOURNAL = {Comm. Math. Phys.},
  FJOURNAL = {Communications in Mathematical Physics},
    VOLUME = {214},
      YEAR = {2000},
    NUMBER = {1},
     PAGES = {1--56},
      ISSN = {0010-3616,1432-0916},
   MRCLASS = {17B69 (11F22)},
  MRNUMBER = {1794264},
MRREVIEWER = {Vassily\ Gorbounov},
       DOI = {10.1007/s002200000242},
       URL = {https://doi.org/10.1007/s002200000242},
}

@article{Carnahan:2016guf,
    author = "Carnahan, Scott and Miyamoto, Masahiko",
    title = "{Regularity of fixed-point vertex operator subalgebras}",
    eprint = "1603.05645",
    archivePrefix = "arXiv",
    primaryClass = "math.RT",
    month = "3",
    year = "2016"
}

@article{Grigoletto:2021zyv,
    author = "Grigoletto, Andrea and Putrov, Pavel",
    title = "{Spin-Cobordisms, Surgeries and Fermionic Modular Bootstrap}",
    eprint = "2106.16247",
    archivePrefix = "arXiv",
    primaryClass = "hep-th",
    doi = "10.1007/s00220-023-04710-z",
    journal = "Commun. Math. Phys.",
    volume = "401",
    number = "3",
    pages = "3169--3245",
    year = "2023"
}

@article{Bhardwaj:2024ydc,
    author = "Bhardwaj, Lakshya and Inamura, Kansei and Tiwari, Apoorv",
    title = "{Fermionic Non-Invertible Symmetries in (1+1)d: Gapped and Gapless Phases, Transitions, and Symmetry TFTs}",
    eprint = "2405.09754",
    archivePrefix = "arXiv",
    primaryClass = "hep-th",
    month = "5",
    year = "2024"
}

@article{Arias-Tamargo:2025xdd,
    author = "Arias-Tamargo, Guillermo and Hull, Chris and Vel\'asquez Cotini Hutt, Maxwell L.",
    title = "{Non-invertible symmetries of two-dimensional Non-Linear Sigma Models}",
    eprint = "2503.20865",
    archivePrefix = "arXiv",
    primaryClass = "hep-th",
    reportNumber = "Imperial-TP-2025-CH-3",
    month = "3",
    year = "2025"
}

@article{Caldararu:2025eoj,
    author = "Caldararu, A. and Pantev, T. and Sharpe, E. and Sung, B. and Yu, X.",
    title = "{Noninvertible symmetries in the B model TFT}",
    eprint = "2504.02023",
    archivePrefix = "arXiv",
    primaryClass = "hep-th",
    month = "4",
    year = "2025"
}

@article{Gaiotto:2018ypj,
    author = "Gaiotto, Davide and Johnson-Freyd, Theo",
    title = "{Holomorphic SCFTs with small index}",
    eprint = "1811.00589",
    archivePrefix = "arXiv",
    primaryClass = "hep-th",
    doi = "10.4153/S0008414X2100002X",
    journal = "Can. J. Math.",
    volume = "74",
    number = "2",
    pages = "573--601",
    year = "2022"
}

@article{Gu:2012ib,
    author = "Gu, Zheng-Cheng and Wen, Xiao-Gang",
    title = "{Symmetry-protected topological orders for interacting fermions: Fermionic topological nonlinear {\ensuremath{\sigma}} models and a special group supercohomology theory}",
    eprint = "1201.2648",
    archivePrefix = "arXiv",
    primaryClass = "cond-mat.str-el",
    doi = "10.1103/PhysRevB.90.115141",
    journal = "Phys. Rev. B",
    volume = "90",
    number = "11",
    pages = "115141",
    year = "2014"
}

@article{Kapustin:2014dxa,
    author = "Kapustin, Anton and Thorngren, Ryan and Turzillo, Alex and Wang, Zitao",
    title = "{Fermionic Symmetry Protected Topological Phases and Cobordisms}",
    eprint = "1406.7329",
    archivePrefix = "arXiv",
    primaryClass = "cond-mat.str-el",
    doi = "10.1007/JHEP12(2015)052",
    journal = "JHEP",
    volume = "12",
    pages = "052",
    year = "2015"
}

@article{Gaiotto:2015zta,
    author = "Gaiotto, Davide and Kapustin, Anton",
    editor = "Dokshitzer, Yuri L. and Levai, Peter and Nyiri, Julia",
    title = "{Spin TQFTs and fermionic phases of matter}",
    eprint = "1505.05856",
    archivePrefix = "arXiv",
    primaryClass = "cond-mat.str-el",
    doi = "10.1142/S0217751X16450445",
    journal = "Int. J. Mod. Phys. A",
    volume = "31",
    number = "28n29",
    pages = "1645044",
    year = "2016"
}

@article{Brumfiel:2016vpy,
    author = "Brumfiel, Greg and Morgan, John",
    title = "{The Pontrjagin Dual of 3-Dimensional Spin Bordism}",
    eprint = "1612.02860",
    archivePrefix = "arXiv",
    primaryClass = "math.AT",
    month = "12",
    year = "2016"
}

@article{Runkel:2020zgg,
    author = "Runkel, Ingo and Watts, G{\'e}rard M. T.",
    title = "{Fermionic CFTs and classifying algebras}",
    eprint = "2001.05055",
    archivePrefix = "arXiv",
    primaryClass = "hep-th",
    reportNumber = "kcl-mth-20-01",
    doi = "10.1007/JHEP06(2020)025",
    journal = "JHEP",
    volume = "06",
    pages = "025",
    year = "2020"
}

@article{Aasen:2017ubm,
    author = "Aasen, David and Lake, Ethan and Walker, Kevin",
    title = "{Fermion condensation and super pivotal categories}",
    eprint = "1709.01941",
    archivePrefix = "arXiv",
    primaryClass = "cond-mat.str-el",
    doi = "10.1063/1.5045669",
    journal = "J. Math. Phys.",
    volume = "60",
    number = "12",
    pages = "121901",
    year = "2019"
}

@article{Johnson-Freyd:2017ble,
    author = "Johnson-Freyd, Theo",
    title = "{The Moonshine Anomaly}",
    eprint = "1707.08388",
    archivePrefix = "arXiv",
    primaryClass = "math.QA",
    doi = "10.1007/s00220-019-03300-2",
    journal = "Commun. Math. Phys.",
    volume = "365",
    number = "3",
    pages = "943--970",
    year = "2019"
}

@article{Gaiotto:2017zba,
    author = "Gaiotto, Davide and Johnson-Freyd, Theo",
    title = "{Symmetry Protected Topological phases and Generalized Cohomology}",
    eprint = "1712.07950",
    archivePrefix = "arXiv",
    primaryClass = "hep-th",
    doi = "10.1007/JHEP05(2019)007",
    journal = "JHEP",
    volume = "05",
    pages = "007",
    year = "2019"
}
\end{document}